\def\draft{1}

\documentclass[11pt]{article}

\usepackage[utf8]{inputenc}
\usepackage{amsthm,mathtools}
\usepackage[margin=1in]{geometry}
\usepackage{amsfonts}
\usepackage{color,soul}
\usepackage{amssymb}
\usepackage{footmisc}
\usepackage[T1]{fontenc}
\usepackage{lmodern}
\usepackage{hyperref}
\usepackage{algpseudocode}
\usepackage{algorithm}
\usepackage{float}
\usepackage{tikz}
\usepackage[framemethod=tikz]{mdframed}
\newtheorem{theorem}{Theorem}[section]
\newtheorem{lemma}[theorem]{Lemma}
\newtheorem{proposition}[theorem]{Proposition}

\newtheorem{corollary}[theorem]{Corollary}

\theoremstyle{definition}
\newtheorem{definition}[theorem]{Definition}
\newcommand{\eps}{\epsilon}
 
\newcommand{\dmax}{d_{\mathrm{max}}}
\newcommand{\ds}{\text{\textcircled{s}}}

\newcommand{\R}{\mathbb{R}}
\newcommand{\Q}{\mathbb{Q}}
\newcommand{\Space}{\mathrm{Space}}
\newcommand{\class}[1]{\mathbf{#1}}

\newcommand{\DSPACE}{\class{DSPACE}}

\newcommand{\eqdef}{\mathbin{\stackrel{\rm def}{=}}}
\newcommand{\capprox}{\mathbin{\stackrel{\rm \circ}{\approx}}}
\newcommand{\tO}{\tilde{O}}
\newcommand{\poly}{\mathrm{poly}}
\newcommand{\polylog}{\mathrm{polylog}}

\def\textprob#1{\textmd{\textsc{#1}}}
\newcommand{\USTConn}{\textprob{Undirected S-T Connectivity}}

\newcommand{\Real}{\mathrm{Re}}

\ifnum\draft=1
\newcommand{\authnote}[2]{{ $\ll$\textsf{\color{green} \footnotesize #1 notes: #2}$\gg$}}
\else
\newcommand{\authnote}[2]{}
\fi

\newcounter{myalgctr}
\numberwithin{myalgctr}{section}
\newtheoremstyle{named}{}{}{\itshape}{}{\bfseries}{.}{.5em}{\thmnote{#3}}
\theoremstyle{named}
\newtheorem{namedtheorem}{Theorem}[section]
\newcommand{\super}[2]{#1^{(#2)}}

\title{High-precision Estimation of Random Walks in Small Space}

\date{\today}
\author{AmirMahdi Ahmadinejad\\
Stanford University\\
\texttt{ahmadi@stanford.edu}
\and Jonathan Kelner \\
Massachusetts Institute of Technology\\
\texttt{kelner@mit.edu}
\and Jack Murtagh\thanks{Research supported by NSF grant CCF-1763299.} \\
Harvard University\\ 
\texttt{jmurtagh@g.harvard.edu}
\and John Peebles\thanks{Research supported by Swiss National Science Foundation Grant \#200021\_182527.}\\
Yale University\\
\texttt{john.peebles@yale.edu}
\and Aaron Sidford\thanks{Research supported by NSF CAREER Award CCF-1844855.} \\
Stanford University\\
\texttt{sidford@stanford.edu}
\and Salil Vadhan\thanks{Research supported by NSF grant CCF-1763299 and a Simons Investigator Award.}\\
Harvard University\\
\texttt{salil\_vadhan@harvard.edu}
}
\begin{document}
\global\long\def\R{\mathbb{R}}%
 
\global\long\def\Z{\mathbb{Z}}%
 
\global\long\def\C{\mathbb{C}}%
\global\long\def\Q{\mathbb{Q}}%

\global\long\def\ellOne{\ell_{1}}%
 
\global\long\def\ellTwo{\ell_{2}}%
 
\global\long\def\ellInf{\ell_{\infty}}%

\global\long\def\boldVar#1{\mathbf{#1}}%
\global\long\def\mvar#1{\boldVar{#1}}%
\global\long\def\vvar#1{\vec{#1}}%


\global\long\def\defeq{\stackrel{\mathrm{{\scriptscriptstyle def}}}{=}}%
\global\long\def\E{\mathbb{E}}%
\global\long\def\otilde{\tilde{O}}%


\global\long\def\gradient{\bigtriangledown}%
 
\global\long\def\grad{\gradient}%
 
\global\long\def\hessian{\gradient^{2}}%
 
\global\long\def\hess{\hessian}%
 
\global\long\def\jacobian{\mvar J}%

 
\global\long\def\setVec#1{\onesVec_{#1}}%
 
\global\long\def\indicVec#1{\onesVec_{#1}}%

\global\long\def\innerProduct#1#2{\big\langle#1 , #2 \big\rangle}%
 
\global\long\def\norm#1{\|#1\|}%
\global\long\def\normFull#1{\left\Vert #1\right\Vert }%

\global\long\def\opt{\mathrm{opt}}%

\global\long\def\va{\vvar a}%
 
\global\long\def\vb{\vvar b}%
 
\global\long\def\vc{\vvar c}%
 
\global\long\def\vd{\vvar d}%
 
\global\long\def\ve{\vvar e}%
 
\global\long\def\vf{\vvar f}%
 
\global\long\def\vg{\vvar g}%
 
\global\long\def\vh{\vvar h}%
 
\global\long\def\vl{\vvar l}%
 
\global\long\def\vm{\vvar m}%
 
\global\long\def\vn{\vvar n}%
 
\global\long\def\vo{\vvar o}%
 
\global\long\def\vp{\vvar p}%
 
\global\long\def\vq{\vvar q}%
 
\global\long\def\vr{\vvar r}%
 
\global\long\def\vs{\vvar s}%
 
\global\long\def\vu{\vvar u}%
 
\global\long\def\vv{\vvar v}%
 
\global\long\def\vw{\vvar w}%
 
\global\long\def\vx{\vvar x}%
 
\global\long\def\vy{\vvar y}%
 
\global\long\def\vz{\vvar z}%

\global\long\def\vpi{\vvar{\pi}}%
\global\long\def\vxi{\vvar{\xi}}%
\global\long\def\vchi{\vvar{\chi}}%
 
\global\long\def\valpha{\vvar{\alpha}}%
 
\global\long\def\veta{\vvar{\eta}}%
 
\global\long\def\vlambda{\vvar{\lambda}}%
 
\global\long\def\vmu{\vvar{\mu}}%
\global\long\def\vdelta{\vvar{\Delta}}%
 
\global\long\def\vsigma{\vvar{\sigma}}%
 
\global\long\def\vzero{\vvar 0}%
 
\global\long\def\vones{\vvar 1}%
\global\long\def\onesVec{\vvar 1}%

\global\long\def\xopt{\vvar x^{*}}%
 
\global\long\def\varVec{\vvar x}%
 
\global\long\def\varVecA{\vvar x}%
 
\global\long\def\varVecB{\vvar y}%
 
\global\long\def\ib{\mvar I \mvar B}%
 
\global\long\def\varMat{\mvar A}%
 
\global\long\def\varMatA{\mvar A}%
 
\global\long\def\varMatB{\mvar B}%
 
\global\long\def\varSubedges{H}%

\global\long\def\ma{\mvar A}%
 
\global\long\def\mb{\mvar B}%
 
\global\long\def\mc{\mvar C}%
 
\global\long\def\md{\mvar D}%
\global\long\def\mE{\mvar E}%
 
\global\long\def\mf{\mvar F}%
 
\global\long\def\mg{\mvar G}%
 
\global\long\def\mh{\mvar H}%
\global\long\def\mI{\mvar I}%
\global\long\def\mJ{\mvar J}%
\global\long\def\mK{\mvar K}%
 
\global\long\def\mL{\mvar L}%
\global\long\def\mm{\mvar M}%
 
\global\long\def\mn{\mvar N}%
\global\long\def\mq{\mvar Q}%
 
\global\long\def\mr{\mvar R}%
 
\global\long\def\ms{\mvar S}%
 
\global\long\def\mt{\mvar T}%
 
\global\long\def\mU{\mvar U}%
 
\global\long\def\mv{\mvar V}%
 
\global\long\def\mw{\mvar W}%
 
\global\long\def\mx{\mvar X}%
 
\global\long\def\my{\mvar Y}%
\global\long\def\mz{\mvar Z}%
 
\global\long\def\mproj{\mvar P}%
 
\global\long\def\mSigma{\mvar{\Sigma}}%
 
\global\long\def\mDelta{\mvar{\Delta}}%
\global\long\def\mLambda{\mvar{\Lambda}}%
 
\global\long\def\mha{\hat{\mvar A}}%
 
\global\long\def\mzero{\mvar 0}%
\global\long\def\mlap{\mvar{\mathcal{L}}}%
\global\long\def\mpi{\mvar{\mathcal{\Pi}}}%

\global\long\def\mdiag{\mvar{diag}}%
\global\long\def\diag{\mathrm{diag}}%

 
\global\long\def\oracle{\mathcal{O}}%
 
\global\long\def\moracle{\mvar O}%
 
\global\long\def\oracleOf#1{\oracle\left(#1\right)}%
 
\global\long\def\nSamples{s}%
 
\global\long\def\simplex{\Delta}%

\global\long\def\abs#1{\left|#1\right|}%
\global\long\def\tr{\mathrm{tr}}%

\global\long\def\timeNearlyOp{\tilde{\mathcal{O}}}%
 
\global\long\def\timeNearlyLinear{\timeNearlyOp}%

\global\long\def\im{\mathrm{im}}%

\global\long\def\ceil#1{\left\lceil #1 \right\rceil }%

\global\long\def\runtime{\mathcal{T}}%
 
\global\long\def\timeOf#1{\runtime\left(#1\right)}%

\global\long\def\domain{\mathcal{D}}%

\global\long\def\argmin{\mathrm{argmin}}%
\global\long\def\argmax{\mathrm{argmax}}%
\global\long\def\nnz{\mathrm{nnz}}%
\global\long\def\vol{\mathrm{vol}}%
\global\long\def\supp{\mathrm{supp}}%
\global\long\def\dist{\mathcal{D}}%
\global\long\def\energy{\xi}%
\global\long\def\indicDiff{\vec{\delta}}%
\global\long\def\congest{\mathrm{cong}}%
\global\long\def\poly{\mathrm{poly}}%
\global\long\def\congest{\mathrm{cong}}%
\global\long\def\conductance{\mathrm{\phi}}%
\global\long\def\sparsity{\mathrm{\sigma}}%
\global\long\def\prodop{\mathrm{\otimes}}%
\global\long\def\kron{\mathrm{\otimes}}%
\global\long\def\imbal{\mathrm{imbalance}}%
\global\long\def\boundary{\partial}%
\global\long\def\opt{\mathrm{OPT}}%
\newcommand{\bone}{\mathbf{1}}%
\newcommand{\bzero}{\mathbf{0}}%

\global\long\def\schur{\mathrm{Sc}}%

\begin{titlepage}

\maketitle
\thispagestyle{empty}

\begin{abstract}
In this paper, we provide a deterministic $\tilde{O}(\log N)$-space algorithm for estimating random walk probabilities on 
undirected graphs, and more generally Eulerian directed graphs, to within inverse polynomial additive error ($\eps=1/\poly(N)$) where $N$ is the length of the input. Previously, this problem was known to be solvable by a randomized algorithm using space $O(\log N)$ (following Aleliunas et al., FOCS `79) and by a deterministic algorithm using space $O(\log^{3/2} N)$ (Saks and Zhou, FOCS `95 and JCSS `99), both of which held for arbitrary directed graphs but had not been improved even for undirected graphs.  We also give improvements on the space complexity of both of these previous algorithms for non-Eulerian directed graphs when the error is negligible ($\eps=1/N^{\omega(1)}$), generalizing what Hoza and Zuckerman (FOCS `18) recently showed for the special case of distinguishing whether a random walk probability is $0$ or greater than $\eps$.

We achieve these results by giving new reductions between powering Eulerian random-walk matrices and inverting Eulerian Laplacian matrices, providing a new notion of spectral approximation for Eulerian graphs that is preserved under powering, and giving the first deterministic $\tilde{O}(\log N)$-space algorithm for inverting Eulerian Laplacian matrices. The latter algorithm builds on the work of Murtagh et al. (FOCS `17) that gave a deterministic $\tilde{O}(\log N)$-space algorithm for inverting undirected Laplacian matrices, and the work of Cohen et al. (FOCS `19) that gave a randomized $\tilde{O}(N)$-time algorithm for inverting Eulerian Laplacian matrices.  A running theme throughout these contributions is an analysis of ``cycle-lifted graphs,'' where we take a graph and ``lift'' it to a new graph whose adjacency matrix is the tensor product of the original adjacency matrix and a directed cycle (or variants of one). 
\end{abstract}

\vfill


\end{titlepage}

\section{Introduction}
\label{sec:introduction}

In this paper, we give  the first deterministic, nearly logarithmic-space algorithm for accurately estimating random walk probabilities on undirected graphs.  Our algorithm extends to Eulerian digraphs, which are directed graphs where the indegree of a vertex $v$ is equal to its outdegree for every vertex $v$.  (Note that a random walk on an undirected graph is equivalent to a random walk on the associated Eulerian digraph obtained by replacing each undirected edge $\{u,v\}$ with two directed edges $(u,v)$ and $(v,u)$.) 
In more detail, our main result is as follows:
\begin{theorem}[informally stated (see also Theorem \ref{thm:state_probs})] \label{thm:EulerianPowering-intro}\label{thm:main-intro}
There is a deterministic, $\tO(\log (k\cdot N))$-space algorithm that given an
Eulerian digraph $G$ (or an undirected graph $G$), two vertices $s,t$, and a positive integer $k$, outputs the probability that a $k$-step random walk in $G$ started at $s$ ends at $t$, to within
additive error of $\epsilon$, where $N$
is the length of the input and $\epsilon=1/\poly(N)$ is any desired polynomial accuracy parameter.
\end{theorem}

Estimating random walk probabilities to inverse polynomial accuracy, even in general digraphs, can easily be done by randomized algorithms running in space $O(\log N)$. Since that much space is sufficient to simulate random walks~\cite{AleliunasKaLiLoRa79}. In fact, estimating random walk probabilities in general digraphs is complete for randomized logspace.\footnote{Formally, given $G,s,t,k$, a threshold $\tau$, and a unary accuracy parameter $1^a$, the problem of deciding whether the $k$-step random walk probability from $s$ to $t$ is larger than $\tau+1/a$ or smaller than $\tau$ is complete for the class BPL of promise problems having randomized logspace algorithms with two-sided error.  By binary search over the threshold $\tau$, this problem is log-space equivalent to estimating the same probability to within error $1/a$.}
The best known deterministic algorithm prior to our work was that of Saks and Zhou~\cite{SaksZh99}, which 
runs in 
space $O(\log^{3/2} N)$, and works for general digraphs. 
(See the excellent survey of Saks~\cite{Saks96} for more discussion of the close connection between randomized space-bounded computation and random walks, as well as the state-of-art in derandomizing such computations up to the mid-1990's.)

For undirected graphs, Murtagh et al.~\cite{MRSV19} recently gave a deterministic $\tO(\log (k\cdot N))$-space algorithm that computes a much weaker form of approximation for random walks: given any subset $S$ of vertices, the algorithm estimates, to within a multiplicative factor of $(1\pm 1/\polylog(N))$, the {\em conductance} of the set $S$, namely the probability that a $k$-step random walk started at a random vertex of $S$ (with probability proportional to vertex degrees) ends outside of $S$.  Our result is stronger because in undirected graphs, all nonzero conductances can be shown to be of magnitude at least $1/\poly(N)$ and can be expressed as a sum of at most $N^2/4$ random walk probabilities.  Consequently, with the same space bound our
algorithm can estimate the conductance of any set $S$ to within a
multiplicative factor of $(1\pm 1/\poly(N))$.

Like \cite{MRSV19}, our work is part of a larger project, initiated in \cite{MRSV17}, that seeks to
make progress on the derandomization of space-bounded computation by importing ideas
from the literature on time-efficient randomized algorithms for solving graph Laplacian systems \cite{ST04,KMP10,KMP11,KOSZ13,LS13,PS13,CohenKMPPRX14,kyng2016sparsified,KyngS16,CKPPSV16,CKPPRSV17,CKKPPRS18,AhmadinejadJSS19} .
While we consider Theorem~\ref{thm:main-intro} to be a natural derandomization result in its own right, it
and our analogous result for solving Eulerian Laplacian systems (Theorem~\ref{thm:EulerianLaplacians-intro} below)
can also be viewed
as a step toward handling general directed graphs and thereby having an almost-complete derandomization of randomized logspace.
Indeed, in recent years, nearly linear-time randomized algorithms for estimating properties of random walks on general directed graphs (with polynomial mixing time, which also yield complete problems for randomized logspace~\cite{ReingoldTrVa06})  were obtained by reduction to
solving Eulerian Laplacian systems~\cite{CKPPSV16,CKPPRSV17,CKKPPRS18}.  A deterministic and sufficiently space-efficient analogue of such
a reduction, combined with our results, would put randomized logspace in deterministic space $\tO(\log N)$.

To achieve our main result and prove Theorem~\ref{thm:main-intro}, we provide several results that may be of interest even  outside of the space-bounded derandomization context, such as a new notion of
spectral approximation and new reductions between estimating random walk probabilities and inverting Laplacian systems.
Below we provide more details on how our work relates to both the space-bounded derandomization and the Laplacian solving literature, and describe some of our other contributions.

\subsection{Derandomization of Space-Bounded Computation}

It is widely conjectured that every problem solvable
in randomized logarithmic space can also be solved in deterministic logarithmic space (i.e. RL = L, BPL=L for the one-sided and two-sided error versions, respectively)~\cite{Saks96}. Though this is known to follow from mild assumptions in complexity theory (e.g. that there is a Boolean function in $\DSPACE(n)$ that requires branching programs of size $2^{\Omega(n)}$ ~\cite{KlivansMe02}), the best known unconditional derandomization is the aforementioned, quarter-century-old result of Saks and Zhou~\cite{SaksZh99}, which places randomized logspace in deterministic space $O(\log^{3/2} N)$.

Most of the effort on derandomizing logspace computations over the past three decades has gone towards the design of {\em pseudorandom generators} that fool {\em ordered branching programs}.  An ordered branching program of {\em width} $w$ and {\em length} $k$ is given by a directed graph on vertex set $[k]\times [w]$. which we view as consisting of $k$ layers of $w$ vertices. All of the edges from the $i$th layer go to the $(i+1)$'st layer (so there are no edges entering the first layer or exiting the last layer).  We call the first vertex of the first layer (i.e. vertex $(1,1)$) the
{\em start vertex},
and the first vertex of the last layer (i.e. vertex $(k,1)$) the \emph{accept vertex} $t$.  Typically, every vertex in layers 1 to $k-1$ has outdegree 2, with the two edges labelled by 0 and 1.  Intuitively, the vertices in the $i$th layer correspond to possible states of a space-bounded algorithm before it makes its $i$th coin toss, and the two edges lead to its two possible states after that coin toss.  The acceptance probability of an ordered branching program is exactly the probability that a random walk from the start vertex $s$ of length $k-1$ ends at accept vertex $t$.  Generating such a truly random walk takes $k-1$ random bits, so the goal of a pseudorandom generator for ordered branching programs is to generate a walk of length $k-1$ using a much shorter random seed such that the acceptance probability is preserved up to an additive $\epsilon$. 
Given such a pseudorandom generator, we can obtain a deterministic algorithm for estimating the acceptance probability by enumerating all seeds of the pseudorandom generator.

A general $O(\log N)$-space computation can have $w=2^{O(\log N)}=\poly(N)$ states and toss $k=\poly(N)$ coins.
The best known pseudorandom generator for such ordered branching programs is the classic generator of Nisan~\cite{Nisan91}, which has a seed length of $O(\log^2 N)$ (for any error $\epsilon\geq 1/\poly(N)$) and thus does not directly yield a derandomization of space complexity better than $O(\log^2 N)$ (due to enumerating the seeds), which can be achieved more easily by recursively multiplying the transition matrices between layers.  (Multiplying $k$ boolean $w\times w$ matrices to within a final accuracy of $\epsilon$ can be done recursively in space $O((\log k)\cdot (\log w+\log\log(k/\epsilon)))$.) Nisan's generator is also a crucial tool in the algorithm of Saks and Zhou~\cite{SaksZh99}. 

Due to the long lack of progress in improving Nisan's generator, effort has turned to restricted classes of branching programs, such as those of constant width ($w=O(1)$), with there being significant progress in recent years for the case of width $w=3$.~\cite{SimaZa11,GopalanMeReTrVa12,MekaReTa18}.
Another restriction that has been studied is that of {\em regular} branching programs, where every vertex in layers $2,\ldots,k$ in the branching program has indegree $2$. 
For this case, Braverman, Rao, Raz, and Yehudayoff~\cite{BravermanRaRaYe10} give a pseudorandom generator with seed length $\tO(\log N)$ when $w\leq \polylog(N)$ and $\epsilon\geq 1/\polylog(N)$, which again does not yield a deterministic algorithm that improves upon recursive matrix multiplications.

In contrast, our algorithm for Eulerian graphs can be used to estimate the acceptance probability of a regular branching program in space $\tO(\log N)$ even when $w=\poly(N)$ and $\epsilon=1/\poly(N)$.  Indeed, by adding edges from the $k$th layer back to the first layer, a regular branching program can be made into an Eulerian graph, without changing the probability that a random walk of length $k-1$ from $s$ ends at $t$.
In addition, our techniques also yield an improved pseudorandom
generator for {\em permutation branching programs} (regular branching programs where the labelling is constrained so that for each $b\in \{0,1\}$, the edges labelled $b$ form perfect matchings between the vertices in consecutive layers).  Specifically, \cite{HPV21} use our results to derive a pseudorandom generator
with seed length $\tO(\log N)$ for permutation branching programs (with a single accept vertex in layer $k$) of width $w=\poly(N)$, albeit with error only $\epsilon=1/\polylog(N)$.  Even for the special case
of permutation branching programs, seed length $\tO(\log N)$ was previously only achieved for width $w=\polylog(N)$~\cite{KouckyNiPu11,De11,Steinke12}.

It is also worth comparing our result with Reingold's Theorem, which gives a deterministic logspace algorithm for deciding $s$-$t$ connectivity in undirected graphs.   Reingold, Trevisan, and Vadhan~\cite{ReingoldTrVa06} interpreted and generalized Reingold's methods to obtain ``pseudoconverging walk generators'' for  regular digraphs where each edge label forms a permutation of the vertices, as in the permutation branching programs described above.  These generators provide a way to use a seed of $O(\log N)$ random bits to generate walks of length $\poly(N)$ that converge to the uniform distribution on the connected component of the start vertex (just like a truly random walk would).  Such generators turn out to suffice for deciding $s$-$t$ connectivity on arbitrary Eulerian digraphs.  However, these generators are not guaranteed to closely approximate the behavior of a truly random walk at shorter walk lengths.  Indeed, even the length of the walks needed for mixing is polynomially larger than with a truly random walk.
Nevertheless, one of the techniques we use, the {\em derandomized square}, originated from an effort to simplify Reingold's algorithm and these pseudoconverging walk generators~\cite{RozenmanVa05}.

Our work builds on recent papers of Murtagh et al.~\cite{MRSV17,MRSV19}, which gave deterministic, nearly logarithmic-space algorithms for estimating certain quantities associated with random walks on undirected graphs.  Specifically, the first of these papers~\cite{MRSV17} gave a ``Laplacian solver'' (defined below) that implied
accurate (to within $1/\poly(N)$ error) estimates of {\em escape probabilities} (the probability that a random walk from $s$ visits $t$ before visiting another vertex $v$); these again refer to the long-term behavior of random walks, rather than the behavior at a given time below mixing.   The second paper~\cite{MRSV19} dealt with random walks of a fixed length $k$, but as discussed earlier, only gave a weak approximation to the conductance of subsets of vertices.

\subsection{Inverting Laplacian Systems}

We prove Theorem~\ref{thm:main-intro} by a novel reduction from estimating $k$-step random walk probabilities to solving linear systems given by graph Laplacians, and giving a small-space algorithm for the latter in the case of
Eulerian graphs.

Let $G$ be a digraph on $n$ vertices, and let $\mw$ be the $n\times n$ transition matrix for the random walk on $G$.  Then we will call $\mL=\mI-\mw$ the {\em random-walk Laplacian of $G$}.\footnote{The standard Laplacian of $G$, which we simply refer to as the {\em Laplacian of $G$} is  $\md-\ma$, where $\md$ is the diagonal matrix of outdegrees and $\ma$ is the adjacency matrix (where we define $(\ma)_{ij}$ to be the number of edges from $j$ to $i$ in $G$).  Notice that $\mI-\mw=(\md-\ma)\md^{-1}$.  For undirected graphs, it is common to use a symmetric normalization of the Laplacian given by $\md^{-1/2} (\md-\ma)\md^{-1/2} = \md^{-1/2}(\mI-\mw) \md^{1/2}$, but
the $\mI-\mw$ formulation will be more convenient for us.} 
Solving Laplacian systems refers to the problem of given a vector $b\in \R^n$, finding a vector $x\in \R^n$ such that $\mL x=b$ (if one exists).  Since the matrix $\mL$ is not invertible (a stationary distribution of the random walk on $G$ is in the kernel), a Laplacian system can be solved by instead computing its {\em pseudoinverse} $\mL^+$, which acts as an inverse on $\mathrm{Image}(\mL)$, and is zero on the orthogonal complement of $\mathrm{Image}(\mL)$.

We show that we can compute the pseudoinverse of an Eulerian Laplacian in nearly logarithmic space.

\begin{theorem}[informally stated (see also Theorem \ref{thm:space_main})] \label{thm:EulerianLaplacians-intro}
There is a deterministic, $\tO(\log N)$-space algorithm that given an
Eulerian digraph $G$ with random-walk transition matrix $\mw$, outputs a matrix $\widetilde{\mL^+}$ whose entries differ from $\mL^+$ by at most an additive $\epsilon$, where $N$
is the length of the input and $\epsilon=1/\poly(N)$ is any desired polynomial accuracy parameter.
\end{theorem}
Previously, Cohen et al.~\cite{CKPPRSV17,CKKPPRS18} showed how to solve Eulerian Laplacian systems by randomized, nearly linear-time algorithms, and Murtagh et al.~\cite{MRSV17} showed how to solve undirected Laplacian systems by deterministic, nearly logarithmic-space algorithms.  Our proof of Theorem~\ref{thm:EulerianLaplacians-intro} draws on all of these works. 

As explained below, the extension from undirected graphs (handled by \cite{MRSV17}) to Eulerian graphs (as in Theorem~\ref{thm:EulerianLaplacians-intro}) is 
crucial for obtaining our high-precision estimation of random walks (Theorem~\ref{thm:main-intro}) even for the case of undirected graphs.
In addition, as discussed earlier, this extension 
can also be viewed
as a step toward handling general directed graphs and thereby having an almost-complete derandomization of randomized logspace.
Recall that nearly linear-time randomized algorithms for estimating properties of random walks on general directed graphs were obtained by reduction to
solving Eulerian Laplacian systems~\cite{CKPPSV16,CKPPRSV17,CKKPPRS18}.  A deterministic and sufficiently space-efficient analogue of such
a reduction, combined with Theorem~\ref{thm:EulerianLaplacians-intro}, would put randomized logspace in deterministic space $\tO(\log N)$, i.e.
$\mathrm{BPL}\subseteq \widetilde{\mathrm{L}}$.

We will describe the ideas in the proof of Theorem~\ref{thm:EulerianLaplacians-intro} below in Section~\ref{sec:techniques-intro}.  Here we
describe our reduction from powering (Theorem~\ref{thm:EulerianPowering-intro}) to computing the pseudoinverse of a Laplacian (Theorem~\ref{thm:EulerianLaplacians-intro}).

Let $G$ be an $n$-vertex digraph with random-walk transition matrix $\mw$.  
Let $\mproj_k$ be the adjacency matrix of a $k$-vertex directed path. Note that $\mproj_k$ is not stochastic, but rather {\em substochastic} (nonnegative with column sums at most 1), since there are no edges leaving the last vertex of the path (i.e. random walks ``die off'' when leaving that vertex).  Then the $kn\times kn$ matrix $\mw'=\mproj_k\otimes \mw$,  i.e. the Kronecker product of $\mproj_k$ and $\mw$, is a $k\times k$ block matrix consisting of $n\times n$ blocks that equal $\mw$ just below the diagonal and are zero elsewhere.  For example,
when $k=4$, we have:
\[
\mw'=\mproj_k\otimes \mw
= \begin{bmatrix}
0 & 0 & 0 & 0 \\
\mw & 0 & 0 & 0 \\
0 & \mw & 0 & 0 \\
0 & 0 & \mw & 0
\end{bmatrix}.
\]
$\mw'$ is
also a substochastic matrix describing random walks on a graph with $k$ layers of $n$ vertices each, where there is a bipartite version of $G$ going from the $i$th layer to the $(i+1)$'st layer for each $i=1,\ldots,k-1$, and again there are no edges leaving the $k$th layer. 
We call this the {\em path-lift of $G$ of length $k$}, or the {\em path-lifted graph} when $G$ and $k$ are clear from context.
(This construction is inspired by the ordered branching programs that arise in space-bounded computation, as described above.) 

The ``Laplacian'' of this layered graph, $\mL=\mI_{kn}-\mw'$, is invertible, and noting that $(\mw')^k=0$, we can calculate its inverse as:
\begin{equation} \label{eqn:path-intro}
\mL^{-1} = \mI_{nk} + \mw' + (\mw')^2+\cdots+ (\mw')^{k-1}
= \mI_k\otimes \mI_n + \mproj_k\otimes \mw+ \mproj_k^2 \otimes \mw^2 + \cdots 
\mproj_k^{k-1}\otimes \mw^{k-1}.
\end{equation}
Thinking of $\mL^{-1}$ as a block matrix, the term $\mproj_k^j \otimes \mw^j$ places a copy of $\mw^j$ in each of the blocks that are in the $j$th diagonal below the main diagonal.  (So on the main diagonal are blocks of $\mI_n$, just below the main diagonal are blocks of $\mw$, below that $\mw^2$, and so on.)  For example, for $k=4$, we can write $\mL^{-1}$ in block form as
\[
\mL^{-1} =
\begin{bmatrix}
\mI & 0 & 0 & 0 \\
\mw & \mI & 0 & 0 \\
\mw^2 & \mw & \mI & 0 \\
\mw^3 & \mw^2 & \mw & \mI
\end{bmatrix}.
\]
Thus from an accurate estimate of $\mL^{-1}$, we can read off accurate estimates of the powers of $\mw$.  For example, entry $((\ell,t),(1,s))$ of $\mL^{-1}$ is exactly the probability that a length $\ell-1$ random walk in $G$ started at $s$ ends at $t$.

However, we can only apply Theorem~\ref{thm:EulerianLaplacians-intro} directly if $\mL$ is the Laplacian of an Eulerian graph, and the above $\mw'$ is not even stochastic.  We can fix this by (a) starting with an Eulerian graph $G$ and (b)
considering a {\em cycle-lifted graph} instead of a path-lifted graph, i.e. considering transition matrix $\mc_k\otimes \mw$.  Additionally, it is convenient to collapse all of the vertices in layer $k$ to a single vertex $v$.
Then it turns out from the Laplacian pseudoinverse $\mL^+$, it is possible to read off {\em escape probabilities} --- the probability that a random walk from one vertex, say $(1,s)$, visits another vertex, say $(\ell,t)$, before visiting a third vertex, say $v$.   The condition ``before visiting $v$'' allows us to not worry about walks that traverse all the way around the cycle, and thus we get exactly the probability that a length $\ell-1$ random walk in $G$ started at $s$ ends at $t$.\footnote{Alternatively (and essentially equivalently), we could note that if $D-A$ is the Laplacian of $G$, then 
$I_k\otimes D-P_k\otimes A = I_k\otimes D\cdot (I_{nk}- P_k\otimes W)$ is a ``row-column diagonally dominant matrix'' and apply the reduction from inverting such matrices to pseudo-inverting Eulerian Laplacian systems~\cite{CKPPSV16}.}

Note that even if $G$ is undirected, this reduction requires inverting a Laplacian of a directed layered graph.  Thus, our extension of the small-space undirected Laplacian solver of \cite{MRSV17} to Eulerian graphs (Theorem~\ref{thm:EulerianLaplacians-intro}) seems essential for obtaining high-precision estimates of powers even for undirected graphs (Theorem~\ref{thm:EulerianPowering-intro}). 

The reduction above from computing powers to inverting also allows us to obtain new algorithms for general digraphs and Markov chains:

\begin{theorem}[informally stated (see also Theorem \ref{thm:arbitrary_powers})] \label{thm:GeneralPowers-intro}
Given a Markov chain $G$ specified by a stochastic matrix $\mw$, two states $s,t$, and a positive integer $k$, we can compute the probability that a $k$-step random walk in $G$ started at $s$ ends at $t$, to within
an additive error of $\eps$:
\begin{enumerate}
    \item In randomized space $O((\log Nk)\cdot \log(\log_{Nk}(1/\eps)))$, or
    \item In deterministic space $O(\log^{3/2}(Nk) + (\log(Nk))\cdot \log(\log_{Nk}(1/\eps)))$.
\end{enumerate}
where $N$
is the length of the input.
\end{theorem}
This theorem generalizes one of the results from recent work of Hoza and Zuckerman~\cite{HozaZu18} that gave the same bounds for the 1-sided version of the problem, namely 
distinguishing probability 0 from probability greater than $\eps$.  For the two-sided version of the problem that we consider, a randomized algorithm using space $O(\log(Nk/\eps))$ follows from performing $\poly(1/\eps)$ random walks and counting how many end at $t$.  For deterministic algorithms, the best previous algorithm is from Saks and Zhou~\cite{SaksZh99}, which uses space $O(\log(Nk/\eps)\cdot \log^{1/2} k)$.  Note that Theorem~\ref{thm:GeneralPowers-intro} has a doubly-logarithmic dependence on $\eps$ rather than a singly-logarithmic one.  In particular, Saks and Zhou~\cite{SaksZh99} only achieves space $O(\log^{3/2} Nk)$ for $\eps=1/\poly(Nk)$ whereas we achieve it for a much smaller $\eps=1/\exp(\exp(\sqrt{\log Nk}))$.

The proof of Theorem~\ref{thm:GeneralPowers-intro} begins with the observation that we can approximate $\mL^{-1}=(\mI_{nk}-\mproj_k\otimes \mw)^{-1}$ to within accuracy $1/\poly(N,k)$ in randomized space $O(\log Nk)$ or deterministic space $O(\log^{3/2} Nk)$.  Indeed, by Equation~(\ref{eqn:path-intro}), it suffices to estimate $\mI,\mw,\mw^2,\ldots,\mw^{k-1}$ up to accuracy $\pm 1/\poly(N,k)$.  

We then use the fact that matrix inversion has a very efficient error reduction procedure, equivalent to what is commonly known as ``preconditioned Richardson iterations''.  Let 
$\widetilde{\mL^{-1}}$ denote our approximation to $\mL^{-1}$ with error $1/\poly(N,k)$.  For an appropriate choice of the polynomial error bound,  it follows that the ``error matrix'' $\mE=\mI_{nk}-\widetilde{\mL^{-1}}\mL$ has norm at most $1/Nk$ (in, say, spectral norm).  Then we can obtain a
more accurate estimate of $\mL^{-1}$ by using the identity:
$$\mL^{-1} = (\mI_{nk}-\mE)^{-1} \widetilde{\mL^{-1}} 
=  (\mI_{nk}+\mE+\mE^2+\mE^3+\cdots)\widetilde{\mL^{-1}}.$$
Since $\mE$ has norm at most $1/Nk$, the series converges very quickly, and can be truncated at $O(\log_{Nk}(1/\eps))$ terms to achieve an approximation to within $\pm \eps$.  As noted earlier, from such an accurate estimation of $\mL^{-1}=(\mI_{nk}-\mproj_k\otimes \mw)^{-1}$, we can accurately estimate random walks of length $k-1$.

This same error reduction procedure is also used in
our proof of Theorem~\ref{thm:EulerianLaplacians-intro} (and also throughout the literature on time-efficient Laplacian solvers), and thus is also key to the high precision estimates we obtain in Theorem~\ref{thm:EulerianPowering-intro}.  Although we fixed error $1/\poly(N)$ in the statement of the theorem, we can also achieve smaller error $\eps$ at a price of $O((\log N)\cdot \log(\log_N(1/\eps)))$ in the space complexity.

Interestingly, early work on randomized space-bounded computation~\cite{Gill77,Simon81,BorodinCoPi83,Jung81} reduced the problem of {\em exactly} computing {\em hitting probabilities} of Markov chains (the probability that an {\em infinite} random walk from $s$ ever hits $t$) to
computing $(\mI-\mw)^{-1}$ for a substochastic matrix $\mw$, and used this to show that {\em unbounded-error} and {\em non-halting} randomized logspace is contained in deterministic space $O(\log^2 N)$.   As far as we know, ours is the first application of inverting Laplacian systems to estimating finite-time random-walk probabilities to within polynomially small error, and consequently it is also the first application of inverting Laplacian systems to the commonly accepted formulation
of randomized logspace 
(i.e. bounded error and halting).

\subsection{Complex spectral approximation, cycle-lifted graphs, and powering} \label{sec:techniques-intro}

We now describe the techniques underlying our space-efficient Eulerian Laplacian inverter (Theorem~\ref{thm:EulerianLaplacians-intro}).
Let $\mw$ be the transition matrix for the random walk on an $n$-vertex Eulerian graph $G$, for which we want to estimate $(\mI-\mw)^{-1}$. Because of (a generalization of) the error-reduction procedure described above, it suffices to compute a rough approximation to $(\mI-\mw)^{-1}$.  For symmetric matrices (as arising from undirected graphs), a sufficient notion of approximation is {\em spectral approximation} as introduced by Spielman and Teng~\cite{ST11}.  Applying the Spielman--Teng notion to symmetric Laplacians $\mI-\mw$ and $\mI-\widetilde{\mw}$, we say that $\widetilde{\mw}$ is an {\em $\epsilon$-approximation} of $\mw$ (written $\widetilde{\mw}\approx_\epsilon\mw$)
if 
\begin{equation} \label{eqn:ST-intro}
\forall x\in \R^n\qquad
\left|x^\top(\mw-\widetilde{\mw})x\right|\leq \epsilon\cdot x^{\top}(\mI-\mw)x=\epsilon\cdot\left(\|x\|^2-x^{\top}\mw x\right).
\end{equation}
Cohen et al.~\cite{CKPPRSV17} introduced a generalization of spectral approximation for asymmetric matrices and directed graphs:

\begin{equation} 
\forall x,y\in \R^n\qquad
\left|x^\top(\mw-\widetilde{\mw})y\right|\leq 
\frac{\epsilon}{2}\cdot \left( 
\|x\|^2 + \|y\|^2 - x^\top \mw x - y^\top \mw y
\right).
\end{equation}
In the case of symmetric matrices, their definition is equivalent to the Spielman--Teng notion, and thus we use the same terminology {\em $\epsilon$-approximation} and notation 
$\widetilde{\mw}\approx_\epsilon\mw$ for their notion.

In this paper, we introduce a stronger notion of spectral approximation.  Specifically, we say
 $\widetilde{\mw}$ is a \emph{unit-circle $\epsilon$-approximation} of $\mw$ (written $\widetilde{\mw}\capprox_\epsilon\mw$) if
\begin{equation} \label{eqn:unitcircle-intro}
\forall x,y\in\C^n\qquad
\left|x^*(\mw-\widetilde{\mw})y\right|\leq \frac{\epsilon}{2}\cdot\left(\|x\|^2+\|y\|^2-\left|x^{*}\mw x+y^{*}\mw y\right|\right),
\end{equation}
where $v^*$ refers to the conjugate transpose of $v$. Note that we now allow
the vectors to range over $\C^n$ rather than $\R^n$.  As we show (see Section~\ref{sect:definitions}) this in itself does not make the definition stronger as the earlier notions of \cite{ST04,CKPPRSV17} have equivalent formulations using complex vectors. The more important change is the introduction of the complex magnitude $|\cdot|$ in the term $\left|x^{*}\mw x+y^{*}\mw y\right|$.  

The significance of this change can be seen by considering an eigenvector $v$ of $\mw$ whose eigenvalue $\lambda\in \C$ has magnitude 1.  Consider what happens if we set $x=y=v$ in both the Spielman--Teng definition (\ref{eqn:ST-intro}) and our definition (\ref{eqn:unitcircle-intro}).  If $\lambda=1$ (e.g. $v$ is a stationary distribution of the random walk specified by $\mw$), then the right-hand side of the inequality in both cases is zero, so we must have exact equality on the left-hand side, i.e.
$v^*\widetilde{\mw}v = v^* \mw v$.  On the other hand, if $\lambda$ is some other root of unity
 (e.g. an eigenvalue of the $k$-cycle $\mc_k$, or in any random walk with periodicity), then only our definition requires exact equality.  This also explains our terminology {\em unit-circle} approximation.  It also can be shown that $\widetilde{\mw}\capprox_\epsilon\mw$ if and only if 
$z\widetilde{\mw}\approx_\epsilon z\mw$ for all complex $z$ of magnitude 1.  That is, our definition amounts to demanding that all unit-circle multiples of the matrices approximate each other in the previous sense.  In the case of symmetric matrices (undirected graphs), it suffices to consider $z=\pm 1$, corresponding to the fact that the eigenvalues are all real and the only periodicity that can occur is 2.

The benefit of unit-circle approximation is that, unlike the previous notions of spectral approximation, it is preserved under cycle-lifts and powers.
\begin{lemma} \label{lem:unitcircle-powering-intro}
Suppose $\widetilde{\mw}\capprox_\epsilon \mw$.  Then for all $k\in \mathbb{N}$, we have:
    \begin{enumerate}
        \item $\mc_k \otimes \widetilde{\mw} \capprox_\epsilon \mc_k \otimes \mw$, and \label{itm:tensorpreserve-intro}
        \item $\widetilde{\mw}^k \capprox_{\epsilon/(1-\epsilon)} \mw^k$. \label{itm:powerpreserve-intro}
    \end{enumerate}
\end{lemma}
We note that in previous work (\cite{cheng2015,MRSV17}), it was observed that for symmetric matrices, if
$\widetilde{\mw}\approx_{\epsilon}\mw$
and 
$-\widetilde{\mw}\approx_{\epsilon}-\mw$
then we do get
$\widetilde{\mw}^2\approx_{\epsilon}\mw^2$. Lemma~\ref{lem:unitcircle-powering-intro} holds even for asymmetric matrices and handles all powers $k$.

Item~\ref{itm:tensorpreserve-intro} is proven by observing that the diagonalization of $\mc_k$ (using the discrete Fourier basis, which are its eigenvectors) has all $k$th roots of unity along the diagonal, so approximation of the cycle-lifted graphs
$\mc_k \otimes \widetilde{\mw}$ and $\mc_k \otimes \mw$
amounts to requiring that the approximation of
$\widetilde{\mw}$ and $\mw$
is preserved under multiplication by $k$'th roots of unity. 
Item~\ref{itm:powerpreserve-intro} is derived from Item~\ref{itm:tensorpreserve-intro} by observing that the $k$th powers can be obtained by ``shortcutting'' random walks through all but one layer of the cycle-lifted graphs. 
This shortcutting amounts to taking the Schur complements of the corresponding Laplacians, and we show that taking Schur complements of Eulerian Laplacians preserves spectral approximation (generalizing analogous results for undirected and symmetrized Schur complements~\cite{MP13,CKKPPRS18}).

Now we can sketch our algorithm that we use to prove  Theorem~\ref{thm:EulerianLaplacians-intro}.
Given an Eulerian digraph $G$, we want to approximate the pseudoinverse of the Laplacian $\mI_n-\mw$.  By standard reductions, we may assume that $G$ is regular, connected, and aperiodic, and therefore, it has polynomial mixing time.
Rather than directly approximating the inverse of the Laplacian $\mI_n-\mw$ of the original graph, we instead approximate the inverse of the Laplacian of the cycle-lifted graph, i.e. $\mI_{2^k\cdot n} - \mc_{2^k}\otimes \mw$, for $2^k$ larger than the mixing time of $\mw$. 
Then the pseudoinverse of $\mI_n-\mw$ can be well-approximated by an appropriate $n\times n$ projection of the pseudoinverse of $\mI_{2^k\cdot n} - \mc_{2^k}\otimes \mw$.

To approximate the pseudoinverse of $\mI_{2^k\cdot n} - \mc_{2^k}\otimes \mw$, we follow the recent approach of
\cite{CKKPPRS18} and recursively compute an LU factorization (i.e. a product of a lower-triangular and upper-triangular matrix) that approximates $\mI_{2^k\cdot n} - \mc_{2^k}\otimes \mw$, as LU factorizations can be easily inverted.  Each recursive step reduces the task to computing an LU factorization of the Laplacian of the random-walk on a chosen subset $S$ of the vertices, where we short-cut steps of the walk through $S^c$.  For our algorithm, we choose the set $S$ to consist of every other layer of the cycle-lifted graph, as opposed to using a randomly chosen and pruned set of vertices as in \cite{CKKPPRS18}.  Then shortcutting walks through $S^c$ yields a graph on $S$ whose transition matrix is equal to $\mc_{2^{k-1}}\otimes \mw^2$ --- a cycle-lifted version of the two-step random walk, with a cycle of half the length!   Unfortunately, we can't just directly recurse, because repeatedly squaring $\mw$ $k$ times takes space $O((k\cdot \log N))$.  Thus, following \cite{MRSV17}, we utilize the ``derandomized square'' of \cite{RozenmanVa05}, which produces an explicit sparse $\eps$-approximation to $\mw^2$ such that $k$ iterated derandomized squares can be computed in space $O(\log N+ k\cdot\log(1/\eps))=\tO(\log N)$.  (We take $\eps=1/O(k)$ so that we can tolerate incurring an $\eps$ error in approximation for each of the $k$ levels of recursion.)  To make the analysis work, we prove that for regular digraphs, the derandomized square produces a graph $\widetilde{\mw^2}$ that is a unit-circle approximation to $\mw^2$, so that we can deduce that $\mc_{2^{k-1}}\otimes \widetilde{\mw^2}$ approximates $\mc_{2^{k-1}}\otimes \mw^2$ via Lemma~\ref{lem:unitcircle-powering-intro}.  Previous work~\cite{MRSV17} only showed approximation for undirected graphs, and with respect to the original Spielman-Teng notion of spectral approximation.  (Rozenman and Vadhan~\cite{RozenmanVa05} showed that for regular digraphs, the derandomized square improves spectral expansion nearly as much as the true square, but that is weaker than spectral approximation, as it only refers to the 2nd singular value rather than the entire spectrum.)

We remark that the $n\times n$ projection of the pseudoinverse of the approximate LU factorization we obtain is exactly the matrix we would get if we applied the repeated-squaring-based Laplacian inversion algorithm of Peng and Spielman~\cite{PS13} (or, more accurately, its space-efficient implementation via derandomized squaring \cite{MRSV17}) to the original Laplacian $\mI-\mw$.  Thus, another conceptual contribution of our paper is connecting the LU factorization approach of \cite{CKKPPRS18} to the squaring-based approach of \cite{PS13} via the cycle-lifted graph.  (However, for technical reasons in our analysis, we don't do the $n\times n$ projection until after applying the error-reduction procedure to obtain a highly accurate pseudoinverse of the cycle-lifted Laplacian.)

\section{Preliminaries}
\label{sec:prelim}

In this section we introduce notation and facts we use through out the paper.
\subsection{Notation}
We denote by $\C$ the set of complex numbers. For $w = x + yi \in \C$, we use $w^* = x - yi$ to denote the conjugate of $w$. We use $|w|=\sqrt{x^2+y^2}$ to denote the magnitude of $w$.
\paragraph{Matrices and vectors.} We use bold capital letters to denote matrices. We use $\mI_n\in \R^{n\times n}$ to denote the identity matrix. For a matrix $\ma \in \C^{n\times n}$ we use $\ma^*$ to denote the conjugate transpose of $\ma$ and we write $\mU_\ma = \frac{\ma + \ma^*}{2}$ to denote its \emph{symmetrization}. We use $\vec{\bone}_k \in \R^k$ to denote the all $1$'s vector or just $\vec{\bone}$ when $k$ is clear from context.  We denote the conjugate transpose of a vector similarly. We use $\ma^\top$ to denote the transpose of a real matrix.

\paragraph{Positive Semidefinite (PSD) matrices.} For Hermitian matrices $\ma,\mb \in \C^{n\times n}$ we say $\ma$ is PSD or write $\ma \succeq 0$ if $x^* \ma x \geq 0$ for all $x\in \C^n$. If $\ma$ is real the condition is equivalent to $x^\top \ma x\geq 0$ for all $x\in \R^n$. Further we use $\ma \succeq \mb$ to denote $\ma - \mb \succeq 0$. We define $\preceq$, $\prec$, and $\succ$ analogously.
\begin{proposition}\label{prop:psd_mult_sides}
Given a PSD matrix $\ma \in \C^{n\times n}$ and matrix $\mb \in \C^{n\times m}$ 
$$
\mb^* \ma \mb \succeq 0.
$$
\end{proposition}

\paragraph{Pseudo-inverse and square root of matrices.}
For a matrix $\ma$, we use $\ma^+$ to denote the (Moore-Penrose) pseudo-inverse of $\ma$.
For a PSD matrix $\mb$, we let $\mb^{1/2}$ to denote the square root of $\mb$, which is the unique PSD matrix such that $\mb^{1/2}\mb^{1/2} = \mb$. Furthermore, we let $\mb^{+/2}$
denote the pseudo-inverse of the square root of $\mb$. 

\paragraph{Operator norms.}
For any vector norm $\|\cdot\|$ defined on $\C^{n}$ we define the operator semi-norm it induces on $\C^{n\times n}$ by $\|\ma\| = \max_{x\neq 0} \frac{\|\ma x\|}{\|x\|}$. For a PSD matrix $\mh$ and vector $x$ we let $\|x\|_{\mh} = \sqrt{x^* \mh x}$, and define the operator semi-norm $\|\ma\|_{\mh}$ accordingly. We can relate the $\|\cdot\|_\mh$ and $\|\cdot\|_2$ operator norms as follows. For a matrix $\ma$, we have  $\|\ma\|_\mh = \|\mh^{1/2} \ma \mh^{+/2}\|_2$. We use the term \emph{spectral norm} to refer to the operator norm induced by $\|\cdot\|_2$. We write $\|\ma|_{V}\|_2$ to denote the spectral norm restricted to a subspace $V$. That is, $\|\ma|_{V}\|_2=\max_{x\in V,x\neq 0} \frac{\|\ma x\|}{\|x\|}$.

\begin{lemma}
\label{lem:extension_to_complex}
Let $\mm\colon\mathbb{R}^n\rightarrow\mathbb{R}^n$ be a linear operator and extend it to $\mm\colon\mathbb{C}^n\rightarrow\mathbb{C}^n$ by defining $\mm(u+iv)=\mm u+i\mm v$ for all $z=u+iv\in \mathbb{C}^{n}$. Then $\|\mm\|_{\mathbb{C}^n\rightarrow\mathbb{C}^n}=\|\mm\|_{\mathbb{R}^n\rightarrow\mathbb{R}^n}$.
\end{lemma}

\paragraph{Graphs}
Throughout this paper we work with unweighted directed multigraphs (digraphs). These graphs can have parallel edges and self loops and can be viewed as digraphs with integer edge weighs. We specify graphs by $G=(V,E)$ where $V$ is the set of vertices and $E$ is the multiset of edges. 

\paragraph{Adjacency and Random Walk Matrices.} The adjacency matrix of a digraph $G$ on $n$ vertices is the matrix $\ma\in\R^{n\times n}$ where $\ma_{ij}$ is the number of edges from vertex $j$ to vertex $i$ in $G$.\footnote{Often the adjacency matrix is defined to be $\ma^\top$ but we find the current formulation more convenient for our purposes.} The degree matrix $\md$ of a digraph $G$ is the diagonal matrix containing the out-degrees of the vertices in $G$. The random walk matrix or transition matrix of a digraph $G$ is $\mw=\ma\md^{-1}$. $\mw_{ij}$ is the probability that a random step from vertex $j$ leads to $i$ in $G$. Note that $\vec{\bone}^\top \mw = \vec{\bone}^\top$. A matrix $\mw \in \R_{\geq 0}^{n\times n}$ is called substochastic if $\vec{\bone}^\top \mw \leq \vec{\bone}^\top$ (the inequality is entry-wise). 

\paragraph{Directed Laplacians.}
We follow the approach in \cite{CKKPPRS18} to define graph Laplacians. A matrix $\mL \in \R^{n\times n}$ is a directed Laplacian, if its off-diagonal entries are non-positive, i.e. $\mL_{ij} \leq 0$ for $i \neq j$, and $\vec{\bone}^\top \mL = 0$. 
Every digraph is associated with a directed Laplacian.
Occasionally we write $\mL = \md - \ma$ to express the decomposition of $\mL$ into the degree matrix and adjacency matrix of the corresponding digraph.
The random-walk Laplacian of a digraph with Laplacian $\md-\ma$ is the matrix $(\md-\ma)\md^{-1}=\mI-\mw$, where $\mw$ is the transition matrix of $G$. We will often write $\mc_k$ to denote the adjacency matrix of the $k$-vertex uni-directional directed cycle. 

\paragraph{Eulerian graphs and Eulerian Laplacians.} A directed graph is Eulerian if the in-degree of every node is equal to its out-degree. A directed Laplacian $\mL$ is Eulerian if $\mL \vec{\bone} = 0$. A graph is Eulerian if and only if its Laplacian is Eulerian.

\subsection{Kronecker Product}
Given matrices $\ma\in \C^{n\times m}, \mb\in \C^{p\times q}$, the Kronecker product or tensor product of $\ma$ and $\mb$ denoted by $\ma \otimes \mb \in \C^{pn\times qm}$ is
$$
\ma\otimes \mb = 
\begin{bmatrix}
\ma_{11}\mb & \ma_{12}\mb & \cdots & \ma_{1m}\mb\\
\vdots & \vdots & \vdots & \vdots\\
\ma_{n1}\mb & \ma_{n2}\mb & \cdots & \ma_{nm}\mb
\end{bmatrix}.
$$

\begin{proposition}\label{prop:abcd}
Given four matrices $\ma$, $\mb$, $\mc$, and $\md$, if the matrix dimensions make $\ma \mc$ and $\mb \md$ well-defined, then
$$
(\ma\otimes \mb)(\mc\otimes \md) = (\ma \mc)\otimes (\mb \md).
$$
\end{proposition}

\subsection{Schur Complement}
For a matrix $\ma \in \C^{n\times n}$ and sets $F,C\subseteq [n]$, let $\ma_{FC}$ denote the submatrix corresponding to the rows in $F$ and columns in $C$. Similarly, for a vector $v \in \C^n$ let $v_F \in \C^{|F|}$ be the restriction of $v$ onto coordinates in $F$. If $F, C$ partition $[n]$ and $\ma_{FF}$ is invertible, then we denote the Schur complement of $A$ onto the set $C$
by
$$
\schur{(\ma, C)} \eqdef \ma_{CC} - \ma_{CF} \ma_{FF}^{-1} \ma_{FC}.
$$
When it is clear from context we may reload this notation as follows to make the Schur complement dimension consistent with $\ma$.
$$
\schur{(\ma, C)} \eqdef
\begin{bmatrix}
0_{FF} & 0_{FC} \\
0_{CF} & \ma_{CC} - \ma_{CF} \ma_{FF}^{-1} \ma_{FC}
\end{bmatrix}.
$$

\section{Spectral Approximation}
\label{sect:definitions}
Since its introduction by Spielman and Teng \cite{ST11}, spectral approximation of graphs and their associated matrices \cite{ST11} has served as a powerful tool for graph-theoretic algorithm development. Below we review the original definition and later generalizations to directed graphs and asymmetric matrices \cite{CKPPRSV17}, and then present our new, stronger definition of unit-circle approximation in several equivalent formulations.

\subsection{Definitions}
\begin{definition}[Undirected Spectral Approximation \cite{ST11}]
\label{def:undir_approx}
Let $\mw,\widetilde{\mw}\in\R^{n\times n}$ be symmetric matrices. We say that $\widetilde{\mw}$ is an \emph{undirected $\epsilon$-approximation} of $\mw$ (written $\widetilde{\mw}\approx_\epsilon\mw$) if 
\[
\forall x\in\R^n,~(1-\epsilon)\cdot x^{\top}(\mI-\mw)x\leq x^{\top}(\mI-\widetilde{\mw})x \leq (1+\epsilon)\cdot x^{\top}(\mI-\mw)x
\]
or equivalently,
\[
\forall x\in\R^n,~\left|x^\top(\mw-\widetilde{\mw})x\right|\leq \epsilon\cdot x^{\top}(\mI-\mw)x=\epsilon\cdot\left(\|x\|^2-x^{\top}\mw x\right).
\]
\end{definition}

Typically Definition \ref{def:undir_approx} is phrased in terms of Laplacian matrices of the form $\mI-\mw$ and approximation is denoted by $\mI-\widetilde{\mw}\approx_{\epsilon}\mI-\mw$ to indicate the multiplicative approximation between the quadratic forms defined by $\mI-\widetilde{\mw}$ and $\mI-\mw$. However, in the more general definitions of this paper it will be more convenient to think of spectral approximation as a measure of approximation between $\widetilde{\mw}$ and $\mw$ rather than between $\mI-\widetilde{\mw}$ and $\mI-\mw$. Note that the definition is asymmetric in $\widetilde{\mw}$ and $\mw$ but $\widetilde{\mw}\approx_{\epsilon}\mw$ for $\epsilon<1$ implies $\mw\approx_{\epsilon/(1-\epsilon)}\widetilde{\mw}$.

Spectral approximation is a strong definition that guarantees the two matrices have similar eigenvalues, and their corresponding graphs have similar cuts and random walk behavior \cite{ST11,batson2013spectral}. Below we show the generalization to directed graphs from \cite{CKPPRSV17}.
\begin{definition}[Directed Spectral Approximation \cite{CKPPRSV17}]
\label{def:dir_approx}
Let $\mw,\widetilde{\mw}\in\R^{n\times n}$ be (possibly asymmetric) matrices. We say that $\widetilde{\mw}$ is a \emph{directed $\epsilon$-approximation} of $\mw$ (written $\widetilde{\mw}\approx_\epsilon\mw$) if
\begin{eqnarray*}
\forall x,y \in\R^n,~\left|x^\top(\mw-\widetilde{\mw})y\right|&\leq& \frac{\epsilon}{2}\cdot \left(x^{\top}(\mI-\mw)x+y^{\top}(\mI-\mw)y\right)\\
&=&\frac{\epsilon}{2}\cdot\left(\|x\|^2+\|y\|^2-x^{\top}\mw x-y^{\top}\mw y\right)\\
&=&\frac{\epsilon}{2}\cdot\left(\|x\|^2+\|y\|^2-x^{\top}\mU_{\mw} x-y^{\top}\mU_{\mw} y\right).
\end{eqnarray*}
\end{definition}
The main difference between the above and Definition~\ref{def:undir_approx} is the introduction of the $y$ vector instead of having $y=x$. Indeed, using the same vector on both sides would lose the asymmetric information in the matrices $\mw$ and $\widetilde{\mw})$. However, note that the last inequality shows that the right-hand side depends only on the 
symmetrization $\mU_{\mw}$.

We are justified using the same notation for undirected and directed spectral approximation because of the following lemma 
\begin{lemma}[\cite{MRSV19} Lemma 2.9]
\label{lem:undir_dir_equiv}
Let $\mw,\widetilde{\mw}\in\R^{n\times n}$ be symmetric matrices. Then $\widetilde{\mw}$ is a directed $\epsilon$-approximation of $\mw$ if and only if it is an undirected $\epsilon$-approximation of $\mw$.
\end{lemma}

It will be convenient for us to generalize Definition \ref{def:dir_approx} to complex matrices. In that case, we will quantify over $x,y\in\C^n$ and replace the transposes with Hermitian transposes.  

\begin{definition}[Complex Spectral Approximation]
\label{def:complex_approx}
Let $\mw,\widetilde{\mw}\in\C^{n\times n}$ be (possibly asymmetric) matrices. We say that $\widetilde{\mw}$ is a \emph{complex $\epsilon$-approximation} of $\mw$ (written $\widetilde{\mw}\approx_{\epsilon}\mw$) if 
\begin{eqnarray*}
\forall x, y\in\C^n,~\left|x^*(\mw-\widetilde{\mw})y\right|
&\leq& \frac{\epsilon}{2}\cdot\left(\|x\|^2+\|y\|^2-x^{*}\mU_{\mw} x-y^{*}\mU_{\mw} y\right)\\
&=& \frac{\epsilon}{2}\cdot\left(\|x\|^2+\|y\|^2-\Real\left(x^{*}\mw x+y^{*}\mw y\right)\right)\\
\end{eqnarray*}
\end{definition}
The equality in Definition \ref{def:complex_approx} comes from the observation that for all $v\in\C^n$ and all matrices $\ma\in\C^{n\times n}$ we have
\[
    v^{*}\mU_{\ma} v =\frac{1}{2}\left( v^{*}\ma v+v^{*}\ma^*v\right)=\frac{1}{2}\left( v^{*}\ma v+(v^{*}\ma v)^*\right)=\Real(v^{*}\ma v)
\]
because the average of a complex number and its conjugate is simply its real part. Notice that the definitions of undirected, directed, and complex spectral approximation are only achievable when the matrix $\mw$ has the property that $\Real(x^{*}\mw x)\leq \|x\|^2$ for all vectors $x\in\C^n$ (when $\mw$ is real and symmetric as in the case of undirected spectral approximation, this requirement is equivalent to $x^{\top}\mw x\leq \|x\|^2$ for all $x\in\R^n$). When working with these types of approximation, we will often implicitly restrain the matrices to have this property. 

Again, we are justified in using the same notation for complex approximation that we use for directed and undirected approximations because of the following lemma.
\begin{lemma}
\label{lem:dirapprox_to_complex}
Let $\mw,\widetilde{\mw}\in\R^{n\times n}$ be (possibly asymmetric) matrices. Then $\widetilde{\mw}$ is a directed $\epsilon$-approximation of $\mw$ if and only if $\widetilde{\mw}$ is a complex $\epsilon$-approximation of $\mw$.
\end{lemma}
A proof of Lemma \ref{lem:dirapprox_to_complex} can be found in Appendix \ref{app:definitions}. Now we introduce our new stronger definition, which we call \emph{unit-circle spectral approximation}.

\begin{definition}[Unit-circle Spectral Approximation]
\label{def:unitcirc_approx}
Let $\mw,\widetilde{\mw}\in\C^{n\times n}$ be (possibly asymmetric) matrices. We say that $\widetilde{\mw}$ is a \emph{unit-circle $\epsilon$-approximation} of $\mw$ (written $\widetilde{\mw}\capprox_\epsilon\mw$) if  
\[
\forall x,y\in\C^n,~ \left|x^*(\mw-\widetilde{\mw})y\right|\leq \frac{\epsilon}{2}\cdot\left(\|x\|^2+\|y\|^2-\left|x^{*}\mw x+y^{*}\mw y\right|\right).
\]
\end{definition}

The change from Definition \ref{def:complex_approx} is that we have replaced the real part with the complex magnitude $|\cdot|$ on the quadratic forms $x^{*}\mw x+y^{*}\mw y$ on the right-hand side. To understand what we gain from this, suppose $x=y$ is an eigenvector of $\mw$ with eigenvalue $\lambda$ such that $|\lambda|=1$. Then the right-hand side of the inequality equals zero and so we must have $x^*\widetilde{\mw}y=x^*\mw y$. In other words, $\widetilde{\mw}$ and $\mw$ must behave identically on the entire unit circle of eigenvalues with magnitude 1. This is in contrast to the previous definitions, which only required exact preservation in the case where $\lambda=1$. For example, can an undirected bipartite graph (which has a periodicity of 2 and an eigenvalue of $-1$) have a non-bipartite spectral approximation? Under previous definitions, the answer is yes but under unit-circle approximation, the answer is no because we require exact preservation on all eigenvalues of magnitude 1, not just $\lambda=1$.

Unit circle approximation applies to a smaller class of matrices than the previous definitions of spectral approximation. While the previous definitions only required that $\Real(x^*\mw x)\leq \|x\|^2$ for all $x\in\C^n$, unit circle approximation requires that $|x^*\mw x|\leq\|x\|^2$ for all $x\in\C^n$. Again, we will often implicitly restrict our matrices to have this property. Note that all complex matrices $\mw$ such that $\|\mw\|_{1}\leq 1$ and $\|\mw\|_{\infty}\leq 1$ satisfy this property. In particular, if $\mw$ is the transition matrix of an Eulerian graph, then $\mw$ satsifies the property as does $z\cdot\mw$ for all $z\in\C$ such that $|z|\leq 1$.  

We will see in the coming sections that unit-circle approximation is preserved under powering of $\widetilde{\mw}$ and $\mw$ and is useful for achieving spectral approximation of a class of graphs we call \emph{cycle-lifted graphs}, which are essential for the anlaysis of our Eulerian Laplacian solver.

\subsection{Equivalent Formulations}
There are many useful equivalent formulations of Definition~\ref{def:unitcirc_approx}. First we look at what our definition gives in the case of real and symmetric matrices. 
\begin{lemma}[Real, Symmetric Equivalence]
\label{lem:real_sym_equiv}
Let $\mw,\widetilde{\mw}\in\R^{n\times n}$ be symmetric matrices. Then the following are equivalent:
\begin{enumerate}
    \item $\widetilde{\mw}\capprox_{\epsilon}\mw$.
    \item $\widetilde{\mw}\approx_{\epsilon}\mw$ and $-\widetilde{\mw}\approx_{\epsilon}-\mw$.
    \item For all $x\in\R^{n}$ we have
\[
\left|x^{\top}(\mw-\widetilde{\mw})x\right|\leq \epsilon\cdot\left(\|x\|^2-|x^{\top}\mw x|\right).
\]
\end{enumerate} 
\end{lemma}
A proof of Lemma \ref{lem:real_sym_equiv} can be found in Appendix \ref{app:definitions}. In the original \cite{ST11} formulation of spectral approximation as multiplicative approximation between quadratic forms, Lemma \ref{lem:real_sym_equiv} says that in the real, symmetric setting, unit-circle spectral approximation is equivalent to $\mI-\widetilde{\mw}$ approximating $\mI-\mw$ \emph{and} $\mI+\widetilde{\mw}$ approximating $\mI+\mw$. This makes intuitive sense because symmetric matrices have real eigenvalues so the only eigenvalues that can lie on the unit circle are $+1$ and $-1$.

This ``plus and minus'' approximation has been studied before in \cite{cheng2015,MRSV19}, where it was found to be useful because spectral approximation is preserved under squaring when both the ``plus'' and ``minus'' approximations hold. We will see in Section \ref{sect:ring_powers} that even in the general directed, complex case, unit-circle approximation is preserved under all powering.  

We now show some convenient equivalent formulations of unit-circle spectral approximation. 

\begin{lemma}
\label{lem:unitcirc_equivalences}
Let $\mw,\widetilde{\mw}\in\mathbb{C}^{n\times n}$ be (possibly asymmetric) matrices. Then the following are equivalent 
\begin{enumerate}
    \item $\widetilde{\mw}\capprox_{\epsilon}\mw$
    \item For all $z\in\C$ such that $|z|=1$, $z\cdot\widetilde{\mw}\approx_{\epsilon}z\cdot\mw$
    \item For all $z\in\C$ such that $|z|=1$,  
    \begin{itemize}
        \item $\mathrm{ker}(\mU_{\mI-z\cdot \mw})\subseteq \mathrm{ker}(\widetilde{\mw}-\mw)\cap\mathrm{ker}((\widetilde{\mw}-\mw)^\top)$ and
        \item $\left\|\mU_{\mI-z\cdot \mw}^{+/2}(\widetilde{\mw}-\mw)\mU_{\mI-z\cdot \mw}^{+/2}\right\|\leq \epsilon$
    \end{itemize} 
\end{enumerate}
\end{lemma}

A proof of Lemma \ref{lem:unitcirc_equivalences} can be found in Appendix \ref{app:definitions}.

\section{Approximating Cycle-Lifted Graphs and Powers}
\label{sect:ring_powers}
In this section we discuss how unit-circle spectral approximation allows us to approximate powers of random walk matrices of digraphs and a class of graphs we call \emph{cycle-lifted graphs}, which play an essential role in our Eulerian Laplacian solver. Preservation under powering is a useful property for a definition of matrix approximation but even in the case of symmetric transition matrices, the original definition of spectral approximation does not guarantee this, as is seen in the following proposition. 

\begin{proposition}
    For all rational $\epsilon\in(0,1)$, there exist undirected graphs with transition matrices $\widetilde{\mw},\mw$ such that $\widetilde{\mw}\approx_{\epsilon}\mw$ but $\widetilde{\mw}^2\not\approx_{c}\mw^2$ for any finite $c>0$.
\end{proposition}
\begin{proof}
Let $\mw$ be the transition matrix of a connected undirected bipartite graph and define $\widetilde{\mw}=(1-\epsilon)\cdot\mw+\epsilon\cdot \mI$. Fix a vector $x$ and observe 
\begin{align*}
    \left|x^\top(\mw-\widetilde{\mw})x\right|&=\epsilon\cdot\left|x^\top(\mI-\mw)x\right|\\
    &=\epsilon\cdot\left(\|x\|^2-x^{\top}\mw x\right)
\end{align*}
where we can drop the absolute value in the second line because $\mI-\mw$ is PSD and hence has a non-negative quadratic form. So $\widetilde{\mw}\approx_{\epsilon}\mw$. However $\mw^2$ has 2 connected components (because walks of length 2 must start and end at the same side of the bipartition) while $\widetilde{\mw}^2$ is strongly connected. This means that  $\mw^2$ has eigenvalue $\lambda=1$ with multiplicity 2 while $\widetilde{\mw}^2$ has eigenvalue $\lambda=1$ with multiplicity 1. Hence, they cannot spectrally approximate one another (see Lemma \ref{lem:exact_eigenspace} below).
\end{proof}

\begin{lemma}
\label{lem:exact_eigenspace}
Fix $\widetilde{\mw},\mw\in\C^{n\times n}$. For any matrix $M\in \C^{n\times n}$, let $V_{\lambda}(M)$ denote the eigenspace of $M$ of eigenvalue $\lambda$. 
\begin{enumerate}
\item If $\widetilde{\mw}\approx_{c}\mw$ for a finite $c>0$, then $V_1(\mw)\subseteq V_1(\widetilde{\mw})$. If $c<1$, then $V_1(\mw)= V_1(\widetilde{\mw})$.
\item If $\widetilde{\mw}\capprox_{c}\mw$ for a finite $c>0$, then for all $\lambda\in \C$ such that $|\lambda|=1$, $V_{\lambda}(\mw)\subseteq V_{\lambda}(\widetilde{\mw})$ and if $c<1$ then $V_{\lambda}(\mw)= V_{\lambda}(\widetilde{\mw})$.
\end{enumerate}
\end{lemma}
A proof of Lemma \ref{lem:exact_eigenspace} can be found in Appendix \ref{app:ring_powers}. In previous work (\cite{cheng2015,MRSV17}), it was observed that for symmetric matrices, if
$\widetilde{\mw}\approx_{\epsilon}\mw$
and 
$-\widetilde{\mw}\approx_{\epsilon}-\mw$
then we do get
$\widetilde{\mw}^2\approx_{\epsilon}\mw^2$.
Furthermore, when $\mw$ is PSD we have that  $\widetilde{\mw}\approx_{\epsilon}\mw$ implies $-\widetilde{\mw}\approx_{\epsilon}-\mw$. Since $\mw^2$ is trivially PSD, the above can be applied recursively to conclude that $\widetilde{\mw}^{2^k}\approx_{\epsilon}\mw^{2^k}$ for all positive integers $k$. However, we observe that analogous approximation guarantees do not hold for Eulerian graphs (or even regular digraphs).

\begin{proposition}
    For all rational $\epsilon\in(0,1)$, there exist regular digraphs with transition matrices $\widetilde{\mw},\mw$ such that $\widetilde{\mw}\approx_{\epsilon}\mw$ and $-\widetilde{\mw}\approx_{\epsilon}-\mw$ but $\widetilde{\mw}^4\not\approx_{c}\mw^4$ for any finite $c$.
\end{proposition}
\begin{proof}
Let $\mc_4$ be the transition matrix of the directed 4-cycle. Fix $\epsilon\in(0,1)$ and define $\widetilde{\mc}=(1-\epsilon/2)\cdot \mc_4+(\epsilon/2)\cdot \mc_4^\top$. In other words, $\widetilde{\mc}$ is the directed 4-cycle with an $\epsilon/2$ probability of traversing backwards. We claim that  $\widetilde{\mc}\approx_{\epsilon}\mc_4$ and $-\widetilde{\mc}\approx_{\epsilon}-\mc_4$ but $\widetilde{\mc}^4$ does not approximate $\mc_4^4$. 

First, note that each of $\mU_{\mI\pm \mc_4}$ has a one-dimensional kernel, namely
\begin{align*}
    \ker(\mU_{\mI-\mc_4})&=\mathrm{span}(\vec{\bone})\\
    \ker(\mU_{\mI+\mc_4})&=\mathrm{span}\left([1,-1,1,-1]^{\top}\right).
\end{align*}
Furthermore we have,
\[
(\widetilde{\mc}-\mc_4)=(\epsilon/2)\cdot(\mc_4-\mc_4^\top),
\]
which has a two-dimensional kernel,  $\mathrm{span}(\vec{\bone},[1,-1,1,-1])$. Observe that $\mU_{\mI\pm \mc_4}$ each has eigenvalues of $2$, $1$, $1$, and $0$. Hence $\|\mU_{\mI\pm \mc}^{+/2}\|=1$. Finally, we have that 
\[
\|\mc_4-\widetilde{\mc}\|=(\epsilon/2)\cdot\|\mc_4-\mc_4^\top\|=\epsilon 
\]
Putting this together gives
\begin{align*}
\|\mU_{\mI-\mc_4}^{+/2}(\mc_4-\widetilde{\mc})\mU_{\mI-\mc_4}^{+/2}\|&\leq \|\mU_{\mI-\mc_4}^{+/2}\|\cdot\|\mc_4-\widetilde{\mc}\|\cdot\|\mU_{\mI-\mc_4}^{+/2}\|\\
&=\epsilon.
\end{align*} 
It follows that $\widetilde{\mc}\approx_{\epsilon}\mc_4$. A similar calculation lets us conclude that $-\widetilde{\mc}\approx_{\epsilon}-\mc_4$, as desired.

Notice that $\mc_4^4$ has 4 connected components (all of the vertices become isolated with self loops) while $\widetilde{\mc}^4$ has 2 connected components. So $\mc_4^4$ and $\widetilde{\mc}^4$ have eigenvalues of 1 with different multiplicities and hence they fail to spectrally approximate of one another by Lemma \ref{lem:exact_eigenspace}. 
\end{proof}

Now we show that if a matrix is a unit-circle approximation of another, then all of their powers are as well (with small loss in approximation quality). In fact, we show something stronger, namely that their \emph{cycle-lifted graphs} approximate each other. 

\begin{definition}[Cycle-Lifted Graph]
Let $\mc_k$ denote the transition matrix of the $k$-vertex directed cycle. Given a graph $G=(V,E)$ on $n$ vertices with transition matrix $\mw$ the \emph{cycle-lifted graph of length $k$}, $C_k(G)$, is a layered graph with $k$ layers (numbered $1$ to $k$) of $n$ vertices each, where for every $i\in[k]$, there is an edge from vertex $u$ in layer $i$ to vertex $v$ in layer $(i+1)\mod k$ with multiplicity $\ell$ if and only if $(u,v)$ exists with multiplicity $\ell$ in $G$. That is, $C_k(G)=(V',E')$ with $V'=[k]\times V$ and $E'=\{((i,u),(i+1\mod k,v)\colon (u,v)\in E\}$. The transition matrix of $C_k(G)$ is $\mc_k \otimes \mw$.
\end{definition}

\begin{theorem}
\label{thm:capprox_ring}
Fix $\mw,\widetilde{\mw}\in\C^{n\times n}$ and let $\mc_{k}$ be the transition matrix for the directed cycle on $k$ vertices. Then $\mc_{k}\otimes \widetilde{\mw}\approx_{\epsilon}\mc_{k} \otimes \mw$ if and only if for all $z~\mathrm{such~that~} z^k= 1$, we have $
z\cdot\widetilde{\mw}\approx_{\epsilon}z\cdot\mw$.
\end{theorem}
Recall that unit-circle spectral approximation requires that for all $z\in\C$ with $|z|=1$ we have $z\cdot \widetilde{\mw}\approx_{\epsilon}z\cdot\mw$. Theorem \ref{thm:capprox_ring} then tells us that unit-circle spectral approximation implies approximations of the corresponding cycle-lifted graphs of \emph{every length}. 
\begin{corollary}
\label{cor:capprox_ring}
Fix $\mw,\widetilde{\mw}\in\C^{n\times n}$. If $\widetilde{\mw}\capprox_{\epsilon}\mw$ then for all positive integers $k$, $\mc_{k}\otimes \widetilde{\mw}\capprox_{\epsilon}\mc_{k} \otimes \mw$.
\end{corollary}
\begin{proof}
Since $\widetilde{\mw}\capprox_{\epsilon}\mw$, we have that for all $\omega\in\C$ such that $|\omega|=1$ and all $z$ such that $z^k=1$, we have $z\cdot w\cdot\widetilde{\mw}\approx_{\epsilon}z\cdot w\cdot\mw$. By Theorem \ref{thm:capprox_ring}, for all $\omega\in\C$ such that $|\omega|=1$, we have 
\[
\omega\cdot(\mc_k\otimes\widetilde{\mw})=\mc_k\otimes(\omega\cdot\widetilde{\mw})\approx_{\epsilon}\mc_k\otimes(\omega\cdot\mw)=\omega\cdot(\mc_k\otimes\mw).
\]
Thus, $\mc_k\otimes\widetilde{\mw}\capprox_{\epsilon}\mc_k\otimes\mw$.
\end{proof}
To prove Theorem \ref{thm:capprox_ring}, we use the following lemma, which simplifies calculating spectral norm and hence spectral approximation using restrictions onto invariant subspaces. 

\begin{lemma}
\label{lem:subspaces}
Let $\mm\colon\mathbb{C}^n\rightarrow\mathbb{C}^n$ be a linear operator and $V_1,\ldots,V_{\ell}\subseteq\mathbb{C}^n$ subspaces such that 
\begin{enumerate}
    \item $V_{j}\perp V_{k}$ for all $j\neq k$
    \item $V_1 \oplus \ldots \oplus V_\ell=\mathbb{C}^n$
    \item $\mm V_j\subseteq V_j$ for all $j\in[\ell]$.
\end{enumerate}
i.e., $M$ is block diagonal with respect to the subspaces $V_1,\ldots,V_\ell$. Then, 
\[\|\mm\|=\max_{j\in[\ell]}\|\mm|_{V_j}\|\]
\end{lemma}
A proof of Lemma \ref{lem:subspaces} can be found in Appendix \ref{app:ring_powers}. Now we can prove Theorem \ref{thm:capprox_ring}.

\begin{proof}[Proof of Theorem \ref{thm:capprox_ring}]
 The intuition behind the proof is to observe that the diagonalization of the directed $k$-cycle $\mc_k$ using the discrete Fourier basis (which are its eigenvectors), has all $k$th roots of unity along the diagonal. This means that approximation of cycle-lifted graphs $\mc_k \otimes \widetilde{\mw}$ and $\mc_k \otimes \mw$
amounts to requiring that the approximation of
$\widetilde{\mw}$ and $\mw$
is preserved under multiplication by $k$'th roots of unity. 
 
 For each $z\in\C$ with $z^k=1$, define the Fourier basis vector $\chi_z=[1,z,z^2,\ldots,z^{k-1}]^\top$. Consider the subspaces $V_z=\mathrm{span}(\chi_z \otimes \mathbb{C}^n)$. These subspaces are orthogonal, span $\mathbb{C}^k\otimes \mathbb{C}^n$ and are invariant under $\mc_k \otimes \mw$, $\mc_k \otimes \widetilde{\mw}$, $(\mc_k \otimes \mw)^*=\mc_k^* \otimes \mw^*$, $(\mc_k \otimes \widetilde{\mw})^*=\mc_k^* \otimes \widetilde{\mw}^*$, and
$\mU^{+/2}_{\mI_{k\cdot n}-\mc_k \otimes \mw}$. So, by Lemma \ref{lem:subspaces} we have
\begin{align*}
    &\left\|\mU^{+/2}_{\mI_{k\cdot n}-\mc_k \otimes \mw}\left(\mc_k \otimes \widetilde{\mw}-\mc_k \otimes \mw\right)\mU^{+/2}_{\mI_{k\cdot n} - \mc_k \otimes \mw}\right\| = \\ &\max_{z\colon z^k=1} \left\|\mU^{+/2}_{\mI_{k\cdot n}- \mc_k \otimes \mw}\left(\mc_k \otimes \widetilde{\mw}-\mc_k \otimes \mw\right)\mU^{+/2}_{\mI_{k\cdot n}-\mc_k \otimes \mw}|_{V_z}\right\|.
\end{align*}

Now observe that since $\mc_k\chi_z=z^{-1}\chi_z$, we have that for any $\vv\in\mathbb{C}^n$
\begin{align*}
(\mc_k \otimes \mw)(\chi_z \otimes \vv)&=\mc_k\chi_z\otimes \mw \vv\\
&=\chi_z \otimes (z^{-1}\cdot\mw \vv).
\end{align*}
Similarly we have,
\begin{align*}
    (\mc_k\otimes \widetilde{\mw})(\chi_z\otimes \vv)&=\chi_z \otimes (z^{-1}\cdot\widetilde{\mw}\vv)\\
    \mU_{\mI_{k\cdot n-\mc_k\otimes \mw}}(\chi_z\otimes\vv)&=\chi_z \otimes \mU_{\mI_n-z^{-1}\cdot\mw}\vv.
\end{align*}
It follows that
\begin{align*}
 &\left\|\mU^{+/2}_{\mI_{k\cdot n}-\mc_k\otimes \mw }\left(\mc_k \otimes \widetilde{\mw}-\mc_k \otimes \mw\right)\mU^{+/2}_{\mI_{k\cdot n}-\mc_k \otimes \mw}\right\| \\
 &=\max_{z\colon z^k=1}\left\|\mU^{+/2}_{\mI_n-z^{-1}\cdot\mw}\left(z^{-1}\cdot\mw-z^{-1}\cdot\widetilde{\mw}\right)\mU^{+/2}_{\mI_n-z^{-1}\cdot\mw}\right\|.
\end{align*}
So we get that $\mc_{k}\otimes \widetilde{\mw}\approx_{\epsilon}\mc_{k} \otimes \mw$ if and only if for all $z$ such that $z^k=1$,  $z^{-1}\widetilde{\mw}\approx_{\epsilon}z^{-1}\mw$. Since for all $k$th roots of unity $z$, $z^{-1}$ is also a $k$th root of unity, the result follows. 
\end{proof}

Theorem \ref{thm:capprox_ring} allows us to reason about approximation under powering by observing that the $k$th power of a matrix can be expressed in terms of the Schur complement of its cycle-lifted graph of length $k$. In \cite{MP13}, they showed that undirected spectral approximation is preserved under Schur complements. Here we show that the same is true of directed spectral approximation (with a small loss in approximation quality).

\begin{theorem}
\label{thm:sc_approx}
Fix $\mw,\widetilde{\mw}\in\C^{n}$ and suppose that $\widetilde{\mw}\approx_{\epsilon}\mw$ for $\epsilon\in(0,2/3)$. Let $F\subseteq[n]$ such that $(\mI_{|F|}-\mw_{FF})$ is invertible and let $C=[n]\setminus F$. Then 
\[
\mI_{|C|}-\schur(\mI_n-\widetilde{\mw},C)\approx_{\epsilon/(1-3\epsilon/2)}\mI_{|C|}-\schur(\mI_n-\mw,C)
\]
\end{theorem}

A proof of Theorem \ref{thm:sc_approx} can be found in Appendix \ref{app:ring_powers}. The expression in Theorem \ref{thm:sc_approx} has a natural interpretation in terms of random walks. Indeed, notice that 

\[
\mI_{|C|}-\schur(\mI_n-\mw,C)=\mw_{CC}+\mw_{CF}(\mI_{|F|}-\mw_{FF})^{-1}\mw_{FC}.
\]
When $\mw$ is the transition matrix for a random walk, the right-hand side above can be interpreted as the transition matrix for the random walk induced by ``short-cutting'' walks that traverse through the set of vertices in $F$. In other words, walk behavior on $C$ remains the same (the $\mw_{CC}$ term) and walks that go from $C$ to $F$ (via $\mw_{FC}$) can instantly take arbitrary length walks in $F$ (the $(\mI_{|F|}-\mw_{FF})^{-1}$ term) before returning to $C$ (via $\mw_{CF}$). The theorem above says that spectral approximation is preserved under such ``short-cutting''.

Now we get the following corollary, which says that unit-circle approximation is preserved under powering.
\begin{corollary}
\label{cor:capprox_preserved_powering}
Let $\mw$, $\widetilde{\mw}$ be the transition matrices of digraphs $G,\widetilde{G}$. If $\widetilde{\mw}\capprox_{\epsilon}\mw$ then for all $k\in\mathbb{N}$ we have $\widetilde{\mw}^k\capprox_{\epsilon/(1-3\epsilon/2)}\mw^k$.
\end{corollary}

\begin{proof}
Fix $k\in\mathbb{N}$ and $z\in\C$ such that $|z|=1$. Let $\omega$ be such that $\omega^k=z$. 
Let $\mm=\mc_k \otimes \omega\cdot\mw$ and $\widetilde{\mm}=\mc_k \otimes \omega\cdot\widetilde{\mw}$. From Corollary \ref{cor:capprox_ring} we have that $\widetilde{\mm}\capprox_{\epsilon}\mm$. We can think of $\mm$ and $\widetilde{\mm}$ as the transition matrices of the cycle-lifted graph of graphs with transition matrices $\mw$ and $\widetilde{\mw}$. Define $F$ to be the set of vertices in the first layer of these cycle-lifted graphs. Notice that
\begin{align*}
\mI_{|F|}-\schur(\mI_n-\mm,F)&=(\omega\cdot\mw)^k=z\cdot\mw^k\\
\mI_{|F|}-\schur(\mI_n-\widetilde{\mm},F)&=(\omega\cdot\widetilde{\mw})^k=z\cdot\widetilde{\mw}^k.
\end{align*}
It follows from Theorem \ref{thm:sc_approx} that $\widetilde{\mw}^k\capprox_{\epsilon/(1-3\epsilon/2)}\mw^k$.
\end{proof}
Interestingly, in the case of undirected graphs, Corollary \ref{cor:capprox_preserved_powering} says that if $\widetilde{\mw}\approx_{\epsilon}\mw$ and $-\widetilde{\mw}\approx_{\epsilon}-\mw$ then all $k$th powers approximate one another (with small loss in approximation quality). This was not known (to the best of our knowledge) for any $k$ other than powers of 2.

\section{Derandomized Square of Regular Digraphs}
In order to achieve a space-efficient and deterministic implementation of our algorithm, we need a way to efficiently approximate high powers of regular digraphs. To do this, we use the \emph{derandomized square} graph operation of Rozenman and Vadhan \cite{RozenmanVa05}, which uses expander graphs to give sparse approximations to the graph square. We define expander graphs in terms of the measure below.

\begin{definition}[\cite{Mihail89}]
Let $G$ be a regular directed multigraph with transition matrix $\mw$. We define 
\[
\lambda(G) = \max_{v\perp\vec{\bone}}\frac{\|\mw v\|}{\|v\|}\in[0,1],
\]
where the maximum can be taken over either real or complex vectors $v$ (equivalent by applying Lemma~\ref{lem:extension_to_complex}
to the restriction of $\mw$ to the subspace orthogonal to $\vec{\bone}$).
The \emph{spectral gap} of $G$ is defined to be $\gamma(G)=1-\lambda(G)$ and when $\gamma(G)\geq \gamma$, we say that $G$ has \emph{spectral expansion} $\gamma$. When $G$ is undirected, $\lambda(G)$ equals the second largest eigenvalue of $W$ in absolute value. 
\end{definition}
It is well known that the larger $\gamma(G)$ is, the faster a random walk on $G$ converges to the stationary distribution. Families of graphs with $\lambda(G)\leq 1-\Omega(1)$ are called \emph{expanders}.

$\lambda(G)$ also relates naturally to unit-circle approximation as shown in the following lemma, which says that the smaller $\lambda(G)$, the better $G$ approximates the complete graph.

\begin{lemma}
\label{lem:expander_approx_complete}
Let $G$ be a strongly connected, regular directed multigraph on $n$ vertices with transition matrix $\mw$ and let $\mJ\in\mathbb{R}^{n\times n}$ be a matrix with $1/n$ in every entry (i.e. $\mJ$ is the transition matrix of the complete graph with a self loop on every vertex). Then $\lambda(G)\leq \lambda$ if and only if $\mw\capprox_{\lambda}\mJ$.
\end{lemma}
A proof of Lemma \ref{lem:expander_approx_complete} can be found in Appendix \ref{app:expander_approx_complete}. Before defining the derandomized square operation, we introduce two-way labelings and rotation maps.

\begin{definition}[\cite{ReingoldVaWi01, RozenmanVa05}]
A {\em two-way labeling} of a $d$-regular directed multigraph $G$ is a labeling of the edges in $G$ such that 
\begin{enumerate}
\item Every edge $(u,v)$ has two labels in $[d]$, one as an edge incident to $u$ and one as an edge incident to $v$,
\item For every vertex $v$ the labels of the edges incident to $v$ are distinct.
\end{enumerate}
\end{definition}
In a two-way labeling, each vertex $v$ has its own labeling from $1$ to $d$ for the $d$ edges leaving it and its own labeling from $1$ to $d$ for the $d$ edges entering it. Since every edge is incident to two vertices, each edge receives two labels, which may or may not be the same. It is convenient to specify a multigraph with a two-way labeling by a rotation map:
\begin{definition}[\cite{ReingoldVaWi01, ReingoldTrVa06}]
Let $G$ be a $d$-regular directed multigraph on $n$ vertices with a two-way labeling. The {\em rotation map} Rot$_{G}\colon [n]\times[d]\to[n]\times[d]$ is defined as follows: Rot$_{G}(v,i)=(w,j)$ if the $i$th edge leaving vertex $v$ leads to vertex $w$ and this edge is the $j$th edge entering $w$. 
\end{definition}
Now we can define the derandomized square. Recall that the square of a graph $G^2$ is a graph on the same vertex set whose edges correspond to all walks of length 2 in $G$. The derandomized square picks out a pseudorandom subset of the walks of length 2 by correlating the 2 steps via edges on an expander graph. 

\begin{definition}[\cite{RozenmanVa05}]
\label{def:derandsquare}
Let $G$ be a $d$-regular multigraph on $n$ vertices with a two-way labeling. Let $H$ be a $c$-regular undirected graph on $d$ vertices. The {\em derandomized square} $G\ds H$ is a $c\cdot d$-regular graph on $n$ vertices with rotation map Rot$_{G\ds H}$ defined as follows: For $v_{0}\in[n], i_0\in[d]$, and $j_0\in [c]$, we compute Rot$_{G\ds H}(v_0,(i_0,j_0))$ as
\begin{enumerate}
\item Let $(v_1,i_1)=$Rot$_{G}(v_0,i_0)$
\item Let $(i_2,j_1)=$Rot$_{H}(i_1,j_0)$
\item Let $(v_2,i_3)=$Rot$_{G}(v_1,i_2)$
\item Output $(v_2,(i_3,j_1))$
\end{enumerate}
\end{definition}

In the square of a directed graph, for each vertex $v$, there exists a complete, uni-directional bipartite graph from the in-neighbors of $v$, to it's out-neighbors. This corresponds to a directed edge for every two-step walk that has $v$ in the middle of it. A useful way to view the derandomized square is that it replaces each of these complete bipartite graphs with a uni-directional bipartite \emph{expander}. 

\begin{definition}
\label{def:bip}
Let $H=(V,E)$ be an undirected graph on $d$ vertices. We define Bip$(H)$ to be a bipartite graph with $d$ vertices on each side of the bipartition and an edge $(u,v)$ from vertex $u$ on the left to vertex $v$ on the right if and only if $(u,v)\in E$.
\end{definition}

Note that since we're working with multigraphs, the incoming or outgoing neighbors to/from a vertex may form a multi-set rather than a set due to parallel edges. So when we say that a ``copy'' of Bip$(H)$ exists from the in-neighbors of $v$ to its out-neighbors, we mean that if we were to split all of the in-neighbors of $v$ and out-neighbors of $v$ into two sets of $d$ distinct vertices, place the edges from Bip$(H)$ across the sets, and then re-merge vertices that correspond to a repeat neighbor of $v$, then a copy of that resulting graph can be found across vertex $v$ in $G\ds H$. We formalize this view of the derandomized square in the following lemma.

\begin{lemma}
\label{lem:ds_expander_decomp}
Let $G$ be a $d$-regular directed multigraph on $n$ vertices with a two-way labeling and transition matrix $\mw$. Let $H$ be a $c$-regular undirected graph on $d$ vertices with a two-way labeling and transition matrix $\mb$. Let $\mJ$ be the $d\times d$ matrix with $1/d$ in every entry and let $\widetilde{\mw}$ be the transition matrix of $G\ds H$. Define the $2d\times 2d$ matrices
\[
\mm=
\left[
\begin{array}{cc}
    \mzero & \mzero \\
    \mJ & \mzero
\end{array}
\right]
\]
and
\[
\widetilde{\mm}=
\left[
\begin{array}{cc}
    \mzero & \mzero \\
    \mb  & \mzero
\end{array}
\right].
\]
Furthermore, for each $v\in[n]$ define $\super{\mproj}{v}$, $\super{\mq}{v}$, and $\super{\mt}{v}$ as follows
\begin{align*}
\super{\mproj}{v}_{j,w}&=\begin{cases}
1 ~~~~\mathrm{if~}j\mathrm{th~edge~entering~}v\mathrm{~in~}G\mathrm{~comes~from~}w\\
0~~~~\mathrm{otherwise}
\end{cases}\\
\super{\mq}{v}_{j,w}&=\begin{cases}
1 ~~~~\mathrm{if~}j\mathrm{th~edge~leaving~}v\mathrm{~in~}G\mathrm{~goes~to~}w\\
0~~~~\mathrm{otherwise}
\end{cases}\\
\super{\mt}{v}&=\left[\begin{array}{cc}
     \mproj^{(v)}  \\
    \hline
     \mq^{(v)}
\end{array}\right].
\end{align*}
where the horizontal bar in the definition of $\super{\mt}{v}$ denotes stacking the matrix $\mproj^{(v)}$ on top of the matrix $\mq^{(v)}$. Then we have 
\[
\mw^2=\frac{1}{d}\cdot\sum_{v\in [n]}(\super{\mt}{v})^\top\mm\super{\mt}{v}
\]
and 
\[
\widetilde{\mw}=\frac{1}{d}\cdot\sum_{v\in [n]}(\super{\mt}{v})^\top\widetilde{\mm}\super{\mt}{v}
\]
\end{lemma}
\begin{proof}[Proof of Lemma \ref{lem:ds_expander_decomp}]
First we show that in $G\ds H$ there exists a ``copy'' of Bip$(H)$ from the in-neighbors of each vertex to the out-neighbors. Fix a vertex $v\in G$. Let $u_1,\ldots,u_d$ be the multi-set of incoming neighbors of $v$ and let $w_1,\ldots,w_d$ be the multiset of outgoing neighbors. Without loss of generality suppose that for each $i\in[d]$, the edge from $u_i$ to $v$ has label $i^{\mathrm{in}}$ as an edge incident to $v$ and the edge from $v$ to $w_i$ has label $i^{\mathrm{out}}$ as an edge incident to $v$ in $G$. We claim that in the graph $G\ds H$ there exists a copy of Bip$(H)$ from the $u_i$'s to the $w_i$'s such that for each $i\in[d]$, $u_i$ is the $i$th vertex on the left side of Bip$(H)$ and $w_i$ is the $i$th vertex on the right side of Bip$(H)$ (where edge multiplicities correspond to merging vertices in Bip$(H)$).

To see this, suppose $(a,b)$ is an edge in $H$ with labels $\ell_a$ and $\ell_b$, respectively. Also, let $j_a,j_b$ be the labels of edges $(u_a,v)$ and $(v,w_b)$ as edges incident to $u_a$ and $w_b$, respectively (recall that these are labelled $a$ and $b$ from $v$'s perspective). We will show that there is an edge corresponding to $(a,b)$ in $H$ from $u_a$ to $w_b$ in $G\ds H$, namely the one labeled $(j_a,\ell_a)$, which will complete the proof. We compute the rotation map Rot$_{G\ds \mathcal{H}}(u_a,(j_a,\ell_a))$:
\begin{enumerate}
    \item Rot$_G(u_a,j_a)=(v,a)$
    \item Rot$_{H}(a,\ell_a)=(b,\ell_b)$
    \item Rot$_{G}(v,b)=(w_b,j_b)$
    \item Output $(w_b,(j_b,\ell_b))$
\end{enumerate}
So edge $(j_a,\ell_a)$ leaving vertex $u_a$ indeed leads to vertex $w_b$ in $G\ds H$. 

Let 
\begin{align*}
\super{\mw}{v}&=(\super{\mt}{v})^\top\mm\super{\mt}{v}\\
\super{\widetilde{\mw}}{v}&=(\super{\mt}{v})^\top\widetilde{\mm}\super{\mt}{v}
\end{align*}
From the way we have defined $\super{\mt}{v}$, we have that $\super{\mw}{v}$ is exactly the transition matrix of the uni-directional bipartite complete graph from the in-neighbors of vertex $v$ to its out-neighbors in $G^2$ (with all vertices that are not neighbors of $v$ isolated). Likewise $\super{\widetilde{\mw}}{v}$ is the transition matrix of the uni-directional bipartite expander from the in-neighbors of vertex $v$ to its out-neighbors in $G\ds H$ (with non-neighbors of $v$ isolated). It follows from the reasoning above that 

\[
\frac{1}{d}\cdot\sum_{v\in [n]}\super{\mw}{v}=\mw^2
\]
and likewise 
\[
\frac{1}{d}\cdot\sum_{v\in [n]}\super{\widetilde{\mw}}{v}=\widetilde{\mw}.
\]
\end{proof}

Before showing that the derandomized square produces a unit-circle approximation of the true square, we need a lemma about these uni-directional bipartite expanders. The following lemma says that if we convert an expander $H$ into a uni-directional bipartite graph, then it approximates the complete uni-directional bipartite expander, where the quality of approximation depends on how good of an expander $H$ is.

\begin{lemma}
\label{lem:bipapprox}
Let $\mb\in \R^{n\times n}$ be the transition matrix of a regular multigraph $H$ with $\lambda(H)\leq \epsilon$. Let $\mJ$ be the $n\times n$ matrix with $1/n$ in every entry. Then we have that for all $x,y\in\mathbb{C}^{2n}$ 
\[
\left|x^*\left[\begin{array}{cc}
    \mzero & \mzero \\
    \mb-\mJ & \mzero
\end{array}\right]y\right|\leq \frac{\epsilon}{2}\cdot\left(\|x\|^2+\|y\|^2-2\cdot\left|x^*\left[\begin{array}{cc}
    \mzero & \mzero \\
   \mJ & \mzero
\end{array}\right]x+y^*\left[\begin{array}{cc}
    \mzero & \mzero \\
   \mJ & \mzero
\end{array}\right]y\right|\right).
\]
\end{lemma}

This lemma is stronger than saying that the two matrices unit-circle approximate each other because of the extra factor of 2 before the absolute value on the right-hand side.\footnote{An earlier version of our paper only stated unit-circle approximation
in this lemma, but that did not suffice for the proof of Theorem~\ref{thm:eul_derandsq_capprox} below, an error pointed out to us by Ori Sberlo and Dean Doron.}
Intuitively, we can gain that factor of 2 because the left-hand side only involves half of each of the vectors $x$ and $y$, while $\|x\|^2$ and $\|y\|^2$ involve the entire vectors.

\begin{proof}
For a vector $v\in\C^{2n}$, we will write $v_1$ for the first $n$ components of $v$ and $v_2$ for the last $n$ components. For every vector $v \in \mathbb{C}^n$ we can write $v=v^{||} + v^{\perp}$ where $v^{||} \in \mathrm{Span}(\bone_n)$ and $v^{\perp} \in \mathrm{Span}(\bone_n)^{\perp}$.  Further, for every vector $v \in \mathbb{C}^{2n}$ we can write $v=v^{||} + v^{\perp}$ where $v^{||} \in \mathrm{Span}(\bone_n)\times \mathrm{Span}(\bone_n)$ and $v^{\perp} \in \mathrm{Span}(\bone_n)^{\perp} \times \mathrm{Span}(\bone_n)^{\perp}$, that is, $v_1^{\perp}\perp\bone_n$ and $v_2^{\perp}\perp\bone_n$.

 Note that for all $v^{||}\in \mathrm{Span}(\bone_n)$ we have $\mJ v^{||}=v^{||}$ and for all $v^{\perp}\in \mathrm{Span}(\bone_n)^{\perp}$ we have $\mJ v^{\perp} =\vec{\mzero}$. Note also that $\mJ=\mJ^2$, which implies that for any $a,b\in\mathbb{C}^n$, we have $a^*\mJ b=a^*\mJ \mJ b=\langle{a^{||},b^{||}}\rangle$. Now we can write
\begin{align*}
    2\cdot\left|x^*\left[\begin{array}{cc}
    \mzero & \mzero \\
   \mJ & \mzero
\end{array}\right]x+y^*\left[\begin{array}{cc}
    \mzero & \mzero \\
   \mJ & \mzero
\end{array}\right]y\right| &= 2\cdot\left|x_2^*\mJ x_1 +y_2^*\mJ y_1 \right|\\
&=2\cdot\left|\langle{x_2^{||},x_1^{||}}\rangle + \langle{y_2^{||},y_1^{||}}\rangle\right|\\
&\leq 2\cdot\left|\langle{x_2^{||},x_1^{||}}\rangle\right| + 2\cdot\left|\langle{y_2^{||},y_1^{||}}\rangle\right|.
\end{align*}

By Cauchy-Schwartz, we can upper bound each term in the last line above as follows
\begin{align*}
    2\cdot\left|\langle{x_2^{||},x_1^{||}}\rangle\right|&\leq 2\cdot\|x_2^{||}\|\cdot\|x_1^{||}\|\\
    &\leq \|x_2^{||}\|^2 + \|x_1^{||}\|^2\\
    &=\|x^{||}\|^2.
\end{align*}
where the second inequality uses the fact that for all real numbers $a,b$, we have $2\cdot a\cdot b\leq a^{2}+b^{2}$. This means that to prove the lemma it suffices to show

\begin{equation}
\label{eq:goal_to_show}
\left|x^*\left[\begin{array}{cc}
    \mzero & \mzero \\
    \mb-\mJ & \mzero
\end{array}\right]y\right|\leq \frac{\epsilon}{2}\cdot\left(\|x\|^2+\|y\|^2-\|x^{||}\|^2-\|y^{||}\|^2\right)\\
=\frac{\epsilon}{2}\cdot\left(\|x^{\perp}\|^2+\|y^{\perp}\|^2\right)
\end{equation}
where the last equality follows from the fact that for all vectors $v$, we have $v^{||}\perp v^{\perp}$ and hence $\|v\|^2 = \|v^{||}+v^{\perp}\|^2 =\|v^{||}\|^2+\|v^{\perp}\|^2$.

Observe that if $v\in \mathrm{Span}(\bone_n)\times \mathrm{Span}(\bone_n)$ then 
\[
\left[
\begin{array}{cc}
    \mzero & \mzero \\
    (\mb-\mJ) & \mzero
\end{array}
\right]v =\left(\left[
\begin{array}{cc}
    \mzero & \mzero \\
    (\mb-\mJ) & \mzero
\end{array}
\right]\right)^\top v=\vec{\mzero}.
\]
 Since $\mJ v=\vec{\mzero}$ for all $v\perp\vec{\bone}$ it follows that
\begin{align*}
\left|x^*\left[
\begin{array}{cc}
    \mzero & \mzero \\
    (\mb-\mJ) & \mzero
\end{array}
\right]y\right| &=\left|(x^{\perp})^*\left[
\begin{array}{cc}
    \mzero & \mzero \\
    (\mb-\mJ) & \mzero
\end{array}
\right]y^{\perp}\right| \\
&=|(x_2^{\perp})^*\mb y_1^{\perp}|\\
&\leq \epsilon\cdot\|x_2^{\perp}\|\cdot\|y_1^{\perp}\|\\
&\leq \frac{\epsilon}{2}\cdot(\|x_2^{\perp}\|^2+\|y_1^{\perp}\|^2)\\
&\leq \frac{\epsilon}{2}\cdot(\|x^{\perp}\|^2+\|y^{\perp}\|^2),
\end{align*}
where the third line uses that $\|\mb v\|_2\leq \epsilon\cdot\|v\|_2$ for all $v\perp\vec{\bone}$, and the fourth line uses the AM-GM inequality. This establishes Inequality \ref{eq:goal_to_show} and the lemma.

\end{proof}

Now we show that the derandomized square of a regular digraph yields a unit-circle spectral approximation to the true square. 

\begin{theorem}
\label{thm:eul_derandsq_capprox}
 Let $G=(V,E)$ be a $d$-regular directed multigraph with random walk matrix $\mw$. Let $H$ be a $c$-regular expander with $\lambda(H)\leq \epsilon$ and let $\widetilde{\mw}$ be the random walk matrix of $G\ds H$. Then
\[
\widetilde{\mw}\capprox_{2\cdot\epsilon}\mw^2.
\]
\end{theorem}
\begin{proof}
Define $\mb, \mm,\widetilde{\mm},\super{\mt}{v},\super{\mw}{v},\super{\widetilde{\mw}}{v}$ as in Lemma \ref{lem:ds_expander_decomp}.  By Lemma \ref{lem:ds_expander_decomp} we have 
\[
\frac{1}{d}\cdot\sum_{v\in [n]}\super{\mw}{v}=\mw^2.
\]

Also, by Lemma \ref{lem:bipapprox} we have that for all $x,y\in\mathbb{C}^n$
\[
\left|x^*\left[\begin{array}{cc}
    \mzero & \mzero \\
    \mb-\mJ & \mzero
\end{array}\right]y\right|\leq \frac{\epsilon}{2}\cdot\left(\|x\|^2+\|y\|^2-2\cdot\left|x^*\left[\begin{array}{cc}
    \mzero & \mzero \\
   \mJ & \mzero
\end{array}\right]x+y^*\left[\begin{array}{cc}
    \mzero & \mzero \\
   \mJ & \mzero
\end{array}\right]y\right|\right).
\]
which means that we can consider $x'=\super{\mt}{v}x$ and $y'=\super{\mt}{v}y$ to conclude
\[
\left|x^{*}(\super{\widetilde{\mw}}{v}-\super{\mw}{v})y\right|\leq \frac{\epsilon}{2}\cdot\left(x^*(\super{\mt}{v})^\top \super{\mt}{v}x + y^*(\super{\mt}{v})^\top \super{\mt}{v}y-2\cdot\left|x^*\super{\mw}{v}x+y^*\super{\mw}{v}y\right|\right)
\]
Summing both sides of the inequality over all vertices $v\in[n]$, gives 
\begin{eqnarray}
  \lefteqn{ \sum_{v\in[n]}\left|x^{*}(\super{\widetilde{\mw}}{v}-\super{\mw}{v})y\right|}\nonumber\\&\leq \frac{\epsilon}{2}\cdot\left(\sum_{v\in[n]}\left(x^* (\super{\mt}{v})^\top \super{\mt}{v}x + y^*(\super{\mt}{v})^\top \super{\mt}{v}y\right)- 2\cdot\left|x^*\super{\mw}{v}x+y^*\super{\mw}{v}y\right|\right)
  \label{eq:littlebipV}
\end{eqnarray}
Applying the triangle inequality to the left-hand side of Inequality \ref{eq:littlebipV} gives
\begin{align*}
     \sum_{v\in[n]}\left|x^{*}(\super{\widetilde{\mw}}{v}-\super{\mw}{v})y\right|&\geq \left|\sum_{v\in[n]}x^{*}(\super{\widetilde{\mw}}{v}-\super{\mw}{v})y\right|\\
     &=d\cdot \left|x^*(\widetilde{\mw}-\mw^2)y\right|.
\end{align*}
Similarly, we can apply the triangle inequality to the sum of complex magnitudes on the right-hand side of Inequality \ref{eq:littlebipV} to conclude
\[
d\cdot \left|x^*(\widetilde{\mw}-\mw^2)y\right|\leq \frac{\epsilon}{2}\cdot\left(\sum_{v\in[n]}\left(x^* (\super{\mt}{v})^\top \super{\mt}{v}x + y^*(\super{\mt}{v})^\top \super{\mt}{v}y\right)-2\cdot d\cdot \left|x^*\mw^2x+y^*\mw^2y\right|\right).
\]

Finally, we show that 
\begin{equation}
\label{eq:transformation_sum}
\sum_{v\in[n]}(\super{\mt}{v})^\top \super{\mt}{v}=2\cdot d\cdot \mI_n.
\end{equation}
This will complete the proof, because our inequality will become 
\[
d\cdot \left|x^*(\widetilde{\mw}-\mw^2)y\right|\leq 2\cdot d\cdot \frac{\epsilon}{2}\cdot\left(\|x\|^2+\|y\|^2- \left|x^*\mw^2x+y^*\mw^2y\right|\right),
\]
which says that $\widetilde{\mw}\capprox_{2\cdot\epsilon}\mw^2$.
Now we prove Equation \ref{eq:transformation_sum}. Fix $v,i,j\in[n]$. We can write 
\begin{align*}
   \left(\left(\super{\mt}{v}\right)^\top \super{\mt}{v}\right)_{ij}&=\sum_{k\in[2d]}(\super{\mt}{v})^{\top}_{ik}\super{\mt}{v}_{kj}\\
   &=\sum_{k\in[2d]}\super{\mt}{v}_{ki}\super{\mt}{v}_{kj}.
\end{align*}
From the definition of $\super{\mt}{v}$, there is exactly one 1 in every row of the matrix, so when $i\neq j$, the above sum is 0. When $i=j$, the sum contributes a 1 for each edge connecting $v$ and $i$ in the graph (in either direction). Therefore $\sum_{v\in[n]}(\super{\mt}{v})^\top \super{\mt}{v}$ is always 0 off the diagonal and $2\cdot d$ on the diagonal because every vertex has exactly $2\cdot d$ edges incident to it. This confirms Equation \ref{eq:transformation_sum} and completes the proof.  
\end{proof}

\section{Approximate Pseudoinverse for Cycle-Lifted Graphs}\label{sec:squaring_solver}

Let $\mI-\mw$ be the random-walk Laplacian of a strongly connected, aperiodic, regular digraph $G$.
Our goal is to compute an accurate approximation of $(\mI-\mw)^+$.
To do this we consider the Laplacian of a cycle-lifted graph  $\mL=\mI_{2^k n}-\mc_{2^k} \otimes \mw$  for some positive integer $k$, and show how to compute an accurate approximation of  $\mL^+$, namely $\widetilde{\mL^+}$. Then we show that under some conditions, an $n\times n$ projection of $\widetilde{\mL^+}$ (specifically, $(\vec{\bone}_{2^k}\otimes \mI_n)^\top \widetilde{\mL^+} (\vec{\bone}_{2^k}\otimes \mI_n)$) gives an accurate approximation for $(\mI-\mw)^+$ (see Lemma~\ref{lem:pinv_reduce}).

To estimate $\mL^+$ we first show how to obtain a weak approximation to it. Then we show how to get an accurate approximation using Richardson iteration (see Lemma~\ref{lem:richardson-poly}). The following is the main theorem we prove in this section. In this theorem, we only give sufficient conditions for having an approximate pseudo-inverse of the cycle-lifted graphs, and discuss an actual space-efficient algorithm for computing such a matrix in Section~\ref{sec:small-space-solver}. 

\begin{theorem}\label{thm:app_pinv_ring}
Let $\mw$ be the transition matrix of a strongly connected regular digraph with $n$ vertices, $\epsilon \in (0, 1/2)$ and suppose we have a sequence of matrices $\mw = \mw_0, \ldots, \mw_{k}$, such that 
$$
\mw_{i+1} \capprox_{\epsilon/k} \mw_i^2 \quad \forall 0 < i < k
$$
and each $\mw_i$ is a transition matrix of a strongly connected regular digraph. We use $\mw_i$'s to define a sequence of $2^k n$ by $2^k n$ matrices $\mL^{(i)}$ as in Equation~\eqref{eq:L_i_fac}. Then for $\mL = \mI_{2^kn}-\mc_{2^k} \otimes \mw$
and $\widetilde{\mL} \eqdef \mL^{(k)}$, there exists a PSD matrix $\mf$ such that $\|\mI_{2^kn} - \widetilde{\mL}^+ \mL\|_\mf \leq O(k\epsilon)$ and $\mU_\mL/O(k) \preceq \mf \preceq O(2^{2k} n^2 k^5) \mU_\mL$.
\end{theorem}

In the above theorem $\widetilde{\mL}$ is defined in a way so that it has a nice LU factorization. This lets us efficiently compute $\widetilde{\mL}^+$. Below we describe how we use Theorem \ref{thm:app_pinv_ring} in our solver and then we prove the theorem in Section \ref{sect:solver}.  We give the  characterization of $\widetilde{\mL}^+$ in Proposition~\ref{prop:tilde_L_pinv_fac}. In Section~\ref{sec:small-space-solver}, for $k=O(\log n)$, we show  how to space efficiently generate the $\mw_i$'s and compute $\widetilde{\mL}^+$.
 
The following lemma shows how we can obtain an accurate solver by boosting the precision of an approximate pseudo-inverse through the well-known method of preconditioned Richardson iteration \cite{CKPPRSV17,PS13}.

\begin{lemma}\label{lem:richardson-poly}
Given matrices $\ma,\mb,\mf \in \R^{n \times n}$, such that $\mf$ is PSD, and $\|\mI - \mb \ma\|_{\mf} \leq \alpha$ for some constant $\alpha > 0$. Let $\mproj_m = \sum_{i=0}^m (\mI - \mb \ma)^i \mb$. Then 
$$
\|\mI - \mproj_m \ma\|_\mf \leq \alpha^{m+1}
$$
\end{lemma}
\begin{proof}
We have $\mI - \mproj_m \ma = (\mI - \mb \ma)^{m+1}$, and then the proof follows by the submultiplicity of $\|\cdot\|_\mf$.
\end{proof}

Since we can obtain a reasonably good approximate pseudo-inverse for $\mL$ via Theorem~\ref{thm:app_pinv_ring} and boost the quality of that approximation with Lemma~\ref{lem:richardson-poly}, we can ultimately get a very accurate approximate pseudo-inverse. This is stated rigorously in the following corollary.

\begin{corollary}
Given a transition matrix $\mw$ of a regular digraph with $n$ vertices, and $\delta \in (0,1/2)$. Let $\mL = \mI_{2^kn}-\mc_{2^k} \otimes \mw$, and let $\widetilde{\mL}^+$ be the approximate pseudo-inverse obtained from Theorem~\ref{thm:app_pinv_ring} by setting $\epsilon =\frac{1}{c k}$ for a large enough constant $c$. For $m = O\left(k + \log n + \log(\frac{1}{\delta})\right)$, and $\mproj_m = \sum_{i=0}^m (\mI - \widetilde{\mL}^+ \mL)^i \widetilde{\mL}^+$, we have
$$
\|\mI - \mproj_{m} \mL\|_{\mU_\mL} \leq \delta
$$
\end{corollary}
\begin{proof}
Note that with the choice of $\epsilon$, we have $\|\mI - \widetilde{\mL}^+ \mL\|_\mf \leq \frac{1}{2}$. Therefore, by Lemma~\ref{lem:richardson-poly} we get $\|\mI - \mproj_m \mL\|_\mf \leq \frac{\delta}{poly(2^k, n)}$. Since $\mU_\mL/O(k) \preceq \mf \preceq O(2^{2k} n^2 k^5) \mU_\mL$, this implies $\|\mI - \mproj_m \mL\|_{\mU_\mL} \leq \delta$.
\end{proof}

In the lemma below, we show how to obtain an approximate pseudo-inverse of the original Laplacian system, given an approximate pseudo-inverse for the cycle-lifted graph.  

\begin{lemma}\label{lem:pinv_reduce}
    Given a matrix $\mw \in \R_{\geq 0}^{n\times n}$ with $\|\mw\| \leq 1$ and an arbitrary integer $\ell > 0$, let $\mL_{W} = \mI_n - \mw$, and $\mL_{C} = \mI_{\ell n} - \mc_\ell \otimes \mw$. Let $\mb_C \in \R^{\ell n\times \ell n}$ such that 
    $$
        \|\mI_{\ell n} - \mb_C \mL_{C}\|_{\mU_{\mL_{C}}} \leq \delta.
    $$
    Let $\mb_{W} = \frac{1}{\ell}(\vec{\bone}_\ell \otimes \mI_n)^\top \mb_{C} (\vec{\bone}_\ell \otimes \mI_n)$; then
    $$
        \|\mI_n - \mb_{W} \mL_{W}\|_{\mU_{\mL_{W}}} \leq \delta.
    $$
\end{lemma}

\begin{proof}
$\|\mI_n - \mb_{W} \mL_{W}\|_{\mU_{\mL_{W}}} \leq \delta$ is equivalent to $$
 (\mI_n - \mb_W \mL_W)^\top \mU_{\mL_W} (\mI_n - \mb_W \mL_W) \preceq \delta^2 \mU_{\mL_W},
$$
For the RHS we have $\delta^2 \cdot \frac{1}{\ell} (\vec{\bone}_\ell\otimes \mI_n)^\top \mU_{\mL_{C}} (\vec{\bone}_\ell \otimes \mI_n) = \delta^2 \mU_{\mL_W}$. 
For the LHS, let $\Pi =\frac{1}{\ell} (\vec{\bone}_\ell \otimes \mI_n)(\vec{\bone}_\ell \otimes \mI_n)^\top = (\frac{\vec{\bone}_\ell \vec{\bone}_\ell^\top}{\ell})\otimes \mI_n$, then we get
$$
\frac{1}{\ell} (\vec{\bone}_\ell \otimes \mI_n)^\top (\mI_{\ell n} - \mb_{C}\Pi\mL_{C})^\top \Pi \mU_{\mL_{C}} \Pi(\mI_{\ell n} - \mb_{C}\Pi\mL_{C}) (\vec{\bone}_\ell \otimes \mI) = (\mI_n - \mb_W \mL_W)^\top \mU_{\mL_W} (\mI_n - \mb_W \mL_W).
$$
Thus it is sufficient to show,
$$
\frac{1}{\ell} (\vec{\bone}_\ell \otimes \mI_n)^\top (\mI_{\ell n} - \mb_{C}\Pi\mL_{C})^\top \Pi \mU_{\mL_{C}} \Pi(\mI_{\ell n} - \mb_{C}\Pi\mL_{C}) (\vec{\bone}_\ell \otimes \mI)
\preceq \delta^2 \cdot \frac{1}{\ell} (\vec{\bone}_\ell\otimes \mI_n)^\top \mU_{\mL_{C}} (\vec{\bone}_\ell\otimes \mI_n).
$$
Note that $\Pi$ is an orthogonal projection, and $\Pi \mL_C = \mL_C \Pi$. Thus by the lemma assumption and Lemma~\ref{lem:psd_ortho_proj}, we have,
$$
(\mI_{\ell n} - \mb_{C}\mL_{C})^\top \Pi \mU_{\mL_{C}} (\mI_{\ell n} - \mb_{C}\mL_{C})
\preceq (\mI_{\ell n} - \mb_{C}\mL_{C})^\top \mU_{\mL_{C}} (\mI_{\ell n} - \mb_{C}\mL_{C})
\preceq \delta^2 \mU_{\mL_{C}}.
$$
Since $\Pi^2 = \Pi$, and $\Pi$ commutes with $\mL_C$ and $\mU_{\mL_C}$, we get
$$
(\mI_{\ell n} - \mb_{C}\Pi\mL_{C})^\top \Pi \mU_{\mL_{C}} \Pi(\mI_{\ell n} - \mb_{C}\Pi\mL_{C})
\preceq \delta^2 \mU_{\mL_{C}}.
$$
Now by Proposition~\ref{prop:psd_mult_sides}, we get
\begin{equation}
\frac{1}{\ell} (\vec{\bone}_\ell \otimes \mI_n)^\top (\mI_{\ell n} - \mb_{C}\Pi\mL_{C})^\top \Pi \mU_{\mL_{C}} \Pi(\mI_{\ell n} - \mb_{C}\Pi\mL_{C}) (\vec{\bone}_\ell \otimes \mI)
\preceq \delta^2 \cdot \frac{1}{\ell} (\vec{\bone}_\ell\otimes \mI_n)^\top \mU_{\mL_{C}} (\vec{\bone}_\ell\otimes \mI_n).
\end{equation}
which completes the proof.
\end{proof}

\subsection{Approximate Pseudoinverse of Cycle-Lifted Graphs}
\label{sect:solver}

To get an approximate pseudo-inverse of $\mL$ we first compute an approximate LU factorization of it. For any matrix $\mm\in\R^{n\times n}$, if $F,C$ partition $[n]$ and $\mm_{FF}$ is invertible we can write $\mm$ as the product of a lower triangular matrix, a block diagonal matrix, and an upper triangular matrix:
\begin{align}\label{eq:schur_decomp}
\mm
&=\left[
\begin{array}{cc}
    \mm_{FF} & \mm_{FC} \\
    \mm_{CF} & \mm_{CC}
\end{array}
\right]\nonumber\\
&=
\left[
\begin{array}{cc}
    \mI& 0 \\
    \mm_{CF}\mm^{-1}_{FF} & \mI
\end{array}
\right]\left[
\begin{array}{cc}
    \mm_{FF}& 0 \\
    0 & \mm_{CC}-\mm_{CF}\mm_{FF}^{-1}\mm_{FC}
\end{array}
\right]\left[
\begin{array}{cc}
    \mI & \mm^{-1}_{FF}\mm_{FC} \\
    0 & \mI
\end{array}
\right]
\end{align}
Note that 
$\mm_{CC}-\mm_{CF}\mm_{FF}^{-1}\mm_{FC}$ is the Schur complement of $\mm$ onto the set $C$. For an invertible matrix $\mm$, the above factorization gives a formula to compute $\mm^{-1}$:
\[
\mm^{-1}=\left[
\begin{array}{cc}
    \mI & -\mm^{-1}_{FF}\mm_{FC} \\
    0 & \mI
\end{array}
\right]\left[
\begin{array}{cc}
    \mm_{FF}^{-1}& 0 \\
    0 & \schur(\mm, C)^{-1}
\end{array}
\right]
\left[
\begin{array}{cc}
    \mI& 0 \\
    -\mm_{CF}\mm^{-1}_{FF} & \mI
\end{array}
\right]
\]
The lower and upper triangular parts are easy to invert. Therefore, the above formula reduces inverting $\mm$ to inverting $\mm_{FF}$ and $\schur(\mm, C)$. This approach has been used in many recent time efficient algorithms for solving both symmetric and asymmetric diagonally dominant systems \cite{CKKPPRS18,kyng2016sparsified}. We use the same approach of LU factorization to compute an approximate pseudo-inverse of $\mL$. 

Without loss of generality we use the following ordering of rows and columns for $\mc_{2^k}$. Let $\mc_1 = [1]$, and for $k>0$ let

$$
\mc_{2^k} = \begin{bmatrix}
0 & \mc_{2^{k-1}}\\
\mI_{2^{k-1}} & 0
\end{bmatrix}.
$$
Let $H$ be the set of coordinates $\{2^{k-1}n+1, 2^{k-1}n+2, \ldots, 2^{k}n\}$, so that the cycle alternates between $H$ and $H^c$. We get the following nice characterization for the Schur complement of $\mI_{2^k n}-\mc_{2^k} \otimes \mw$ on to set $H$.
$$
\schur{(\mI_{2^k n}-\mc_{2^k} \otimes \mw, H)} = \mI_{2^{k-1}n} - (\mI_{2^{k-1}}\otimes \mw) (\mc_{2^{k-1}}\otimes \mw) = \mI_{2^{k-1}n} -  (\mc_{2^{k-1}}\otimes \mw^2).
$$
Using the above relation, and \eqref{eq:schur_decomp} we can get the following factorization for $\mL$.
\begin{equation}
\label{eq:lu_recursion}
\mL=
\begin{bmatrix}
    \mI_{2^{k-1} n}& 0 \\
    -\mI_{2^{k-1}} \otimes \mw & \mI_{2^{k-1} n}
\end{bmatrix}
\left[
\begin{array}{cc}
    \mI_{2^{k-1} n}& 0 \\
    0 & \mI_{2^{k-1} n}- \mc_{2^{k-1}} \otimes \mw^2
\end{array}
\right]\left[
\begin{array}{cc}
    \mI_{2^{k-1} n}& -\mc_{2^{k-1}} \otimes \mw \\
    0 & \mI_{2^{k-1}  n}
\end{array}
\right]
\end{equation}

Note that the lower right block of the middle matrix is the Laplacian of cycle-lifted graph with cycle length $2^{k-1}$ and adjacency matrix $\mw^2$.
Applying this recursion one more time by short-cutting every other layer in the smaller cycle-lifted graph leads to a cycle-lifted graph of length $2^{k-2}$ and adjacency matrix $\mw^4$.
To get an approximate LU factorization, we repeatedly apply this recurrence and replace all occurrences of powers of $\mw$ with appropriate approximations of them (see Theorem~\ref{thm:app_pinv_ring}).
For $j\in\{1,\ldots, k\}$ let $\mw_j$'s be defined as in Theorem~\ref{thm:app_pinv_ring}.
With $\mL^{(0)} = \mL$, we denote by $\mL^{(i)}$ the matrix obtained after applying $i$ steps of recursion.  For example, we have
$$
\mL^{(1)} = \begin{bmatrix}
    \mI_{2^{k-1} n}& 0 \\
    -\mI_{2^{k-1}} \otimes \mw & \mI_{2^{k-1} n}
\end{bmatrix}
\left[
\begin{array}{cc}
    \mI_{2^{k-1} n}& 0 \\
    0 & \mI_{2^{k-1} n}- \mc_{2^{k-1}} \otimes \mw_1
\end{array}
\right]\left[
\begin{array}{cc}
    \mI_{2^{k-1} n}& -\mc_{2^{k-1}} \otimes \mw \\
    0 & \mI_{2^{k-1}  n}
\end{array}
\right].
$$
More generally, for $1 \leq i \leq k$ we get,
\begin{equation}\label{eq:L_i_fac}
\mL^{(i)} = \mx_1 \cdots \mx_{i}
\begin{bmatrix}
    \mI_{(2^{k}-2^{k-i})n} & 0 \\
    0 & \mI_{2^{k-i}n} - \mc_{2^{k-i}} \otimes \mw_{i}
\end{bmatrix}
\my_i \cdots \my_1
\end{equation}
where $\mx_j$'s and $\my_j$'s are lower and upper triangular matrices, respectively. Specifically, we have
\begin{equation}
\mx_j = \begin{bmatrix}
     \mI_{(2^{k}-2^{k-j+1})n} & 0 & 0 \\
     0 & \mI_{2^{k-j}n} & 0 \\ 
    0 & -\mI_{2^{k-j}} \otimes \mw_{j-1} & \mI_{2^{k-j}n}
\end{bmatrix},
\my_j = \begin{bmatrix}
     \mI_{(2^{k}-2^{k-j+1})n} & 0 & 0 \\
     0 & \mI_{2^{k-j}n} & -\mc_{2^{k-j}} \otimes \mw_{j-1} \\
    0& 0 & \mI_{2^{k-j}n}
\end{bmatrix}.
\end{equation}

Given the above factorization for $\mL^{(k)}$ it is easy to get one for ${\mL^{(k)}}^+$.
We use the following lemma from \cite{CKKPPRS18} to give a factorization of ${\mL^{(k)}}^{+}$.

\begin{lemma}[CKKPPRS18 Lemma C.3]
 \label{lem:pinv_of_prod}
  Consider real matrices
  $\ma \in \R^{m \times m}$,
  $\mb \in \R^{m \times n}$,
  and $ \mc \in \R^{n \times n}$,
  where $\ma$ and $\mc$ are invertible.
  Let $\mm = \ma \mb \mc$, 
  then $\mm^{+} = \mproj_\mm C^{-1} B^{+} A^{-1} \mproj_{\mm^{\top}}$, where $\mproj_\mm$ defines the orthogonal projection into image of $\mm$.
 \end{lemma}
 
 The following proposition gives the characterization of ${\mL^{(k)}}^+$ which we later use for computing ${\mL^{(k)}}^+$ in small space. Note that inverting the $\mx_j$'s, and $\my_j$'s is easy as they are lower and upper triangular matrices.
\begin{proposition}\label{prop:tilde_L_pinv_fac}
Given $\mL^{(k)}$ defined by Equation~\eqref{eq:L_i_fac}, we have
\begin{equation}\label{eq:l_k_pinv}
{\mL^{(k)}}^{+} = 
\mproj_{\mL^{(k)}} \my_1^{-1} \cdots \my_k^{-1}
\begin{bmatrix}
    \mI_{(2^{k}-1)n} & 0 \\
    0 & \mI_{n} - \mw_{k}
\end{bmatrix}^{+}
\mx_k^{-1} \cdots \mx_{1}^{-1}
\mproj_{{\mL^{(k)}}^\top}
\end{equation}
where 
$$
\mx_j^{-1} = \begin{bmatrix}
     \mI_{(2^{k}-2^{k-j+1})n} & 0 & 0 \\
     0 & \mI_{2^{k-j}n} & 0 \\ 
    0 & \mI_{2^{k-j}} \otimes \mw_{j-1} & \mI_{2^{k-j}n}
\end{bmatrix},
\my_j^{-1} = \begin{bmatrix}
     \mI_{(2^{k}-2^{k-j+1})n} & 0 & 0 \\
     0 & \mI_{2^{k-j}n} & \mc_{2^{k-j}} \otimes \mw_{j-1} \\
    0& 0 & \mI_{2^{k-j}n}
\end{bmatrix}.
$$
\end{proposition}

To prove Theorem~\ref{thm:app_pinv_ring}, we first show that $\mL^{(k)}$ is a good approximation of $\mL$. For that, we build a PSD matrix $\mf$ such that $\mL$ and $\mL^{(k)}$ approximate each other with respect to the norm defined by $\mf$.
For $0 \leq j \leq k$, let 
\begin{equation}
\super{\ms}{j}=
\begin{bmatrix}\label{eq:s_i_def}
    0& 0 \\
    0 & \mI_{2^{k-j}\cdot n}-\mc_{2^{k-j}}\otimes \mw_j 
\end{bmatrix}, 
\end{equation}
where the zeros are used as padding to make the dimension of $\ms^{(j)}$'s $2^k n \times 2^k n$. Note that $\ms^{(j)}$'s correspond to the approximate Schur complement blocks appear in our recursive algorithm (see Equation~\ref{eq:L_i_fac}). Using Lemma 2.3 from \cite{CKKPPRS18}, we show that the average of $\mU_{\ms^{(i)}}$'s would be a good choice for $\mf$. We defer the statement of the lemma from \cite{CKKPPRS18} and the proof of Lemma~\ref{lem:CumulativeErrorBlock} to the appendix (see Appendix~\ref{sec:solver_app}).

\begin{lemma}
	\label{lem:CumulativeErrorBlock}
Let $\ms^{(0)}, \ms^{(1)}, \ldots, \ms^{(k)}$, and $\mL^{(0)}, \mL^{(1)}, \ldots, \mL^{(k)}$ be defined as Equations~\eqref{eq:s_i_def} and \eqref{eq:L_i_fac} respectively. Then for,
$$
\mf =\frac{2}{k} \sum_{i=0}^{k} \mU_{\ms^{(i)}}
$$
we have:
\begin{enumerate}
\item  \label{itm:spectral}
for each $0 \leq i \leq k$,
\[
\left\|\mf^{+/2}
	\left(\mL- \mL^{\left(i\right)}\right)
    \mf^{+/2}\right\|_2
\leq \epsilon,
\]
\item  \label{itm:precond}
The final matrix $\mL^{(k)}$ satisfies
\[
\mL^{\left(k \right) \top}
\mf^{+}
\mL^{\left(k \right)}
\succeq \frac{1}{40 k^2} \cdot \mf.
\]
\end{enumerate}
\end{lemma}
Informally, Item~\ref{itm:spectral} states that all of the $\mL^{(i)}$'s and in particular $\mL^{(k)}$ are good approximations of $\mL$ with respect to $\mf$, or equivalently $\mf^{+/2}\mL^{(i)}\mf^{+/2}$ is a good approximation of $\mf^{+/2} \mL \mf^{+/2}$ in spectral norm.  Item~\ref{itm:precond}
states that $\mL^{(k)}$ is not too small with respect to $\mf$,
or equivalently that the spectral norm of  $(\mf^{+/2}\mL^{(k)}\mf^{+/2})^+ = \mf^{1/2} (\mL^{(k)})^+ \mf^{1/2}$ is not too large.
Together they are used to
show that for $\widetilde{\mL} = \mL^{(k)}$,
$\widetilde{\mL}^+$ is a good preconditioner for $\mL$ in the $\mf$ norm, i.e. $\|\mI_{2^kn} - \widetilde{\mL}^+ \mL\|_\mf$ is small (see Lemma~\ref{lem:lem2.6CKK}).
Now we are ready to give a proof of Theorem~\ref{thm:app_pinv_ring}.
\begin{proof}[Proof of Theorem~\ref{thm:app_pinv_ring}]
By Lemma~\ref{lem:CumulativeErrorBlock}, there exists matrix $\mf$ such that,
$$
\|\mf^{+/2} (\mL - \widetilde{\mL}) \mf^{+/2}\| \leq \epsilon
$$
and
$$
\widetilde{\mL}^\top \mf^{+} \widetilde{\mL} \succeq \frac{1}{40k^2} \mf.
$$
Therefore, by Lemma~\ref{lem:lem2.6CKK}, we get $\|\mI_{2^kn} - \widetilde{\mL}^+ \mL\|_\mf \leq \sqrt{40}k\epsilon$. Note that $\mf \succeq \frac{\mU_\mL}{k}$ by construction. To get the upper-bound on $\mf$, we have
$$
\mf \preceq O(k^2) \widetilde{\mL}^\top \mf^+ \widetilde{\mL} \preceq O(k^4) \mL^\top \mf^+ \mL \preceq O(k^5) \mL^\top \mU_\mL^+ \mL 
$$
where the second inequality holds by Lemma~\ref{lem:lemB.3CKK}, and the last inequality comes from $\mf\succeq \frac{\mU_\mL}{k}$. By applying Lemma~\ref{lem:lem13CKP} we get,
$$
\mf \preceq O(k^5) \mL^\top \mU_\mL^+ \mL \preceq O(k^5 2^{2k} n^2)\mU_\mL.
$$
\end{proof}

\section{Space-Efficient Eulerian Laplacian Solver}\label{sec:small-space-solver}
In this section we show that the LU factorization approach to computing the pseudoinverse of an Eulerian Laplacian described in Section \ref{sec:squaring_solver} can be implemented space-efficiently. In particular, we prove the following theorem.
\begin{theorem}
\label{thm:space_main}
There is a deterministic algorithm that, given $\epsilon\in(0,1)$ and an Eulerian digraph $G$ with random-walk Laplacian $\mL=\mI-\mw$, computes a matrix $\widetilde{\mL}^{+}$ whose entries differ from $\mL^+$ by at most $\epsilon$. The algorithm uses space $O(\log N\cdot\log\log(N/\epsilon))$ where $N$ is the bitlength of the input.
\end{theorem}

\subsection{Model of Space-Bounded Computation}
We use a standard model of space bounded computation. The machine has a read-only input tape, a constant number of  read/write work tapes, and a write-only output tape. We say the machine runs in space $s$ if throughout the computation, it only uses $s$ total tape cells on the work tapes. The machine may write outputs to the output tape that are larger than $s$ (in fact as large as $2^{O(s)}$) but the output tape is write-only. We use the following fact about the composition of space-bounded algorithms.
\begin{proposition}[Composition of Space-Bounded Algorithms]
\label{prop:composition}
Let $f_1$ and $f_2$ be functions that can be computed in space $s_1(n),s_2(n)\geq \log n$, respectively, and $f_1$ has output of length $\ell_1(n)$ on inputs of size $n$. Then $f_2\circ f_1$ can be computed in space 
\[
O(s_2(\ell_1(n)) + s_1(n)).
\]
\end{proposition}
\subsection{Reduction to Regular, Aperiodic, Strongly Connected Case}
\label{sect:reduction_to_reg}

It is convenient to work with regular, aperiodic, strongly connected digraphs. In this section, we show that this is without loss of generality. 

We can find the strongly connected components of an Eulerian graph by ignoring the directionality on the edges and running Reingold's algorithm for \USTConn~ on each pair of vertices using $O(\log N)$ space. Then we can solve systems in a disconnected graph by solving systems on each of its strongly connected components separately. So without loss of generality, our graph is strongly connected. 

Given an Eulerian digraph $G$ with maximum degree $\dmax$ and 
Laplacian $\md-\ma$, we can create a $d$-regular, aperiodic graph of any degree $d>\dmax$ by adding $d-\mathrm{degree}(v)$ self loops to each vertex $v$. Notice that self loops do not change the Laplacian $\md-\ma$. 
Our solver works with random-walk Laplacians and so is able to approximate $((\md-\ma)/d)^+$ and now we want to show how to use this to approximate $((\md-\ma)\md^{-1})^+$. 

Let $\mproj$ be the orthogonal projection onto the column space of $(\md-\ma)\md^{-1}$ and let $\mq$ be the orthogonal projection onto the column space of $\md^{-1}(\md-\ma^\top)$. By Lemma \ref{lem:pinv_of_prod} we have
\[
((\md-\ma)\md^{-1})^+=\mproj \md(\md-\ma)^+\mq.
\]
Given an $\epsilon$-approximation to $((\md-\ma)/d)^+=d\cdot (\md-\ma)^{+}$, we can divide this by $d$ to get an $\epsilon$-approximation to $(\md-\ma)^{+}$. Plugging this in to the formula above and making $\epsilon$ sufficiently small says that we just need to compute the matrices $\mproj$ and $\mq$ in order to get an approximation to $((\md-\ma)\md^{-1})^+$. If $s$ is the stationary distribution of $G$, then $\mproj$ is the matrix $\mI-ss^{\top}/\|s\|^2$. Since $G$ is Eulerian, the kernel of $\md^{-1}(\md-\ma^\top)$ is simply the all ones vector and so $\mq=\mI-\vec{\bone}\vec{\bone}^\top/n$. The stationary distribution of a strongly connected Eulerian graph is proprotional to its vertex degrees and so is easy to compute in logspace. Thus $\mproj$ and $\mq$ can both be computed in deterministic logspace and we can approximate $((\md-\ma)\md^{-1})^+$ from $(\md-\ma)^+$.  
\subsection{Proof of Theorem \ref{thm:space_main}}
To prove Theorem \ref{thm:space_main}, we follow the LU factorization approach discussed in Section \ref{sect:solver}. Just as in Section \ref{sect:solver}, we first show that we can compute a weak approximation to the pseudoinverse space-efficiently and then argue that we can afford to reduce the error using preconditioned Richardson iteration. The analysis in this section is similar to the space complexity analysis from \cite{MRSV17}.

Proposition \ref{prop:tilde_L_pinv_fac} says that in order to compute a weak approximation to a pseudoinverse of a random-walk Laplacian $\mI-\mw$ it suffices to compute a particular polynomial $p$ in the matrices $\mI,\mw$, and approximations to $\mw^{2^{k}}$ for $k=\{1,\ldots,O(\log n)\}$. To compute the approximations to $\mw^{2^{k}}$, we will show that the $k$th derandomized square can be computed in space $O(\log n+k\cdot \log c)$, where $c$ is the degree of the expanders used. To bottom out the recursion, we replace the random-walk Laplacian of the final derandomized power with $(\mI-\mJ)=(\mI-\vec{\bone}\vec{\bone}^\top/n)$, the random-walk Laplacian of the complete graph (with self loops), which is its own pseudoinverse ($(\mI-\mJ)^+=\mI-\mJ)$. Then to compute $p$, we note that $p$ has degree $O(\log n)$ and the product of $O(\log n)$ matrices can be computed in space $O(\log n\cdot \log\log n)$. Finally, to boost the weak approximation, we use preconditioned Richardson iteration (Lemma \ref{lem:richardson-poly}).

The discussion above gives the following informal outline for proving Theorem \ref{thm:space_main}:
\begin{enumerate}
    \item Show that the transition matrix of the $k$th derandomized square using $c$-regular expanders can be computed in space $O(\log n+k\cdot \log c)$
    \item Show that the complete graph is a good approximation of the $k$th derandomized power for sufficiently large $k$
    \item Show that $k$ arbitrary $n\times n$ matrices can be multiplied in space $O(\log n\cdot \log k)$
    \item Show that preconditioned Richardson can be used to boost a weak approximation to a pseudoinverse to an $\epsilon$ approximation in space $O(\log n\cdot\log\log(n/\epsilon))$.
    \item Show that an extremely good approximation to a pseudoinverse in a spectral sense is also a good approximation in an entrywise sense.
\end{enumerate}

To show that derandomized powers can be computed in small space, we follow the proof techniques of \cite{Reingold08,RozenmanVa05,vadhan2012pseudorandomness}. First we note that neighbors in the sequence of expanders we use for the iterated derandomized square can be explicitly computed space-efficiently.

\begin{lemma}
\label{lem:expanders}
For every $t\in \mathbb{N}$ and $\mu>0$, there is a graph $H_{t,\mu}$ with a two-way labeling such that:
\begin{itemize}
\item $H$ has $2^t$ vertices and is $c$-regular for $c$ being a power of 2 bounded by $\poly(t,1/\mu)$.
\item $\lambda(H) \leq \mu$
\item \textup{Rot}$_H$ can be evaluated in linear space in its input length, i.e. space $O(t+\log c)$.
\end{itemize}
\end{lemma}
A short proof sketch for Lemma \ref{lem:expanders} can be found in \cite{MRSV17}. The following lemma says that high derandomized powers can be computed space-efficiently. 
\begin{lemma}
\label{lem:derandspacecomplexity}
Let $G_{0}$ be a $d$-regular, directed multigraph on $n$ vertices with a two-way labeling and $H_1,\ldots,H_k$ be $c$-regular undirected graphs with two-way labelings where for each $i\in[k]$, $H_{i}$ has $d\cdot c^{i-1}$ vertices. For each $i\in[k]$ let 
\[
G_i=G_{i-1}\ds H_i.
\]
Then given $v_0\in[n], i_{0}\in[d\cdot c^{i-1}], j_0\in[c]$, \textup{Rot}$_{G_i}(v,(i_0,j_0))$ can be computed in space $O(\log(n\cdot d)+k\cdot\log c)$ with oracle queries to \textup{Rot}$_{H_1},\ldots,\mathrm{Rot}_{H_k}$.
\end{lemma}
Lemma \ref{lem:derandspacecomplexity} is stated and proven in \cite{MRSV17} for the case of undirected multigraphs, however nothing in that proof required undirectedness. We include the proof in Appendix \ref{app:derandspacecomplexity}.

\begin{corollary}
\label{cor:entriesofM}
Let $\mw_0,\ldots,\mw_k$ be the transition matrices of $G_0,\ldots, G_k$ as defined  in Lemma \ref{lem:derandspacecomplexity}. For all $\ell\in[k]$, given coordinates $i,j\in [n]$, entry $i,j$ of $\mw_\ell$ can be computed in space $O(\log (n\cdot d) + k\cdot\log c)$.
\end{corollary}
\begin{proof}
Lemma \ref{lem:derandspacecomplexity} shows that we can compute neighbors in the graph $G_\ell$ in space $O(\log (n\cdot d) + k\cdot\log c)$. Given coordinates $i,j$ the algorithm initiates a tally $t$ at 0 and computes Rot$_{G_\ell}(i,q)$ for each $q$ from $1$ to $d\cdot c^{\ell-1}$, the degree of $G_\ell$. If the vertex outputted by Rot$_{G_\ell}$ is $j$, then $t$ is incremented by 1. After the loop finishes, $t$ contains the number of edges from $i$ to $j$ and the algorithm outputs $t/d\cdot c^{\ell-1}$, which is entry $i,j$ of $\mw_\ell$. This used space $O(\log (n\cdot d) + k\cdot\log c)$ to compute the rotation map of $G_\ell$ plus space $O(\log (d\cdot c^{\ell-1}))$ to store $q$ and $t$. So the total space usage is $O(\log (n\cdot d) + k\cdot\log c) +O(\log (d\cdot c^{\ell-1}))=O(\log (n\cdot d) + k\cdot\log c)$.
\end{proof}

Next we show that the complete graph is a good approximation of the $k$th derandomized power for sufficiently large $k$. We do this by noting that the derandomized square of a graph $G$, reduces $\lambda(G)$ and that, the smaller $\lambda(G)$, the better approximation it is to the complete graph. 

\begin{lemma}[\cite{RozenmanVa05}]
\label{lem:dsquare}
Let $G$ be a $d$-regular digraph with a two-way labeling and $\lambda(G)\leq\lambda$ and let $H$ be a $c$-regular graph on $d$ vertices with a two-way labeling and $\lambda(H)\leq\mu$. Then we have 
\[
\lambda(G\ds H)\leq 1-(1-\lambda^2)\cdot(1-\mu)\leq \lambda^{2}+\mu.
\]
\end{lemma}

 Noting that a $d$-regular digraph $G$ on $n$ vertices has $\lambda(G)\leq 1-1/(2\cdot d^2\cdot n^2)$ \cite{RozenmanVa05}, we get that setting $k=O(\log n)$ and setting $\mu=O(1/\polylog(n/\epsilon))$ (and hence $c=\polylog(n/\epsilon)$) in the recurrence from Lemma \ref{lem:dsquare} yields that the $k$th derandomized square $G_k$ has $\lambda(G_k)\leq \epsilon$. By Lemma \ref{lem:expander_approx_complete} this implies that the transition matrix $\mw_k$ of $G_k$ is such that $\mw_k\capprox_{\epsilon}\mJ$.

Next we note the space complexity of multiplying arbitrary matrices.

\begin{lemma}
\label{lem:matrixprod}
Given $n\times n$ matrices $\mm_1,\ldots,\mm_k$, their product $\mm_1\cdot\ldots\cdot \mm_k$ can be computed using $O(\log N\cdot \log k)$ space, where $N$ is the size of the input $(N=k\cdot n^2\cdot (\mathrm{bitlength~of~matrix~entries}))$.
\end{lemma}
\begin{proof}
This uses the natural divide and conquer algorithm and the fact that two matrices can be multiplied in logarithmic space. A detailed proof can be found in \cite{MRSV17}.
\end{proof}

Fourth, we show that we can boost a weak approximation to the pseudoinverse to a strong approximation in small space.

\begin{lemma}\label{lem:space_boost}
There is a deterministic algorithm that, given matrices $\ma,\mb,\mf \in \R^{n \times n}$, such that $\mf$ is PSD, and $\|\mI - \mb \ma\|_{\mf} \leq \alpha$ for some constant $\alpha < 1$, computes a matrix $\mproj$ such that $\|\mI-\mproj\ma\|_{\mf}\leq \epsilon$. The algorithm uses space $O(\log N\cdot\log(\log_{1/\alpha}(1/\epsilon)))$, where $N$ is the input size.
\end{lemma}
\begin{proof}
From Lemma \ref{lem:richardson-poly} we have that 
\[
\mproj= \sum_{i=0}^{O(\log_{1/\alpha}(1/\epsilon))} (\mI - \mb \ma)^i \mb
\]
has the desired property. Since the above polynomial has degree $O(\log(1/\epsilon))$, From Lemma \ref{lem:matrixprod}, we can compute $\mproj$ in space $O(\log N\cdot\log(\log_{1/\alpha}(1/\epsilon)))$ as desired.
\end{proof}

Finally, we show any extremely good approximation of the pseudoinverse in the ``spectral'' sense of the above Lemma~\ref{lem:space_boost} is also a good entrywise approximation of the true pseudoinverse.
\begin{lemma}\label{lem:spectral_to_entrywise_approx}
Suppose $\|\mI-\mproj\ma\|_{\mf}\leq \epsilon$ for conformable real square matrices $\mproj, \ma,\mf$ where $\mf$ is positive semidefinite. Suppose $\mproj,\ma,\mf$ all have the same left and right kernels, which are equal to each other. Then every entry of $\mproj$ differs by at most $\pm \frac{\epsilon}{2} \cdot (\lambda_\text{max}(\ma^{\top+} \mf \ma^{+})+\lambda_\text{max}(\mf^+))$ additively from the corresponding entry of $\ma^+$.

In particular, if for some parameter $N$, the matrices $\mf,\ma$ have min and max nonzero singular values between $1/\poly(N)$ and $\poly(N)$, then every entry of $\mproj$ differs by at most $\pm \epsilon \cdot \poly(N)$ from the corresponding entry of $\ma^+$.

Even more specifically, if $\ma$ is a regular Eulerian directed Laplacian matrix with integer edge weights $\leq N$ of a graph with $\gamma(G) \geq 1/\poly(N)$ for some $N \geq n$ and $F$ is bounded polynomially in $\mU_\mproj$ or $\mI$, then $\mproj$ differs by at most $\pm \epsilon \cdot \poly(N)$ from the corresponding entry of $\ma^+$.
\end{lemma}
\begin{proof}
The inequality in the lemma statement is equivalent to the statement that for all real conformable vectors $x,y$,
\[
y^\top (\ma^+ - \mproj) x \leq \epsilon \sqrt{ \left( x^\top \ma^{\top+} \mf \ma^{+} x \right) \cdot \left( y^\top \mf^+ y \right) }.
\]
By the AM-GM inequality,
\[
y^\top (\ma^+ - \mproj) x \leq \frac{\epsilon}{2} \left( x^\top \ma^{\top+} \mf \ma^{+} x + y^\top \mf^+ y \right).
\]
Since multiplying $x$ by a negative number changes the sign of the LHS but not the RHS,
\[
|y^\top (\ma^+ - \mproj) x| \leq \frac{\epsilon}{2} \left( x^\top \ma^{\top+} \mf \ma^{+} x + y^\top \mf^+ y \right).
\]
Consider any entry $(i,j)$ of $P^+$. Set $x,y$ to be the $i$th and $j$th standard basis vectors, respectively. Then the inequality says that the $(i,j)$ entries of $P^+$ and $A^+$ differ additvely by at most $\pm$ the RHS.

Thus, we need only bound the RHS. 
We know that the $x$ and $y$ we have selected are unit vectors. Thus, we may upper bound the RHS as
\[
\frac{\epsilon}{2} \left( x^\top \ma^{\top+} \mf \ma^{+} x + y^\top \mf y \right) \leq \frac{\epsilon}{2} \left( \max_x \frac{x^\top \ma^{\top+} \mf \ma^{+} x}{x^\top x} + \max_x \frac{y^\top \mf^+ y}{y^\top y} \right).
\]
The two terms inside the parentheses are the max eigenvalues of the symmetric matrices $\ma^{\top+} \mf \ma^{+}$ and $\mf^+$, respectively. The second part of the result follows from the fact that the operator norm is submultiplicative.

For the final part of the result, we prove that the max and min nonzero singular values of $\ma,\mf$ under the additional assumptions stated are between $1/\poly(N)$ and $\poly(N)$.

Note first that if we have a directed Eulerian graph with integer edge weights between $0$ and $N$ that the maximum singular value of its Laplacian is at most $\poly(n,N) \leq \poly(N)$ since all entries in the matrix are at most $\poly(N)$. 
It's minimum nonzero singular value is at least $1/\poly(N)$ by the assumption that the spectral gap is at least $1/\poly(N)$ and the triangle inequality. Thus, the max and min nonzero singular values of $\ma$ are between $1/\poly(N)$ and $\poly(N)$.

For the maximum and minimum nonzero singular values of $\mf$, or equivalently, max and min nonzero eigenvalues of $\mf$, it suffices to bound the min and max nonzero eigenvalues of $\mU_\mproj$. By the integrality of the edge weights of $\mproj$ and the fact that they are at most $N$, we have that the undirected Laplacian matrix $\mU_\mproj$ has has nonzero eigenvalues between $1/\poly(N)$ and $\poly(N)$.
\end{proof}

Now we can prove Theorem \ref{thm:space_main}
 
\begin{proof}[Proof of Theorem \ref{thm:space_main}]
Set $k=O(\log N)$. Let $\mw=\mw_0,\ldots,\mw_{k-1}$ be the transition matrices of the repeated derandomized square graphs $G_0,\ldots,G_{k-1}$ as defined in Lemma \ref{lem:derandspacecomplexity} using expanders $\{H_i\}$ with degree $c=\polylog(N)$ and hence $\lambda(H_i)\leq 1/\polylog(N))$. Let $\mw_k=\mJ=\vec{\bone}\vec{\bone}^\top/n$. It follows from Theorem \ref{thm:eul_derandsq_capprox} that for each $i\in\{0,\ldots,k-2\}$, $\mw_{i+1}\capprox_{1/\polylog(N)}\mw{i}^2$ and from Lemmas \ref{lem:dsquare} and \ref{lem:expander_approx_complete} that $\mw_{k-1}\capprox_{1/\polylog(N)}\mw_k=\mJ$.

Theorem \ref{thm:app_pinv_ring} and Proposition \ref{prop:tilde_L_pinv_fac} say that the matrix 
\[
{\mL^{(k)}}^{+} = 
\mproj_{\mL^{(k)}} \my_1^{-1} \cdots \my_i^{-1}
\begin{bmatrix}
    \mI_{(2^{k}-1)n} & 0 \\
    0 & \mI_{n} - \mw_{k}
\end{bmatrix}^{+}
\mx_k^{-1} \cdots \mx_{1}^{-1}
\mproj_{{\mL^{(k)}}^\top}
\]
where 
\[
\mx_j^{-1} = \begin{bmatrix}
     \mI_{(2^{k}-2^{k-j+1})n} & 0 & 0 \\
     0 & \mI_{2^{k-j}n} & 0 \\ 
    0 & \mI_{2^{k-j}} \otimes \mw_{j-1} & \mI_{2^{k-j}n}
\end{bmatrix},
\my_j^{-1} = \begin{bmatrix}
     \mI_{(2^{k}-2^{k-j+1})n} & 0 & 0 \\
     0 & \mI_{2^{k-j}n} & \mc_{2^{k-j}} \otimes \mw_{j-1} \\
    0& 0 & \mI_{2^{k-j}n}
\end{bmatrix}.
\]
is such that there exists a PSD matrix $\mf$ such that $\|\mI_{2^kn} - \widetilde{\mL}^+ \mL\|_\mf \leq O(k\epsilon)$ and $\mU_\mL/O(k) \preceq \mf \preceq O(2^{2k} n^2 k^5) \mU_\mL$.

Note that ${\mL^{(k)}}^{+}$ matrix is the product of $O(k)=O(\log N)$ matrices, each of which can be computed in space $O(\log N+k\cdot\log c)=O(\log N\cdot\log\log N)$ by Corollary \ref{cor:entriesofM}. Multiplying them together to compute $\widetilde{L}^+$ adds an additional $O(\log N\cdot \log k)=O(\log N\cdot\log\log N)$ space by Lemma \ref{lem:matrixprod} and the composition of space-bounded algorithms (Proposition \ref{prop:composition}). Thus the total space for computing ${\mL^{(k)}}^{+}$ is $O(\log N\cdot\log\log N)$.

Applying preconditioned Richardson to ${\mL^{(k)}}^{+}$, Lemma \ref{lem:space_boost} says that we can compute a matrix $\widetilde{\mL}^+$ such that $\|\mI_{\im(\mL)}-\widetilde{\mL}^+\mL\|_{\mf}\leq \epsilon$ using an additional $O(\log N\cdot\log\log(1/\epsilon)$ space. 

Since $\mf$ is bounded polynomially by $\mU_{\widetilde{\mL}^+}$, we can adjust $\epsilon$ in Lemma \ref{lem:space_boost} by a $\poly(N)$ factor to get approximation guarantees in $\|\cdot\|_{\mU_{\mL}}$ (or spectral norm) in space $O(\log N\cdot\log\log(N/\epsilon))$.

Applying Lemma~\ref{lem:pinv_reduce} to this approximate pseudoinverse yields an approximate pseudoinverse $\mproj$ for $\mI-\mw_0$.

Applying Lemma~\ref{lem:spectral_to_entrywise_approx}, we get that if we redefine the value of $\epsilon$ we have been using to $\epsilon \gets \epsilon / \poly(N)$ from the very beginning, then $\mproj$ is a $\pm \epsilon$ additive entrywise approximation of $(\mI-\mw_0)^+$. This only changes the space usage by at most a constant factor.
\end{proof}

\section{Estimating Random Walk Probabilities}
In this section we show that a space-efficient algorithm for Eulerian Laplacian systems can be used to get strong approximations to $k$-step random walk probabilities on Eulerian digraphs, and more generally, to approximate the product of potentially distinct Eulerian transition matrices. We then give a new result for approximating random walk probabilities on arbitrary digraphs. 
\begin{theorem}
\label{thm:state_probs}
Let $\mw_1,\ldots, \mw_k$ be the transition matrices of Eulerian digraphs $G_1,\ldots,G_k$ on $n$ vertices such that for vertex $v\in[n]$, the degree of $v$ is the same in $G_1,\ldots, G_k$. There is a deterministic algorithm that, when given $u,v\in [n]$ and $\epsilon>0$, computes 
\[
(\mw_1\mw_2\ldots \mw_k)_{uv}
\]
to within $\pm\epsilon$. The algorithm uses space $O(\log (N)\cdot \log\log(N/\epsilon))$, where $N$ is the size (bitlength) of the input.
\end{theorem}

An important special case of Theorem \ref{thm:state_probs} is when $\mw_1=\mw_2=\ldots=\mw_k$, in which case the algorithm computes an additive $\epsilon$ approximation to the probability that a random walk of length $k$ in the graph starting at $v$ ends at $u$ and uses space $O(\log (k\cdot N)\cdot \log\log(k\cdot N/\epsilon))$ where $N$ is the bitlength of $\mw_1$.

To prove Theorem \ref{thm:state_probs}, we first define escape probabilities.

\begin{definition}
In a graph $G=(V,E)$, for two vertices $u$ and $v$, the {\em escape probability} $p_{w}(u,v)$ denotes the probability that a random walk starting at vertex $w$ reaches $u$ before first reaching $v$. 
\end{definition}

In \cite{CKPPSV16}, they give a useful characterization of escape probabilities on arbitrary digraphs in terms of the pseudoinverse, which we restate below for Eulerian graphs.
\begin{lemma}[\cite{CKPPSV16} (Restated for Eulerian Graphs)]
\label{lem:escape_prob}
Let $\mw$ be a random walk matrix associated with a strongly connected Eulerian graph $G$. Let $s$ be the stationary distribution of $G$ and let $S$ be the diagonal matrix with $s$ on the diagonal. Let $u,v$ be two vertices. Let $p$ be the vector of escape probabilities, where $p_w$ represents the probability that a random walk starting at $w$ reaches $u$ before $v$. Then
\[
p=\beta(\alpha\cdot s+(\mI-\mw)^{+}(e_u-e_v))
\]
where 
\[
\alpha =-e_v^{\top}S^{-1}(\mI-\mw)^{+}(e_u-e_v)
\]
and 
\[
\beta = \frac{1}{s_u(e_u-e_v)^{\top}S^{-1}(\mI-\mw)^{+}(e_u-e_v)}.
\]
Futhermore, $\alpha=\poly(N)$ and $1/\beta=\poly(N)$.
\end{lemma}

It follows from the above lemma that we can use our algorithm to approximate escape probabilities in Eulerian graphs. Now we can prove Theorem \ref{thm:state_probs} by reducing computing certain matrix products to computing escape probabilities.

\begin{proof}[Proof of Theorem \ref{thm:state_probs}]
Note that $(\mw_1\mw_2\ldots \mw_k)_{uv}$ equals the probability that a random walk starting at vertex $v$ ends at $u$, if we take the first step on $G_k$, the second step on $G_{k-1}$ etc. and the final step on $G_1$. 

We construct a layered graph with $k+2$ layers and $n$ vertices in layers $1,\ldots,k+1$ each and a single vertex in layer $k+2$. All edges in the graph proceed from one layer to the subsequent one. For each $i\in[k]$, we place edges from layer $i$ to layer $i+1$ according to the graph $G_{k-i+1}$ (so layer 1 to 2 contains the edges from $G_k$ and layer $k$ to $k+1$ contains the edges from $G_1$). All vertices from layer $k+1$ send edges to the sole vertex in layer $k+2$ equal to their in-degree and the vertex in layer $k+2$ sends edges back to layer 1 equal to each vertex's out-degree. This ensures that the resulting graph is Eulerian. 

Notice that the probability that a random walk starting at vertex $v$ in layer 1 reaches vertex $u$ in layer $k+1$ before reaching the vertex in layer $k+2$ is exactly $(\mw_1\mw_2\ldots \mw_k)_{uv}$ because all walks of length greater than $k$ pass through the vertex in layer $k+2$ and no walks shorter than $k$ beginning in the first layer pass through it. Furthermore, this is an escape probability on an Eulerian graph and so by Lemma \ref{lem:escape_prob}, this can be computed by applying our Eulerian Laplacian solver to the layered graph followed by some vector multiplications and arithmetic. 

Applying the solver to the layered graph, uses space $O(\log (N)\cdot \log\log(N/\epsilon))$ to apply the pseudoinverse in the formula for escape probabilities given in Lemma~\ref{lem:escape_prob} to error $\epsilon/\poly(N)$. Since $\beta$ is at most $\poly(N)$ we can compute $\beta$ times the pseudoinverse term to $\epsilon$ error in this same amount of space. Since the stationary distribution of Eulerian graphs can be computed space efficiently, the remaining operations of computing $\alpha$ times the stationary distribution only add logarithmic space overhead. 
\end{proof}

Using perspectives developed for our space-efficient Eulerian Laplacian solver, we also obtain improved space-efficient algorithms for computing high-quality approximations to random walk probabilities in \emph{arbitrary} digraphs or Markov chains, improving both the best known randomized algorithms and deterministic algorithms. 

Aleliunas et. al. \cite{AleliunasKaLiLoRa79} observed that random walks can be easily simulated in randomized logarithmic space, and this yields the following algorithm for estimating random walk probabilities:

\begin{theorem}
\label{thm:rand_powers}
Suppose that $\mw\in\R^{n\times n}$ is a substochastic matrix, $k$ is a positive integer, and $\epsilon >0$. There is a randomized algorithm that, with high probability, computes a matrix $\widetilde{\mw}$ such that $\|\widetilde{\mw}-\mw^k\|\leq \epsilon$. The algorithm runs in space $O(\log (k\cdot N/\epsilon))$, where $N$ is the bitlength of the input.
\end{theorem}
\begin{proof}[Proof sketch]
For each entry $(\mw^k)_{uv}$, we simulate $\poly(n/\epsilon)$ random walks of length $k$ started at $v$ and count the fraction of them that end at $u$. 
\end{proof}
Saks and Zhou gave the best known derandomization of Theorem \ref{thm:rand_powers} \cite{SaksZh99}. 

\begin{theorem}[\cite{SaksZh99}]
\label{thm:sz}
Suppose that $\mw\in\R^{n\times n}$ is a substochastic matrix, $k$ is a positive integer, and $\epsilon >0$. There is a deterministic algorithm that computes a matrix $\widetilde{\mw}$ such that $\|\widetilde{\mw}-\mw^k\|\leq \epsilon$. The algorithm uses $O(\log(k\cdot n/\epsilon)\cdot \log^{1/2}(k))$ space.\footnote{Saks and Zhou state their result in terms of the $L_{\infty}$ matrix norm rather than the spectral norm. These are equivalent up to polynomial changes in $\epsilon$ (and hence constant factors in the space complexity).} 
\end{theorem}

We strengthen Theorems \ref{thm:rand_powers} and \ref{thm:sz} when $\epsilon=1/(k\cdot N)^{\omega(1)}$. This range of parameters is useful because $k$-step walk probabilities can be exponentially small in $k$, even on undirected graphs. 

\begin{theorem}
\label{thm:arbitrary_powers}
Suppose that $\mw\in\R^{n\times n}$ is a substochastic matrix, $k$ is a positive integer, and $\epsilon >0$. Then:
\begin{enumerate}
    \item There is a randomized algorithm that computes a matrix $\widetilde{\mw}$ such that each entry of $\widetilde{\mw}$ is with $\pm \epsilon$ of that of $\mw^k$ additively and runs in space $O(\log (k\cdot N)\cdot\log(\log_{k\cdot N}(1/\epsilon)))$
    \item There is a deterministic algorithm that computes a matrix $\widetilde{\mw}$ such that each entry of $\widetilde{\mw}$ is with $\pm \epsilon$ of that of $\mw^k$ additively and runs in space $O(\log(k\cdot N)\cdot\log^{1/2}(k)+\log(k\cdot n)\cdot\log(\log_{k\cdot N}(1/\epsilon)))$
\end{enumerate}
\end{theorem}

\begin{proof}
We obtain this result by applying preconditioned Richardson iteration (which is the same technique that gets the doubly logarithmic dependence on $\epsilon$ in Theorem \ref{thm:state_probs}) to a path-lifted (rather than cycle-lifted) graph. Our plan is to use Theorem \ref{thm:rand_powers} or Theorem \ref{thm:sz} to compute an approximation to the inverse of a path-lifted graph with quality $1/\poly(k\cdot n)$. Then we use preconditioned Richardson iteration to boost the error to $\epsilon$ using only additive $O(\log (k\cdot n)\cdot\log(\log_{k\cdot N}(1/\epsilon)))$ space as in Lemma \ref{lem:space_boost}.

Let $\mproj_k$ denote the adjacency matrix of the directed path of length $k$ ($\mproj$ is $(k+1)\times (k+1)$ with all ones just below the diagonal and zeroes elsewhere).  Let $\mm=\mproj_k\otimes\mw$. Note that $\mm^{k+1}$ is the all zeroes matrix and all eigenvalues of $\mm$ are 0. Let $\mI=\mI_{(k+1)\cdot n}$. We will approximate $\mI-\mm$ using the following identity: for all $\ma$ with eigenvalues less than $1$ in magnitude, we have 
\[
(\mI-\ma)^{-1}=\sum_{i=0}^{\infty}\ma^i.
\]
This means that 
\begin{align*}
(\mI-\mm)^{-1}&=\sum_{i=0}^{k}\mm^i\\
&=\sum_{i=0}^{k}(\mproj_{k})^i\otimes \mw^i
\end{align*}
because all powers beyond the $k$th are zero. Note that this matrix has $\mw^k$ in the lower left block, $\mw^{k-1}$ in the second diagonal from the bottom, etc. 

From Theorem \ref{thm:rand_powers}, there is a randomized algorithm that can compute matrices $\mw_{1},\ldots,\mw_{k}$ such that for each $i\in[k]$, we have $\|\mw^i-\mw_i\|\leq 1/(2\cdot k\cdot N)$ in space $O(\log(k\cdot N))$. Likewise, by Theorem \ref{thm:sz}, there is a deterministic algorithm that can compute matrices $\mw_{1},\ldots,\mw_{k}$ such that for each $i\in[k]$, we have $\|\mw^i-\mw_i\|\leq 1/(2\cdot k\cdot N)$ in space $O(\log(k\cdot N)\cdot\log^{1/2}(k))$. We can then approximate $(\mI-\mm)^{-1}$ as
\[
\mn=\mI+\sum_{i=1}^k (\mproj_{k})^i\otimes \mw_i.
\]
Note that for each $i\in[k]$ we have $\|\mproj_k^i\|=1$. Also, we have \begin{align*}
    \|\mI-\mm\|&=\|\mI-\mproj_k\otimes\mw\|\\
    &\leq \|\mI\|+\|\mproj_k\|\cdot\|\mw\|\\
    &\leq 2
\end{align*}
because $\|\mw\|\leq 1$. Putting this together gives
\begin{align*}
    \|\mI -\mn(\mI-\mm)\|&=\|((\mI-\mm)^{-1}-\mn)(\mI-\mm)\|\\
    &=\left\|\sum_{i=1}^k(\mproj_{k})^i\otimes(\mw^i-\mw_i)(\mI-\mm)\right\|\\
    &\leq \sum_{i=1}^{k}\|\mproj_{k}^i\|\cdot\|\mw^i-\mw_i\|\cdot \|\mI-\mm\|\\
    &\leq 2\cdot \sum_{i=1}^{k}1/(2\cdot k\cdot N)\\
    &=1/(k\cdot N)
\end{align*}
From Lemma \ref{lem:space_boost}, we can compute a matrix $\widetilde{\mn}$ such that $\|\mI-\widetilde{\mn}(\mI-\mm)\|\leq \epsilon / \|(\mI-\mm)^+\| $  using an additional space 
\begin{align*}
O(\log(k\cdot N)\cdot\log(\log_{k\cdot N}(\|(\mI-\mm)^+\|/\epsilon))) &= O(\log(k\cdot N)\cdot\log(\log_{k\cdot N}(N/\epsilon))) \\
&= O(\log(k\cdot N)\cdot\log(\log_{k\cdot N}(1/\epsilon)))
\end{align*}
for $\epsilon$ sufficiently small.

This implies $\|\widetilde{\mn}-(\mI-\mm)^{-1}\|\leq \epsilon$. Since the norm of a matrix is an upper bound on the norm of its submatrices, we have that the lower left $n\times n$ block of $\widetilde{\mn}$, is a matrix $\widetilde{\mw}$ satisfying $\|\widetilde{\mw}-\mw^k\|\leq\epsilon$.

To complete the proof, Lemma~\ref{lem:spectral_to_entrywise_approx} lets us convert this guarantee to entrywise approximation with a blowup of $\poly(N)$ in $\epsilon$. If we reduce $\epsilon$ to compensate for this, it only affects the stated space complexity by a constant factor.
\end{proof}

\section{Acknowledgements}
We thank William Hoza for pointing out an error in an earlier version of Theorem \ref{thm:sc_approx} and Ori Sberlo and Dean Doron for finding a bug in an earlier version of the proof of Theorem \ref{thm:eul_derandsq_capprox}.

\newpage

\bibliographystyle{alphanum}
\bibliography{pseudorandomness}

\newpage

\appendix 
\section{Approximation Equivalences Proofs}
\label{app:definitions}
Here we prove some lemmas from Section \ref{sect:definitions}. We begin with some linear algebraic technical lemmas.  

\begin{lemma}[\cite{CKPPRSV17} Lemma B.2]
\label{lem:equivdef_helper}
For all $\ma\in \R^{n\times n}$ and symmetric PSD $\mm,\mn\in\R^{n\times n}$, the following are equivalent:
\begin{enumerate}
    \item For all $x,y\in\R^n$ 
         \[
             |x^\top\ma y|\leq \frac{\epsilon}{2}\cdot\left(x^\top\mm x+y^\top\mn y\right)
          \]
    \item For all $x,y\in\C^n$
        \[
            |x^\top\ma y|\leq \epsilon\cdot \sqrt{(x^\top\mm x)\cdot(y^\top\mn y)}
        \]
    \item $\|\mm^{+/2}\ma \mn^{+/2}\|_{\R^n\rightarrow\R^n}\leq \epsilon$ and $\ker(\mm)\subseteq\ker(\ma^\top)$ and $\ker(\mn)\subseteq\ker(\ma)$
\end{enumerate}
where the notation $\|\cdot\|_{\R^n\rightarrow\R^n}$ denotes the 2-norm on a linear operator that maps $\R^n$ to $\R^n$.
\end{lemma}
Lemma \ref{lem:equivdef_helper} was used in \cite{CKPPRSV17} to give equivalent formulations of directed spectral approximation. Below, we extend the lemma to complex space, which allows us to give even more equivalent formulations of directed spectral approximation as well as unit-circle spectral approximation. 
\begin{lemma}[Extension of \cite{CKPPRSV17} Lemma B.2]
\label{lem:equivdef_helper_complex}
For all $\ma\in \C^{n\times n}$ and Hermitian PSD $\mm,\mn\in\C^{n\times n}$, the following are equivalent:
\begin{enumerate}
    \item For all $x,y\in\C^n$ 
         \[
             |x^*\ma y|\leq \frac{\epsilon}{2}\cdot\left(x^*\mm x+y^*\mn y\right)
          \]
    \item For all $x,y\in\C^n$
        \[
            |x^*\ma y|\leq \epsilon\cdot \sqrt{(x^*\mm x)\cdot(y^*\mn y)}
        \]
    \item $\|\mm^{+/2}\ma \mn^{+/2}\|_{\C^n\rightarrow\C^n}\leq \epsilon$ and $\ker(\mm)\subseteq\ker(\ma^*)$ and $\ker(\mn)\subseteq\ker(\ma)$
\end{enumerate}
where the notation $\|\cdot\|_{\C^n\rightarrow\C^n}$ denotes the 2-norm on a linear operator that maps $\C^n$ to $\C^n$. 
\end{lemma}
\begin{proof}
Let $\mL=\|\mm^{+/2}\ma \mn^{+/2}\|_2$. Since $\|x\|_2=\max_{\|y\|_2=1}|y^*x|$ we have that
\[
\mL=\max_{\|x\|=\|y\|=1}|x^*\mm^{+/2}\ma \mn^{+/2}y|.
\]
Performing the mapping $x\colon \mm^{+/2}x$ and $y\colon \mm^{+/2}y$ we have
\[
\mL=\max_{x^*\mm x=y^*\mn y=1}|x^*\ma y|=\max_{x\not\in\mathrm{ker}(\mm),y\not\in\mathrm{ker}(\mn)}\frac{|x^*\ma y|}{\sqrt{x^*\mm x\cdot y^*\mn y}}.
\]
It follows that if $\mL\leq \epsilon$ and $\ker(\mm)\subseteq\ker(\ma^*)$ and $\ker(\mn)\subseteq\ker(\ma)$ then
\[
|x^*\ma y|\leq \epsilon\cdot \sqrt{(x^*\mm x)\cdot(y^*\mn y)}.
\]
The kernel constraints guarantee that whenever the right-hand side is zero, the left-hand side must also be zero.  Finally, by the AM-GM inequality we get $\sqrt{x^*\mm x\cdot y^*\mn y}\leq (x^*\mm x + y^*\mn y)/2$ and this inequality is tight when $x^*\mm x=y^*\mn y=1$, so for all $x,y\in\C^n$ we have 
\[
|x^*\ma y|\leq \frac{\epsilon}{2}\cdot(x^*\mm x + y^*\mn y)
\]
\end{proof}

If follows from Lemma \ref{lem:extension_to_complex} that in fact all six formulations in Lemmas \ref{lem:equivdef_helper} and \ref{lem:equivdef_helper_complex} are equivalent because the spectral norm of a matrix is the same whether one defines it in terms of real or complex vectors.

Now we can prove the equivalences from Section \ref{sect:definitions}.

\begin{namedtheorem}[Lemma \ref{lem:dirapprox_to_complex} restated]
Let $\mw,\widetilde{\mw}\in\R^{n\times n}$ be (possibly asymmetric) matrices. Then $\widetilde{\mw}$ is a directed $\epsilon$-approximation of $\mw$ if and only if $\widetilde{\mw}$ is a complex $\epsilon$-approximation of $\mw$.
\end{namedtheorem}
\begin{proof}
Set $\ma=(\mw-\widetilde{\mw})$ and $\mm=\mn=\mU_{\mI-\mw}$. Note that $\mU_{\mI-\mw}$ is PSD. Observe that for all $v\in\C^n$
\[
v^*\mU_{\mI-\mw}v = \|v\|^2-v^*\mU_{\mw} v.
\]
Lemmas \ref{lem:equivdef_helper},\ref{lem:equivdef_helper_complex}, and \ref{lem:extension_to_complex} tell us that the following are equivalent:
\begin{enumerate}
    \item For all $x,y\in\R^n$ 
         \[
             |x^\top\ma y|\leq \frac{\epsilon}{2}\cdot\left(x^\top\mm x+y^\top\mn y\right)
          \]
    \item For all $x,y\in\C^n$ 
         \[
             |x^*\ma y|\leq \frac{\epsilon}{2}\cdot\left(x^*\mm x+y^*\mn y\right)
          \]
\end{enumerate}
Plugging in our settings for $\ma$, $\mm$, and $\mn$ completes the proof.
\end{proof}

\begin{namedtheorem}[Lemma \ref{lem:real_sym_equiv} restated]
Let $\mw,\widetilde{\mw}\in\R^{n\times n}$ be symmetric matrices. Then the following are equivalent:
\begin{enumerate}
    \item $\widetilde{\mw}\capprox_{\epsilon}\mw$
    \item $\widetilde{\mw}\approx_{\epsilon}\mw$ and $-\widetilde{\mw}\approx_{\epsilon}-\mw$
    \item For all $x\in\R^{n}$ we have
\[
\left|x^{\top}(\mw-\widetilde{\mw})x\right|\leq \epsilon\cdot\left(\|x\|^2-|x^{\top}Wx|\right).
\]
\end{enumerate} 
\end{namedtheorem}
\begin{proof}
First we show that items 2 and 3 are equivalent. From Definition \ref{def:undir_approx}, we have $\widetilde{\mw}\approx_{\epsilon}\mw$ and $-\widetilde{\mw}\approx_{\epsilon}-\mw$ if and only if for all $x\in\R^{n}$ we have
\[
\left|x^{\top}(\mw-\widetilde{\mw})x\right|\leq \epsilon\cdot\left(\|x\|^2-x^{\top}Wx\right)
\]
and for all $x\in\R^{n}$ we have
\[
\left|x^{\top}(\mw-\widetilde{\mw})x\right|\leq \epsilon\cdot\left(\|x\|^2+x^{\top}Wx\right).
\]
This is equivalent to: for all $x\in\R^{n}$ 
\begin{align*}
\left|x^{\top}(\mw-\widetilde{\mw})x\right|&\leq \min\left\{\epsilon\cdot\left(\|x\|^2-x^{\top}Wx\right),\epsilon\cdot\left(\|x\|^2+x^{\top}Wx\right)\right\}\\
&=\epsilon\cdot\left(\|x\|^2-|x^{\top}Wx|\right),
\end{align*}
which establishes the equivalence between items 2 and 3. Now we show that items 1 and 2 are equivalent, which will complete the proof. From Lemma \ref{lem:dirapprox_to_complex}, we have that item 2 is equivalent to: for all $x,y\in\C^n$
\[
\left|x^*(\mw-\widetilde{\mw})y\right|\leq \frac{\epsilon}{2}\cdot\left(\|x\|^2+\|y\|^2-x^{*}\mU_{\mw} x-y^{*}\mU_{\mw} y\right)
\]
and for all $x,y\in\C^n$
\[
\left|x^*(\mw-\widetilde{\mw})y\right|\leq \frac{\epsilon}{2}\cdot\left(\|x\|^2+\|y\|^2+x^{*}\mU_{\mw} x+y^{*}\mU_{\mw} y\right).
\]
Since $\mw$ is real and symmetric we have $\mU_{\mw}=\mw$ and we have that for all $v\in\C^n$, $v^*\mw v$ is real. Therefore the above is equivalent to: for all $x,y\in\C^n$
\begin{align*}
    \left|x^*(\mw-\widetilde{\mw})y\right|&\leq \min\left\{\frac{\epsilon}{2}\cdot\left(\|x\|^2+\|y\|^2-x^{*}\mw x-y^{*}\mw y\right),\frac{\epsilon}{2}\cdot\left(\|x\|^2+\|y\|^2+x^{*}\mw x+y^{*}\mw y\right)\right\}\\
    &=\frac{\epsilon}{2}\cdot\left(\|x\|^2+\|y\|^2-\left|x^{*}\mw x+y^{*}\mw y\right|\right),
\end{align*}
which is the requirement for unit-circle spectral approximation.
\end{proof}

\begin{namedtheorem}[Lemma \ref{lem:unitcirc_equivalences} restated]
Let $\mw,\widetilde{\mw}\in\mathbb{C}^{n\times n}$ be (possibly asymmetric) matrices. Then the following are equivalent 
\begin{enumerate}
    \item $\widetilde{\mw}\capprox_{\epsilon}\mw$
    \item For all $z\in\C$ such that $|z|=1$, $z\cdot\widetilde{\mw}\approx_{\epsilon}z\cdot\mw$
    \item For all $z\in\C$ such that $|z|=1$,  
    \begin{itemize}
        \item $\mathrm{ker}(\mU_{\mI-z\cdot \mw})\subseteq \mathrm{ker}(\widetilde{\mw}-\mw)\cap\mathrm{ker}((\widetilde{\mw}-\mw)^\top)$ and
        \item $\left\|\mU_{\mI-z\cdot \mw}^{+/2}(\widetilde{\mw}-\mw)\mU_{\mI-z\cdot \mw}^{+/2}\right\|\leq \epsilon$
    \end{itemize} 
\end{enumerate}
\begin{proof}
Using the definition of directed spectral approximation (Definition \ref{def:dir_approx}) and its extenstion to complex space from Lemma \ref{lem:dirapprox_to_complex} lets us write item 2 as: for all $x,y\in \C^n$ 
\begin{align}
\label{eq:unitcirc_equivalences1}
\left|x^*(z\cdot\mw-z\cdot\widetilde{\mw})y\right|&\leq \frac{\epsilon}{2}\cdot\left(\|x\|^2+\|y\|^2-x^{*}\mU_{z\cdot\mw} x-y^{*}\mU_{z\cdot\mw} y\right)\\
\label{eq:unitcirc_equivalences2}
&=\frac{\epsilon}{2}\cdot\left(\|x\|^2+\|y\|^2-\Real(x^{*}(z\cdot\mw) x+y^{*}(z\cdot\mw) y)\right).
\end{align}
Lemma \ref{lem:equivdef_helper_complex} immediately implies the equivalence between line \ref{eq:unitcirc_equivalences1} and item 3 in the lemma statement by setting $\ma=(\widetilde{\mw}-\mw)$ and $\mm=\mn=\mU_{\mI-z\cdot\mw}$.

Now we will show the equivalence between items 1 and 2. Since $|z|=1$, it does not affect the left-hand side of Inequality \ref{eq:unitcirc_equivalences1} and we can drop it. 

Item 2 in the lemma statement says that line \ref{eq:unitcirc_equivalences2} holds for all $z\in\C$ such that $|z|=1$. This is equivalent to it holding for the worst case $z$ (for each pair of vectors, $x,y$). For every complex number $\omega$, we have $\max_{z\colon |z|=1}\Real(z\omega) = |\omega|$ by setting $\Real(z)=\Real(\omega)/|\omega|$ and $\mathrm{Im}(z)=-\mathrm{Im}(\omega)/|\omega|$. Therefore the inequality becomes
\[
\left|x^*(\mw-\widetilde{\mw})y\right|\leq\frac{\epsilon}{2}\cdot\left(\|x\|^2+\|y\|^2-\left|x^{*}\mw x-y^{*}\mw y\right|\right),
\]
which is exactly the condition we require for $\widetilde{\mw}\capprox_{\epsilon}\mw$. 
\end{proof}
\end{namedtheorem}

\section{Omitted Proofs from Section \ref{sect:ring_powers}}
\label{app:ring_powers}
\begin{namedtheorem}[Lemma \ref{lem:exact_eigenspace} restated]
Fix $\widetilde{\mw},\mw\in\C^{n\times n}$. For any matrix $M\in \C^{n\times n}$, let $V_{\lambda}(M)$ denote the eigenspace of $M$ of eigenvalue $\lambda$. 
\begin{enumerate}
\item If $\widetilde{\mw}\approx_{c}\mw$ for a finite $c>0$, then $V_1(\mw)\subseteq V_1(\widetilde{\mw})$. If $c<1$, then $V_1(\mw)= V_1(\widetilde{\mw})$.
\item If $\widetilde{\mw}\capprox_{c}\mw$ for a finite $c>0$, then for all $\lambda\in \C$ such that $|\lambda|=1$, $V_{\lambda}(\mw)\subseteq V_{\lambda}(\widetilde{\mw})$ and if $c<1$ then $V_{\lambda}(\mw)= V_{\lambda}(\widetilde{\mw})$.
\end{enumerate}
\end{namedtheorem}
\begin{proof}
Suppose $\widetilde{\mw}\approx_{c}\mw$ for a finite $c$. Applying the equivalence of Lemma \ref{lem:equivdef_helper_complex}, this means that for all $x,y\in\C^n$,
\[
|x^*(\mw-\widetilde{\mw})y|\leq c\cdot\sqrt{(x^*(\mI-\mw)x)\cdot(y^*(\mI-\mw)y)}.
\]
Let $v$ be an eigenvector of $\mw$ with eigenvalue 1. Then setting $y=v$ above makes the right-hand side equal 0. Therefore we have that for all $x\in\C^n$,
\[
|x^*(\mw-\widetilde{\mw})v|=|x^*(\mI-\widetilde{\mw})v|=0
\]
which can only occur if $(\mI-\widetilde{\mw})v=\bzero$. Suppose now that $c<1$ and let $v$ be an eigenvector of $\widetilde{\mw}$ with eigenvalue 1. Then we have 
\[
|v^*(\mw-\widetilde{\mw})v|=v^*(\mI-\mw)v\leq \frac{c}{2}\cdot\left(v^*(\mI-\mw)v+v^*(\mI-\mw)v\right)=c\cdot v^*(\mI-\mw)v.
\]
Since $c<1$ and $\mU_{\mI-\mw}$ is PSD, the above inequality can only be true when both sides equal zero. Applying the equivalence of Lemma \ref{lem:equivdef_helper_complex}, we have that for all $x\in\C^n$
\[
|x^*(\mw-\widetilde{\mw})v|=|x^*(\mI-\mw)v|\leq c\cdot\sqrt{(x^*(\mI-\mw)x)\cdot(v^*(\mI-\mw)v)}=0.
\]
The above can only happen for all vectors $x$ if $(\mI-\mw)v=\bzero$.

For the second claim, suppose $\widetilde{\mw}\capprox_{c}\mw$ for finite $c$. Then for all $z$ such that $|z|=1$ we have $z\cdot\widetilde{\mw}\approx_{c}z\cdot\mw$. Let $v$ be an eigenvector of $\mw$ with eigenvalue $\lambda$ such that $|\lambda|=1$. Set $z=\lambda^{-1}$ so $z\cdot\mw v=v$. From the first part of the lemma, we have $V_1(z\cdot \mw)\subseteq V_1(z\cdot \widetilde{\mw})$ and hence $z\cdot \widetilde{\mw} v=v$, which implies $\widetilde{\mw} v=\lambda\cdot v$. A similar argument applies when $c<1$.
\end{proof}

\begin{namedtheorem}[Lemma \ref{lem:subspaces} restated]
Let $\mm\colon\mathbb{C}^n\rightarrow\mathbb{C}^n$ be a linear operator and $V_1,\ldots,V_{\ell}\subseteq\mathbb{C}^n$ subspaces such that 
\begin{enumerate}
    \item $V_{j}\perp V_{k}$ for all $j\neq k$
    \item $V_1 \oplus \ldots \oplus V_\ell=\mathbb{C}^n$
    \item $\mm V_j\subseteq V_j$ for all $j\in[\ell]$.
\end{enumerate}
I.E., $M$ is block diagonal with respect to the subspaces $V_1,\ldots,V_\ell$. Then, 
\[\|\mm\|=\max_{j\in[\ell]}\|\mm|_{V_j}\|\]
\end{namedtheorem}
\begin{proof}
For $v\in\mathbb{C}^n$, by property 2 we can write $v=v_1+\ldots+v_{\ell}$ where $v_j\in V_j$ for all $j\in[\ell]$. Using this as well as properties 1 and 3, we get
\begin{align*}
    \|\mm v\|^2 &=\left\|\sum_{j\in[\ell]}\mm v_j\right\|^2\\
    &=\sum_{j\in[\ell]}\|\mm v_j\|^2\\
    &\leq \sum_{j\in[\ell]}\|\mm |_{V_j}\|^2\cdot\|v_j\|^2\\
    &\leq\left(\max_{j\in[\ell]}\|\mm |_{V_j}\|^2\right)\cdot \|v\|^2.
\end{align*}
So $\|\mm \|\leq\max_{j\in[\ell]\|\mm|_{V_j}\|}$. The other direction, $\|\mm \|\geq\|\mm|_{V_j}\|$ for all $j\in[\ell]$ is immediate. 
\end{proof}

\begin{namedtheorem}[Theorem \ref{thm:sc_approx} restated]
Fix $\mw,\widetilde{\mw}\in\C^{n}$ and suppose that $\widetilde{\mw}\approx_{\epsilon}\mw$ for $\epsilon\in(0,2/3)$. Let $F\subseteq[n]$ such that $(\mI_{|F|}-\mw_{FF})$ is invertible and let $C=[n]\setminus F$. Then 
\[
\mI_{|C|}-\schur(\mI_n-\widetilde{\mw},C)\approx_{\epsilon/(1-3\epsilon/2)}\mI_{|C|}-\schur(\mI_n-\mw,C)
\]
\end{namedtheorem}
\begin{proof}
Let $\ma = \mI_n - \widetilde{\mw}$, $\mb = \mI_n - \mw$. Note that since $\widetilde{\mw} \approx_\epsilon \mw$, we have for all $x,y \in \C^n$
\begin{equation}\label{eq:ab_approx}
|x^* (\mb - \ma) y| \leq \frac{\epsilon}{2}\left(x^* \mU_\mb x + y^* \mU_\mb y \right),
\end{equation}
our goal is to show that $
\mI_{|C|}-\schur(\ma,C)\approx_{\epsilon/(1-3\epsilon/2)}\mI_{|C|}-\schur(\mb,C)
$, which is equivalent to 
$$
\left|
[x^{\ell}]^{*}\left(\schur(\ma,C)-\schur(\mb,C)\right)x^{r}
\right|
\leq
\frac{\epsilon}{2-3\epsilon}\left[[x^{\ell}]^{*}\mU_{\schur(\mb,C)}x^{\ell}+[x^{r}]^{*}\mU_{\schur(\mb,C)}[x^{r}]\right].
$$
for all $x^{\ell}, x^{r} \in \C^{|C|}$.

Let $x^{r},x^{\ell}\in\C^{C}$ be arbitrary.
Now define $y^{r}\in\C^{n}$ by $y_{C}^{r}=x^{r}$ and $y_{F}^{r}=-\ma_{FF}^{-1}\ma_{FC}x^{r}$.
Further, define $y^{\ell}\in\C^{n}$ by $y_{C}^{\ell}=x^{\ell}$,
$y_{F}^{\ell}=-[\mb_{FF}^{*}]^{-1}[\mb_{CF}^{*}]x^{\ell}$. Note
that $y^{r}$ and $y^{\ell}$ are defined so that 
\[
\ma y^{r}=\left(\begin{array}{cc}
\ma_{FF} & \ma_{FC}\\
\ma_{CF} & \ma_{CC}
\end{array}\right)\left(\begin{array}{c}
-\ma_{FF}^{-1}\ma_{FC}x^{r}\\
x^{r}
\end{array}\right)=\left(\begin{array}{c}
\vzero_{F}\\
\schur(\ma,C)x^{r}
\end{array}\right)
\]
and
\[
[y^{\ell}]^{*}\mb=\left(\begin{array}{c}
-[\mb_{FF}^{*}]^{-1}[\mb_{CF}^{*}]x^{\ell}\\
x^{\ell}
\end{array}\right)^{*}\left(\begin{array}{cc}
\mb_{FF} & \mb_{FC}\\
\mb_{CF} & \mb_{CC}
\end{array}\right)=\left(\begin{array}{c}
\vzero_{F}\\
\schur(\mb,C)^{*}x^{\ell}
\end{array}\right)^{*}\,.
\]

Consequently $[y^{\ell}]^{*}\ma y^{r}=[x^{\ell}]^{*}\schur(\ma,C)x^{r}$
and $[y^{\ell}]^{*}\mb y^{r}=[x^{\ell}]^{*}\schur(\mb,C)x^{r}$.
Therefore, by Equation~\eqref{eq:ab_approx} we have that 
\begin{equation}
\label{eq:sc}
\left|[x^{\ell}]^{*}\schur(\ma,C)x^{r}-[x^{\ell}]^{*}\schur(\mb,C)x^{r}\right|
=
\left|[y^{\ell}]^{*}\left(\ma-\mb\right)y^{r}\right|
\leq
\frac{\epsilon}{2}\left[[y^{\ell}]^{*}\mU_{\mb}y^{\ell}+[y^{r}]^{*}\mU_{\mb}y^{r}\right]\,.
\end{equation}
Now, 
\[
[y^{\ell}]^{*}\mU_{\mb}y^{\ell}=[y^{\ell}]^{*}\left(\frac{\mb+\mb^*}{2}\right) y^{\ell}=[x^{\ell}]^{*}\mU_{\schur(\mb,C)}x^{\ell}
\]


and since we know that $(1-\epsilon)\mU_{\mb}\preceq\mU_{\ma}\preceq(1+\epsilon)\mU_{\mb}$
we have that 
\[
[y^{r}]^{*}\mU_{B}[y^{r}]\preceq\frac{1}{1-\epsilon}\cdot [y^{r}]^{*}\mU_{\ma}[y^{r}]=\frac{1}{1-\epsilon}\cdot [x^{r}]^{*}\mU_{\schur(\ma,C)}[x^{r}]\,.
\]
Substituting these into Equation \ref{eq:sc} gives 
\[
\left|[x^{\ell}]^{*}(\schur(\ma,C)-\schur(\mb,C))x^{r}\right|
\leq
\frac{\epsilon}{2}\left[[x^{\ell}]^{*}\mU_{\schur(\mb,C)}x^{\ell}+[x^{r}]^{*}\mU_{\schur(\ma,C)}[x^{r}]/(1-\epsilon)\right]\,.
\]
Plugging in  $x^{\ell}=x^{r}$ above,
\begin{align*}
[x^{r}]^{*}\left(\mU_{\schur(\ma,C)}-\mU_{\schur(\mb,C)}\right)x^{r}
&\leq \left|[x^{r}]^{*}\left(\schur(\ma,C)-\schur(\mb,C)\right)x^{r}\right|\\
&\leq
\frac{\epsilon}{2}\left[[x^{r}]^{*}\mU_{\schur(\mb,C)}x^{r}+[x^{r}]^{*}\mU_{\schur(\ma,C)}[x^{r}]/(1-\epsilon)\right]\,.
\end{align*}
and rearranging terms gives
\[
[x^{r}]^{*}\mU_{\schur(\ma,C)}[x^{r}]\leq\frac{1+\epsilon/2}{1-\epsilon/(2\cdot(1-\epsilon))}\cdot [x^{r}]^{*}\mU_{\schur(\mb,C)}[x^{r}],
\]
which implies
\begin{align*}
\left|[x^{\ell}]^{*}\left(\schur(\ma,C)-\schur(\mb,C)\right)x^{r}\right|
&\leq
\frac{\epsilon}{2}\left[[x^{\ell}]^{*}\mU_{\schur(\mb,C)}x^{\ell}+\frac{1+\epsilon/2}{1-\epsilon/(2\cdot(1-\epsilon))}\cdot\frac{1}{1-\epsilon}\cdot[x^{r}]^{*}\mU_{\schur(\mb,C)}[x^{r}]\right]\\
&=\frac{\epsilon}{2}\left[[x^{\ell}]^{*}\mU_{\schur(\mb,C)}x^{\ell}+\frac{1+\epsilon/2}{1-3\epsilon/2}\cdot[x^{r}]^{*}\mU_{\schur(\mb,C)}[x^{r}]\right]
\end{align*}
By symmetry, i.e. repeating the above argument on the conjugate transposes of the matrices we also have that 
\[
\left|[x^{\ell}]^{*}\left(\schur(\ma,C)-\schur(\mb,C)\right)x^{r}
\right|
\leq
\frac{\epsilon}{2}\cdot \frac{1+\epsilon/2}{1-3\epsilon/2}\cdot \left[x^{\ell}]^{*}\mU_{\schur(\mb,C)}x^{\ell}+[x^{r}]^{*}\mU_{\schur(\mb,C)}[x^{r}]\right].
\]
Taking the average of these two equations then yields that 
\[
\left|
[x^{\ell}]^{*}\left(\schur(\ma,C)-\schur(\mb,C)\right)x^{r}
\right|
\leq
\frac{\epsilon}{4}\cdot\frac{2-\epsilon}{1-3\epsilon/2}\cdot\left[[x^{\ell}]^{*}\mU_{\schur(\mb,C)}x^{\ell}+[x^{r}]^{*}\mU_{\schur(\mb,C)}[x^{r}]\right]\,.
\]
Since 
\[
\frac{\epsilon}{4}\cdot\frac{2-\epsilon}{1-3\epsilon/2}\leq \frac{1}{2}\cdot\frac{\epsilon}{2-3\epsilon}
\]
and $x^{\ell}$ and $x^{r}$ were arbitrary, the result follows.
\end{proof}

The loss of $1/(1-3\epsilon/2)$ in the approximation quality in Theorem \ref{thm:sc_approx} is avoidable with a small change to the notion of approximation. 

\newcommand{\minapprox}{\mathbin{\stackrel{\rm \min}{\approx}}}

\begin{definition}[Min Spectral Approximation]
\label{def:min_approx}
Let $\mw,\widetilde{\mw}\in\C^{n\times n}$ be (possibly asymmetric) matrices. We say that $\widetilde{\mw}$ is a \emph{min $\epsilon$-approximation} of $\mw$ (written $\widetilde{\mw}\minapprox_{\epsilon}\mw$) if for all $x,y\in\C^n$
\[
\left|x^*(\mw-\widetilde{\mw})y\right|
\leq \frac{\epsilon}{2}\cdot\min\left\{x^*\mU_{\mI-\mw}x+y^*\mU_{\mI-\widetilde{\mw}}y,y^*\mU_{\mI-\mw}y+x^*\mU_{\mI-\widetilde{\mw}}x\right\}
\]
\end{definition}
In min spectral approximation, the error on the right-hand side is measured with respect to both matrices rather than just one of them. Min spectral approximation is equivalent to the original notion of spectral approximation up to a small loss in approximation quality, as seen in the following lemma. 
\begin{lemma}
Let $\mw,\widetilde{\mw}\in\C^{n\times n}$ be (possibly asymmetric) matrices and suppose $\epsilon\in(0,1)$.  
\begin{enumerate}
    \item If $\widetilde{\mw}\minapprox_{\epsilon}\mw$ then $\widetilde{\mw}\approx_{\epsilon/(1-\epsilon/2)}\mw$
    \item If $\widetilde{\mw}\approx_{\epsilon}\mw$ then $\widetilde{\mw}\minapprox_{\epsilon/(1-\epsilon)}\mw$.
\end{enumerate}
\end{lemma}
\begin{proof}
Suppose $\widetilde{\mw}\minapprox_{\epsilon}\mw$ and fix $x\in\C^n$. Then we can write
\begin{align*}
    x^*\left(\mU_{\mI-\widetilde{\mw}}-\mU_{\mI-\mw}\right)x&=\Real\left( x^*(\mw-\widetilde{\mw})x\right)\\
    &\leq \left|x^*(\mw-\widetilde{\mw})x\right|\\
    &\leq \frac{\epsilon}{2}\cdot\left(x^*\mU_{\mI-\mw}x+x^*\mU_{\mI-\widetilde{\mw}}x\right).
\end{align*}
Rearranging the above gives
\[
x^*\mU_{\mI-\widetilde{\mw}}x\leq \frac{1+\epsilon/2}{1-\epsilon/2}\cdot x^*\mU_{\mI-\mw}x
\]
which implies
\begin{align*}
\left|x^*(\mw-\widetilde{\mw})y\right|
&\leq \frac{\epsilon}{2}\cdot\left(y^*\mU_{\mI-\mw}y+x^*\mU_{\mI-\widetilde{\mw}}x\right)\\
&\leq \frac{\epsilon}{2}\cdot\left(y^*\mU_{\mI-\mw}y+\frac{1+\epsilon/2}{1-\epsilon/2}\cdot x^*\mU_{\mI-\mw}x\right).
\end{align*}
A similar calculation gives
\[
\left|x^*(\mw-\widetilde{\mw})y\right|
\leq \frac{\epsilon}{2}\cdot\left(\frac{1+\epsilon/2}{1-\epsilon/2}\cdot y^*\mU_{\mI-\mw}y+ x^*\mU_{\mI-\mw}x\right).
\]
Averaging these two inequalities and noting that 
\[
\frac{1}{2}\cdot\frac{\epsilon}{2}\cdot\left(1+\frac{1+\epsilon/2}{1-\epsilon/2}\right)=\frac{\epsilon}{2\cdot(1-\epsilon/2)}
\]
gives 
\[
\left|x^*(\mw-\widetilde{\mw})y\right|\leq \frac{\epsilon}{2\cdot(1-\epsilon/2)}\cdot\left(x^*\mU_{\mI-\mw}x+y^*\mU_{\mI-\mw}y\right),
\]
or equivalently $\widetilde{\mw}\approx_{\epsilon/(1-\epsilon/2)}\mw$. 

Now suppose $\widetilde{\mw}\approx_{\epsilon}\mw$. It follows that $\mU_{\widetilde{\mw}}\approx_{\epsilon}\mU_{\mw}$ and hence 
\[
(1-\epsilon)\cdot\mU_{\mI-\mw}\preceq \mU_{\mI-\widetilde{\mw}}\preceq (1+\epsilon) \cdot\mU_{\mI-\mw},
\]
which implies that $\mU_{\mI-\mw}\preceq (1/(1-\epsilon))\cdot \mU_{\mI-\widetilde{\mw}}$. Now we can write that for all $x,y\in\C^n$
\begin{align*}
\left|x^*(\mw-\widetilde{\mw})y\right|&\leq\frac{\epsilon}{2}\cdot\left(x^*\mU_{\mI-\mw}x+y^*\mU_{\mI-\mw}y\right)\\
&\leq \frac{\epsilon}{2}\cdot\left(\frac{1}{1-\epsilon}\cdot x^*\mU_{\mI-\widetilde{\mw}}x+y^*\mU_{\mI-\mw}y\right)\\
&\leq \frac{\epsilon}{2\cdot(1-\epsilon)}\cdot\left(x^*\mU_{\mI-\widetilde{\mw}}x+y^*\mU_{\mI-\mw}y\right).
\end{align*}
A similar calculation shows
\[
\left|x^*(\mw-\widetilde{\mw})y\right|\leq \frac{\epsilon}{2\cdot(1-\epsilon)}\cdot\left(x^*\mU_{\mI-\mw}x+y^*\mU_{\mI-\widetilde{\mw}}y\right)
\]
and hence $\widetilde{\mw}\minapprox_{\epsilon/(1-\epsilon)}\mw$.
\end{proof}

Now we show that min spectral approximation is preserved under schur complements with no loss in approximation quality.

\begin{theorem}
Let $\mw,\widetilde{\mw}\in\C^{n}$ and suppose that $\widetilde{\mw}\minapprox_{\epsilon}\mw$ for $\epsilon\in(0,1)$. Let $F\subseteq[n]$ such that $(\mI_{|F|}-\mw_{FF})$ is invertible and let $C=[n]\setminus F$. Then 
\[
\mI_{|C|}-\schur(\mI_n-\widetilde{\mw},C)\minapprox_{\epsilon}\mI_{|C|}-\schur(\mI_n-\mw,C)
\]
\end{theorem}
\begin{proof}
Let $\ma = \mI_n - \widetilde{\mw}$, $\mb = \mI_n - \mw$. Note that since $\widetilde{\mw} \minapprox_\epsilon \mw$, we have for all $x,y \in \C^n$
\begin{equation}\label{eq:ab_approx2}
|x^* (\mb - \ma) y| \leq \frac{\epsilon}{2}\min\left\{x^* \mU_\mb x + y^* \mU_\ma y, y^* \mU_\mb y + x^* \mU_\ma x \right\},
\end{equation}

Our goal is to show that $
\mI_{|C|}-\schur(\ma,C)\minapprox_{\epsilon}\mI_{|C|}-\schur(\mb,C)
$, which is equivalent to 
\begin{align*}
&\left|
[x^{\ell}]^{*}\left(\schur(\ma,C)-\schur(\mb,C)\right)x^{r}\right|\\
&\leq
\frac{\epsilon}{2}\cdot\min\left\{[x^{\ell}]^{*}\mU_{\schur(\ma,C)}x^{\ell}+[x^{r}]^{*}\mU_{\schur(\mb,C)}[x^{r}],[x^{\ell}]^{*}\mU_{\schur(\mb,C)}x^{\ell}+[x^{r}]^{*}\mU_{\schur(\ma,C)}[x^{r}]\right\}.
\end{align*}
for all $x^{\ell}, x^{r} \in \C^{|C|}$.

Let $x^{r},x^{\ell}\in\C^{C}$ be arbitrary.
Now define $y^{r}\in\C^{n}$ by $y_{C}^{r}=x^{r}$ and $y_{F}^{r}=-\ma_{FF}^{-1}\ma_{FC}x^{r}$.
Further, define $y^{\ell}\in\C^{n}$ by $y_{C}^{\ell}=x^{\ell}$,
$y_{F}^{\ell}=-[\mb_{FF}^{*}]^{-1}[\mb_{CF}^{*}]x^{\ell}$. Note
that $y^{r}$ and $y^{\ell}$ are defined so that 
\[
\ma y^{r}=\left(\begin{array}{cc}
\ma_{FF} & \ma_{FC}\\
\ma_{CF} & \ma_{CC}
\end{array}\right)\left(\begin{array}{c}
-\ma_{FF}^{-1}\ma_{FC}x^{r}\\
x^{r}
\end{array}\right)=\left(\begin{array}{c}
\vzero_{F}\\
\schur(\ma,C)x^{r}
\end{array}\right)
\]
and
\[
[y^{\ell}]^{*}\mb=\left(\begin{array}{c}
-[\mb_{FF}^{*}]^{-1}[\mb_{CF}^{*}]x^{\ell}\\
x^{\ell}
\end{array}\right)^{*}\left(\begin{array}{cc}
\mb_{FF} & \mb_{FC}\\
\mb_{CF} & \mb_{CC}
\end{array}\right)=\left(\begin{array}{c}
\vzero_{F}\\
\schur(\mb,C)^{*}x^{\ell}
\end{array}\right)^{*}\,.
\]
Consequently $[y^{\ell}]^{*}\ma y^{r}=[x^{\ell}]^{*}\schur(\ma,C)x^{r}$
and $[y^{\ell}]^{*}\mb y^{r}=[x^{\ell}]^{*}\schur(\mb,C)x^{r}$. We also have $[y^{\ell}]^{*}\mU_{\mb}y^{\ell}=[x^{\ell}]^{*}\mU_{\schur(\mb,C)}x^{\ell}$ and $[y^{r}]^{*}\mU_{\ma}y^{r}=[x^{r}]^{*}\mU_{\schur(\ma,C)}x^{r}$. Therefore, by Equation~\eqref{eq:ab_approx2} we get
\begin{align*}
\left|[x^{\ell}]^{*}\schur(\ma,C)x^{r}-[x^{\ell}]^{*}\schur(\mb,C)x^{r}\right|
&=
\left|[y^{\ell}]^{*}\left(\ma-\mb\right)y^{r}\right|\\
&\leq
\frac{\epsilon}{2}\cdot\left([y^{\ell}]^{*}\mU_{\mb}y^{\ell}+[y^{r}]^{*}\mU_{\ma}y^{r}\right)\\
&=\frac{\epsilon}{2}\cdot\left([x^{\ell}]^{*}\mU_{\schur(\mb,C)}x^{\ell}+[x^{r}]^{*}\mU_{\schur(\ma,C)}x^{r}\right).
\end{align*}
Redefining $y^{\ell}$ and $y^r$, we can get the analogous inequality 
\[
\left|[x^{\ell}]^{*}\schur(\ma,C)x^{r}-[x^{\ell}]^{*}\schur(\mb,C)x^{r}\right|\leq\frac{\epsilon}{2}\cdot\left([x^{\ell}]^{*}\mU_{\schur(\ma,C)}x^{\ell}+[x^{r}]^{*}\mU_{\schur(\mb,C)}x^{r}\right).
\]
Therefore $\mI_{|C|}-\schur(\ma,C)\minapprox_{\epsilon}\mI_{|C|}-\schur(\mb,C)$.
\end{proof}

\section{Proof of Lemma \ref{lem:expander_approx_complete}}
\label{app:expander_approx_complete}
\begin{namedtheorem}[Lemma \ref{lem:expander_approx_complete} restated]
Let $G$ be a strongly connected, regular directed multigraph on $n$ vertices with transition matrix $\mw$ and let $\mJ\in\mathbb{R}^{n\times n}$ be a matrix with $1/n$ in every entry (i.e. $\mJ$ is the transition matrix of the complete graph with a self loop on every vertex). Then $\lambda(G)\leq \lambda$ if and only if $\mw\capprox_{\lambda}\mJ$.
\end{namedtheorem}
\begin{proof}
First we will show that if $\mw\capprox_{\lambda}\mJ$ then $\lambda(G)\leq\lambda$. By the equivalence of items 1 and 2 in Lemma \ref{lem:equivdef_helper}, we have that for all $x,y\in\C^n$
\[
x^*(\mw-\mJ)y\leq \lambda\cdot\sqrt{x^*(\mI-\mJ)x\cdot y^*(\mI-\mJ)y}.
\]
Note that we can write
\[
\lambda(G)=\max_{u\perp\bone,v\perp\bone}\frac{u^*\mw v}{\|u\|\|v\|}.
\]
Let $\bar{u}$ and $\bar{v}$ be vectors that achieve the maximization above. Setting $x=\bar{u}$ and $y=\bar{v}$, and noting that $\mJ \bar{v}=\mJ \bar{u}=\bzero$ because $\bar{u}$ and $\bar{v}$ are perpendicular to $\bone$, our inequality becomes 
\[
\bar{u}^*\mw\bar{v}\leq \lambda\cdot\sqrt{\|\bar{u}\|^2\cdot\|\bar{v}\|^2}=\|\bar{u}\|\cdot\|\bar{v}\|.
\]
Dividing by $\|\bar{u}\|\cdot\|\bar{v}\|$ completes the proof. 

For the other direction, assume $\lambda(G)\leq \lambda$ and we will show that for all $x,y\in\mathbb{C}^n$ and all $z\in\C$ such that $|z|=1$, 
\begin{equation}
\label{eq:expander_approx_complete}
|x^{*}(\mw-\mJ)y|\leq \lambda\cdot\sqrt{|x^{*}(\mI-z\mJ)x|\cdot |y^{*}(\mI-z\mJ)y|}
\end{equation}

Note that for all constant vectors $v$ (i.e. $v$ has the same entry in every coordinate) we have by the regularity of $\mw$ and $\mJ$,
\[
\mw v = \mw^{\top}v=\mJ v=\mJ^{\top}v=v.
\]
So when $x$ or $y$ is constant, the left-hand side of Inequality (\ref{eq:expander_approx_complete}) is zero and the inequality is true (because the right-hand side is always non-negative). Furthermore, orthogonality to $\vec{\bone}$ is preserved under $\mI,\mJ,\mw,$ and $\mw^{\top}$ so it suffices to consider vectors $x,y\perp\vec{\bone}$. Fix two such vectors. For $v\perp\vec{\bone}$, we have $\mJ v=\vec{\bzero}$. So Inequality (\ref{eq:expander_approx_complete}) becomes 
\begin{align*}
|x^{*}\mw y|&\leq \lambda\cdot\sqrt{x^{*}x\cdot y^{*}y}\\
&=\lambda\cdot\|x\|\cdot\|y\|.
\end{align*}
Applying Cauchy-Schwarz, we have
\begin{align*}
  |x^{\top}\mw y|&\leq \|x\|\cdot \|\mw y\|  \\
  &\leq \lambda\cdot\|x\|\cdot\|y\|
\end{align*}
where the last line follows from the assumption that $\lambda(G)\leq \lambda$.
\end{proof}

\section{Omitted Proofs from Section~\ref{sec:squaring_solver}}\label{sec:solver_app}
Below is the statement of Lemma 2.3 from \cite{CKKPPRS18}. we give a proof of Lemma~\ref{lem:CumulativeErrorBlock} using this lemma.

\begin{lemma}[CKKPPRS18 Lemma 2.3]
	\label{lem:CumulativeErrorNew}
Consider a sequence of $m$-by-$m$ matrices
$\ms^{(0)}$,$\ms^{(1)}$,$\ldots$, $\ms^{(m)}$ such that
\begin{enumerate}
\item $\ms^{(i)}$ has non-zero entries only on the indices $[i + 1, m]$, 
\item The left/right kernels of $\ms^{(i)}$ are equal, and 
  after restricting $\ms^{(i)}$ to the indices $[i + 1, m]$, the
  kernel of the resulting matrix equals the coordinate restriction
  of the vectors in the kernel of $\ms$. Formally, 
  $\ker\left(\ms^{(i)}_{[i + 1, m], [i + 1, m]}\right) = \{ b_{[i + 1,m]} : b \in \ker(\ms^{(0)})\}$.
\item The symmetrization
of each $\ms^{(i)}$,
$\mU_{\ms^{(i)}} = \frac{1}{2} (\ms^{(i)} + (\ms^{(i)})^\top)$,
is positive semi-definite.
\end{enumerate}
Let $\mm = \mm^{(0)} = \ms^{(0)}$, and
define matrices $\mm^{(1)}, \mm^{(2)}, \ldots, \mm^{(m)}$
iteratively by
\[
\mm^{\left( i + 1 \right)}
\defeq
\mm^{\left( i \right)}
+
\left( \ms^{\left( i + 1 \right)}
- \schur\left(\mm^{\left( i \right)},\left[i + 1, m\right]\right) \right)
\qquad
\forall~0 \leq i < p_{\max}.
\]
If for a subsequence of indices
$1 = i_0 < i_1 < i_2 < \ldots < i_{p_{\max}}$
associated scaling parameters
$0 < \theta_{0}, \theta_{1}, \ldots , \theta_{p_{\max} - 1} < 1/2$ such that $\sum_{p=0}^{p_{\max}-1} \theta_p = 1$,
and some global error $0 < \epsilon < 1/2$,
we have for every $0 \leq p < p_{\max}$:
\[
\left\|\mU_{\ms^{\left(i_p\right)}}^{+/2}
	\left(\mm^{\left(i_p\right)} - \mm^{\left(i_{p + 1}\right)} \right)
	\mU_{\ms^{\left(i_p\right)}}^{+/2}\right\|
\leq
\theta_{p} \epsilon,
\]
then for a matrix-norm defined from the symmetrization of the
$\ms^{(i_p)}$ matrices and the scaling parameters:
\[
\mf
=
\sum_{0 \leq p < p_{\max}}
\theta_{p} \mU_{\ms^{\left(i_{p}\right)}},
\]
we have:
\begin{enumerate}
\item \label{part:CumulativeError}
for each $0 \leq i \leq p_{\max}$,
\[
\left\|\mf^{+/2}
	\left(\mm- \mm^{\left(i\right)}\right)
    \mf^{+/2}\right\|_2
\leq \epsilon,
\]
\item \label{part:FNormBounds}
The final matrix $\mm^{(p_{\max})}$ satisfies
\[
\mm^{\left(p_{\max} \right) \top}
\mf^{+}
\mm^{\left(p_{\max} \right)}
\succeq \frac{1}{10 p_{\max}^2} \cdot \mf.
\]
\end{enumerate}
\end{lemma}

\begin{lemma}\label{lem:l_i_diff}
Let $\mL^{(i)}$'s be $2^k n$ by $2^k n$ matrices defined as Equation~\eqref{eq:L_i_fac}, then 
$$
\mL^{(i+1)}-\mL^{(i)} = 
\begin{bmatrix}
    0& 0 \\
    0 & \mI_{2^{k-i-1} n}-\mc_{2^{k-i-1}}\otimes \mw_{i+1}
\end{bmatrix}
- \begin{bmatrix}
    0 & 0 \\
    0 & \mI_{2^{k-i-1} n} - \mc_{2^{k-i-1}} \otimes \mw_{i}^2
\end{bmatrix}.
$$
\end{lemma}
\begin{proof}
From Equation~\eqref{eq:L_i_fac}, we have
$$
\mL^{(i+1)} = \mx_1 \cdots \mx_{i+1}
\begin{bmatrix}
    \mI_{(2^{k}-2^{k-i-1})n} & 0 \\
    0 & \mI_{2^{k-i-1}n} - \mc_{2^{k-i-1}} \otimes \mw_{i+1}
\end{bmatrix}
\my_{i+1} \cdots \my_1,
$$
applying $\mx_{i+1}$, and $\my_{i+1}$ to the middle matrix we get,
$$
\mL^{(i+1)} = \mx_1 \cdots \mx_{i}
\begin{bmatrix}
    \mI_{(2^{k}-2^{k-i})n} & 0 &0 \\
    0 & \mI_{\frac{2^{k}n}{2^{i+1}}} & - \mc_{\frac{2^{k}n}{2^{i+1}}} \otimes \mw_{i+1} \\
    0 &  -\mI_{\frac{2^{k}n}{2^{i+1}}} \otimes \mw_{i+1} & \mI_{\frac{2^{k}n}{2^{i+1}}} - \mc_{\frac{2^{k}n}{2^{i+1}}} \otimes \mw_{i+1} + \mc_{\frac{2^{k}n}{2^{i+1}}} \otimes \mw_{i}^2
\end{bmatrix}
\my_{i} \cdots \my_1.
$$
Therefore,
$$
\mL^{(i+1)} - \mL^{(i)} = \mx_1 \cdots \mx_{i}
\begin{bmatrix}
    0 & 0 \\
    0 &  - \mc_{2^{k-i-1}} \otimes \mw_{i+1} + \mc_{2^{k-i-1}} \otimes \mw_{i}^2
\end{bmatrix}
\my_{i} \cdots \my_1.
$$
Note that the error only resides on the lower right block of the middle matrix. It is easy to verify that hitting the middle matrix by $\mx^{(i)}$'s from left and $\my^{(i)}$'s from right does not change the matrix (as the non-zero part is only multiplied by identity blocks). Therefore,
$$
\mL^{(i+1)} - \mL^{(i)} = 
\begin{bmatrix}
    0 & 0 \\
    0 &  - \mc_{2^{k-i-1}} \otimes \mw_{i+1} + \mc_{2^{k-i-1}} \otimes \mw_{i}^2
\end{bmatrix}.
$$

\end{proof}

Now we are ready to give a proof of Lemma~\ref{lem:CumulativeErrorBlock}.

\begin{proof}[Proof of Lemma~\ref{lem:CumulativeErrorBlock}]
From $\ms^{(i)}$'s and $\mL^{(i)}$'s, we build a sequence of $\hat{\ms}^{(j)}$'s and $\hat{\mm}^{(j)}$'s that satisfy the conditions of Lemma~\ref{lem:CumulativeErrorNew}, and using that we derive the statement of Lemma~\ref{lem:CumulativeErrorBlock}. For $0 \leq i < k$, and $0 \leq j < 2^{k-i-1}n$, let $a_i\eqdef (2^k-2^{k-i})n$, $\hat{\ms}^{(a_k)} = \ms^{(k)}$, and
$$
\hat{\ms}^{(a_i+j)} = \begin{cases}
\ms^{(i)}, & \text{if }j = 0 \\
\schur{(\hat{\mm}^{a_i+j-1}, [a_i+j,2^k n])} & \text{otherwise.}
\end{cases}
$$
and
\[
\hat{\mm}^{\left( h + 1 \right)}
\defeq
\hat{\mm}^{\left( h \right)}
+
\left( \hat{\ms}^{\left( h + 1 \right)}
- \schur\left(\hat{\mm}^{\left( h \right)},\left[h + 1, 2^k n\right]\right) \right)
\qquad
\forall~0 \leq h < (2^k -1) n.
\]

Note that $\hat{\ms}$'s satisfy all the three premises in Lemma~\ref{lem:CumulativeErrorNew}. First $\hat{\ms}^{(i)}$ has non-zero entries only on the indices $[i+1, 2^k n]$. Further as all $\hat{\ms}^{(i)}$'s are random-walk Laplacian of Eulerian connected aperiodic regular digraphs, the left and right kernels are the same, all have the same kernel up to the restriction to non-zero entries, and all $\mU_{\ms^{(i)}}$'s are PSD. Therefore all the three premises of the lemma are satisfied.

Next we show that for all $i$'s $\mL^{(i+1)}$ approximates $\mL^{(i)}$ in the norm defined by $\mU_{\ms^{(i)}}$. By Lemma~\ref{lem:l_i_diff},
$$
\mL^{(i+1)}-\mL^{(i)} = \begin{bmatrix}
    0 & 0\\
    0 & \mc_{2^{k-i-1}} \otimes \mw_i^2 - \mc_{2^{k-i-1}} \otimes \mw_{i+1}.
\end{bmatrix}
$$
Now, given $\mw_{i+1} \capprox_{\epsilon/k} \mw_i^2$, by Theorem~\ref{thm:capprox_ring} we get
$$
\|\mU_{\schur(\ms_i,H_i)}^{+/2} \left(\mL^{(i+1)}-\mL^{(i)}\right)\mU_{\schur(\ms_i,H_i)}^{+/2}\| \leq \epsilon/k.
$$
Since $\mU_{\schur(\ms_i,H_i)} \preceq 2 \mU_{\ms^{(i)}}$, 
$$
\|\mU_{\ms^{(i)}}^{+/2} \left(\mL^{(i+1)}-\mL^{(i)}\right)\mU_{\ms^{(i)}}^{+/2}\| \leq 2\epsilon/k.
$$
By construction, we have $\hat{\ms}^{(a_i)} = \ms^{(i)}$ and $\hat{\mm}^{(a_i)} = \mL^{(i)}$ for all $0\leq i \leq k$. Therefore, we get
\[
\left\|\mU_{\hat{\ms}^{\left(a_i\right)}}^{+/2}
	\left(\hat{\mm}^{\left(a_i\right)} - \hat{\mm}^{\left(a_{i + 1}\right)} \right)
	\mU_{\hat{\ms}^{\left(a_i\right)}}^{+/2}\right\|
\leq
2\epsilon/k,
\]
Thus by Lemma~\ref{lem:CumulativeErrorNew}, for $\mf = \frac{2}{k}\sum_{i=0}^k \mU_{\ms^{(i)}}$,
$$
\|\mf^{+/2} (\mL - \mL^{(i)}) \mf^{+/2}\| \leq \epsilon \quad \forall 0 \leq i \leq k
$$
and
$$
{\mL^{(k)}}^\top \mf^{+} {\mL^{(k)}} \succeq \frac{1}{40k^2} \mf.
$$
\end{proof}

\begin{lemma}\label{lem:psd_ortho_proj}
Given a PSD matrix $\mU \in \R^{n\times n}$, and an orthogonal projection $\Pi \in \R^{n\times n}$ such that $\mU \Pi = \Pi \mU$ we have
$$
\mU \Pi \preceq \mU
$$

\end{lemma}
\begin{proof}
Let $x \in \R^n$ be an arbitrary vector, and let $x_1$ denote the projection of $x$ into $\im(\Pi)$ and $x_{2}$ be the projection to the subspace orthogonal to $\im(\Pi)$, such that $x = x_1 + x_2$. Note that $\Pi x_1 = x_1$, and $\Pi x_2 = 0$. Therefore we have,
$$
(x_1+x_2)^\top \mU (x_1+x_2) = (\Pi x_1+x_2)^\top \mU (\Pi x_1+x_2) = x_1^\top \mU x_1 + x_2^\top \mU x_2,
$$
and 
$$
(x_1+x_2)^\top \mU \Pi (x_1+x_2) = (x_1+x_2)^\top \mU \Pi x_1
= (x_1+x_2)^\top \Pi \mU  x_1
= x_1^\top \Pi \mU \Pi x_1
= x_1^\top \mU x_1.
$$
Therefore,
$$
x^\top \mU \Pi x = x_1^\top \mU x_1 \leq x_1^\top \mU x_1 + x_2^\top \mU x_2 = x^\top \mU x.
$$
Since the choice of $x$ was arbitrary this completes the proof.
\end{proof}

\begin{lemma}[Lemma 2.6 in \cite{CKKPPRS18}]\label{lem:lem2.6CKK}
Suppose we are given matrices $\mL$, $\widetilde{\mL}$ and a positive semi-definite matrix $\mf$ such that
$\ker(\mf) \subseteq \ker(\mL) = \ker(\mL^\top) = \ker(\widetilde{\mL}) = \ker(\widetilde{\mL}^\top)$ and
\begin{enumerate}
    \item $\|\mf^{+/2} (\mL - \widetilde{\mL}) \mf^{+/2}\| \leq \epsilon$
    \item $\widetilde{\mL}^\top \mf^{+} \widetilde{\mL} \succeq \gamma \mf,$
\end{enumerate}
then $\|\mI_{\im(\mL)} - \widetilde{\mL}^+ \mL \|_{\mf} \leq \epsilon \sqrt{\gamma^{-1}}$.
\end{lemma}

\begin{lemma}[Lemma B.3 in \cite{CKKPPRS18}]\label{lem:lemB.3CKK}
Let $\mL,\widetilde{\mL}, \mf$ be arbitrary matrices with $\ker(\mf) =\ker(\mf^\top) = \ker(\mL) = \ker(\mL^\top) = \ker(\widetilde{\mL}) = \ker(\widetilde{\mL}^\top)$. If $\|\mf^{+/2} (\mL - \widetilde{\mL}) \mf^{+/2}\| \leq \epsilon$ and $\widetilde{\mL}^\top \mf^{+} \widetilde{\mL} \succeq \gamma \mf,$ then $\mL^\top F^+ \mL \approx_{O(\frac{\epsilon}{\sqrt{\gamma}}+\frac{\epsilon^2}{\gamma})} \widetilde{\mL}^\top F^+ \widetilde{\mL}$.
\end{lemma}

\begin{lemma}[Lemma 13 in \cite{CKPPSV16}]\label{lem:lem13CKP} Let $\mL$ be an Eulerian directed Laplacian, $\tr(\mU^{+/2}\mL^\top \mU^+ \mL \mU^{+/2}) \leq 2(n-1)^2$.
\end{lemma}

\section{Proof of Lemma \ref{lem:derandspacecomplexity}}
\label{app:derandspacecomplexity}
\begin{namedtheorem}[Lemma \ref{lem:derandspacecomplexity} restated]
Let $G_{0}$ be a $d$-regular, directed multigraph on $n$ vertices with a two-way labeling and $H_1,\ldots,H_k$ be $c$-regular undirected graphs with two-way labelings where for each $i\in[k]$, $H_{i}$ has $d\cdot c^{i-1}$ vertices. For each $i\in[k]$ let 
\[
G_i=G_{i-1}\ds H_i.
\]
Then given $v_0\in[n], i_{0}\in[d\cdot c^{i-1}], j_0\in[c]$, \textup{Rot}$_{G_i}(v,(i_0,j_0))$ can be computed in space $O(\log(n\cdot d)+k\cdot\log c)$ with oracle queries to \textup{Rot}$_{H_1},\ldots,\mathrm{Rot}_{H_k}$.
\end{namedtheorem}
\begin{proof}
When $c$ is subpolynomial we are reasoning about sublogarithmic space complexity, which can depend on the model. So we will be explicit about the model we are using. We compute the rotation map of $G_i$ on a multi-tape Turing machine with the following input/output conventions:

\begin{itemize}
\item Input Description:
\begin{itemize}
\item Tape 1 (read-only): Contains the initial input graph $G_0$, with the head at the leftmost position of the tape.
\item Tape 2 (read-write): Contains the input triple $(v_0,(i_0,j_0))$, where $v_0$ is a vertex of $G_i$, $i_0\in [d\cdot c^{i-1}]$ is an edge label in $G_i$, and $j_0\in[c]$ is an edge label in $H_i$ on a {\em read-write} tape, with the head at the {\em rightmost} position of $j_{0}$. The rest of the tape may contain additional data.
\item Tapes 3+ (read-write): Blank worktapes with the head at the leftmost position.
\end{itemize}

\item Output Description:
\begin{itemize}
\item Tape 1: The head  should be returned to the leftmost position.
\item Tape 2: In place of $(v_0,(i_0,j_0))$, it should contain the output $(v_2,(i_3,j_1))=\mathrm{Rot}_{G_i}(v_0,(i_0,j_0))$ as described in Definition \ref{def:derandsquare}. The head should be at the rightmost position of $j_1$ and the rest of the tape should remain unchanged from its state at the beginning of the computation.
\item Tapes 3+ (read-write): Are returned to the blank state with the heads at the leftmost position.
\end{itemize}
\end{itemize}

Let $\Space(G_i)$ be the amount of space required to compute the rotation map of graph $G_i$. We will show that for all $i\in[k]$, $\Space(G_i)=\Space(G_{i-1})+O(\log c)$. Note that $\Space(G_0)=O(\log (nd))$. 

Fix $i\in [k]$. We begin with $v_0\in[n], i_{0}\in[d\cdot c^{i-1}]$ and $j_0\in[c]$ on tape 2 with the head on the rightmost position of $j_0$ and we want to compute Rot$_{G_i}(v_0,(i_0,j_0))$. We move the head left to the rightmost position of $i_0$, recursively compute Rot$_{G_{i-1}}(v_0,i_0)$ so that tape 2 now contains $(v_1,i_1,j_0)$.  Then we move the head to the rightmost position of $j_0$ and compute Rot$_{H_i}(i_1,j_0)$ so that tape 2 now contains $(v_1,i_2,j_1)$. Finally, we move the head to the rightmost position of $i_2$ and compute Rot$_{G_{i-1}}(v_1,i_2)$ so that tape 2 contains $(v_2,i_3,j_1)$. 

This requires 2 evaluations of the rotation map of $G_{i-1}$ and one evaluation of the rotation map of $H_i$. Note that we can reuse the same space for each of these evaluations because they happen in succession. The space needed on top of $\Space(G_{i-1})$ is the space to store edge label $j_0$, which uses $O(\log c)$ space. So $\Space(G_i)=\Space(G_{i-1})+O(\log c)$. Since $i$ can be as large as $k$ and $\Space(G_0)=O(\log (n\cdot d))$ we get that for all $i\in[k]$, $\Space(G_i)=O(\log n\cdot d + k\cdot\log c)$.
\end{proof}

\end{document}